\documentclass[12pt,twoside,twocolumn]{article}
\pdfoutput=1

\usepackage[square ,sort&compress,comma,numbers]{natbib} 

\usepackage[version=3]{mhchem}
\usepackage{times,mathptmx}
\usepackage{sectsty}
\usepackage{balance} 

\usepackage{amsmath} 
\usepackage{bm}
\usepackage{lastpage}
\usepackage[format=plain,justification=raggedright,singlelinecheck=false,font=small,labelfont=bf,labelsep=space]{caption}
 
\usepackage{fancyhdr}
\usepackage[caption=false,position=top]{subfig}
\usepackage{graphicx}

\usepackage{etoolbox}

\usepackage{url} 
\usepackage{amsthm,amssymb}
\usepackage{thmtools,thm-restate}

\newtheorem{thm}{Theorem}
\newtheorem{prop}{Proposition}
\newtheorem{lem}{Lemma}
\newtheorem{cor}{Corollary}
\newtheorem{defn}{Definition}

\usepackage{array}
\usepackage{bbold}
\usepackage[dvipsnames,table]{xcolor}
\usepackage{booktabs}
\usepackage{mathtools}

\usepackage[perpage]{footmisc}

\renewcommand{\thefootnote}{\fnsymbol{footnote}}

\usepackage{verbatim}
\usepackage{rotating}

\renewcommand{\vec}[1]{\underline{#1}}

\newcommand{\cC}{\mathcal{C}}

\newcommand{\cE}{\mathcal{E}}

\newcommand{\cF}{\mathcal{F}}

\newcommand{\cH}{\mathcal{H}}

\newcommand{\cI}{\mathcal{I}}

\newcommand{\cM}{\mathcal{M}}
\newcommand{\cN}{\mathcal{N}}

\newcommand{\cO}{\mathcal{O}}

\newcommand{\cR}{\mathcal{R}}

\newcommand{\bC}{\mathbb{C}}
\newcommand{\bD}{\mathbb{D}}
\newcommand{\bE}{\mathbb{E}}

\newcommand{\bR}{\mathbb{R}}

\newcommand{\bZ}{\mathbb{Z}}

\newcommand{\bra}[1]{\left\langle #1\right\vert}
\newcommand{\ket}[1]{\left\vert #1\right\rangle}
\newcommand{\braket}[2]{\left\langle \left. #1 \right\vert #2\right\rangle}

\newcolumntype{L}{>{$}l<{$}} 

\newcolumntype{R}{>{$}r<{$}} 

\newcolumntype{C}{>{$}c<{$}} 

\usepackage{hyperref}

\begin{document}


\thispagestyle{plain}
\fancypagestyle{plain}{
\renewcommand{\headrulewidth}{0pt}}
\renewcommand{\thefootnote}{\fnsymbol{footnote}}
\renewcommand\footnoterule{\vspace*{1pt}} 
\setcounter{secnumdepth}{5}

\makeatletter 
\def\subsubsection{\@startsection{subsubsection}{3}{10pt}{-1.25ex plus -1ex minus -.1ex}{0ex plus 0ex}{\normalsize\bf}} 
\def\paragraph{\@startsection{paragraph}{4}{10pt}{-1.25ex plus -1ex minus -.1ex}{0ex plus 0ex}{\normalsize\textit}} 
\renewcommand\@biblabel[1]{#1}            
\renewcommand\@makefntext[1]%
{\noindent\makebox[0pt][r]{\@thefnmark\,}#1}
\makeatother 
\renewcommand{\figurename}{\small{Fig.}~}
\sectionfont{\large}
\subsectionfont{\normalsize} 

\fancyfoot{}
\fancyfoot[CO]{\footnotesize{\sffamily{\thepage}}}
\fancyfoot[CE]{\footnotesize{\sffamily{\thepage}}}
\fancyhead{}
\renewcommand{\headrulewidth}{0pt} 
\renewcommand{\footrulewidth}{0pt}
\setlength{\arrayrulewidth}{1pt}
\setlength{\columnsep}{6.5mm}
\setlength\bibsep{1pt}


\twocolumn[
  \begin{@twocolumnfalse}
\noindent\LARGE{\textbf{Combinatorics of chemical reaction systems}}
\vspace{0.6cm}

\noindent\large{\textbf{Nicolas Behr\textit{$^{a}$}, G\'{e}rard H. E.~Duchamp\textit{$^{b}$} and Karol A.~Penson\textit{$^{c}$}}}\vspace{0.5cm}

\noindent \normalsize{%
We propose a concise stochastic mechanics framework for chemical reaction systems that allows to formulate evolution equations for three general types of data: the probability generating functions, the exponential moment generating functions and the factorial moment generating functions. This formulation constitutes an intimate synergy between techniques of statistical physics and of combinatorics. We demonstrate how to analytically solve the evolution equations for all six elementary types of single-species chemical reactions by either combinatorial normal-ordering techniques, or, for the binary reactions, by means of Sobolev-Jacobi orthogonal polynomials. The former set of results in particular highlights the relationship between infinitesimal generators of stochastic evolution and parametric transformations of probability distributions.}
\vspace{0.5cm}
 \end{@twocolumnfalse}
  ]

\footnotetext{\textit{$^{a}$~Institut de Recherche en Informatique Fondamentale (IRIF), Universit\'{e} Paris-Diderot (Paris 07), France; E-Mail: nicolas.behr@irif.fr}}

\footnotetext{\textit{$^{b}$~Universit\'{e} Paris 13, Sorbonne Paris Cit\'{e}, LIPN, CNRS UMR 7030, F-93430 Villetaneuse, France; E-Mail: ghed@lipn.univ-paris13.fr}}

\footnotetext{\textit{$^{c}$~Laboratoire de Physique Theorique de la Mati\`{e}re Condens\'{e}e (LPTMC),  CNRS UMR 7600, Sorbonne Universit\'{e}s, Universit\'{e} Pierre et Marie Curie (Paris 06), France; E-Mail: penson@lptl.jussieu.fr}}

\section{Introduction}

Intended as an invitation to interdisciplinary researchers and in particular to combinatorists, we present in this work an extension of the early work of Delbr\"uck~\cite{delbruck1940statistical} on probability generating functions for chemical reaction systems to a so-called stochastic mechanics framework. While the idea to study chemical reaction systems in terms of probability generating functions is thus not new, and on the contrary constitutes one of the standard techniques of this field (see e.g.\ \cite{mcquarrie1967stochastic} for a historical compendium), we believe that the reformulation of these techniques in terms of the stochastic mechanics formalism could lead to fruitful interaction of a broader audience. In the spirit of the ideas presented by M.~Doi in his seminal paper~\cite{doi1976second}, the main motivation for such a reformulation lies in a clear conceptual separation of (i) the state space of the system and (ii) the linear operators implementing the evolution of the system. Combined with insights obtained in a recent study of stochastic graph rewriting systems~\cite{bdg2016,bdgh2016,bCCDD2017}, one may add to this list (iii) the linear operators that implement observable quantities. It is only through combining this Ansatz with the standard notions of combinatorial generating functions that we find the true strengths of the stochastic mechanics approach: providing an avenue to obtain exact solutions to dynamical evolution equations. Combinatorists will recognize in our formulation of evolution equations intrinsic notions of normal-ordering problems, and indeed certain semi-linear normal-ordering techniques~\cite{blasiak2005boson,blasiak2005combinatorics,blasiak2007combinatorics,blasiak2011combinatorial} will prove immensely fruitful in this direction. Chemists and other practitioners might appreciate that our solutions not only provide asymptotic information on the time-evolution of the reaction systems, but also provide full information on the evolution of reaction systems from any initial state at time $t=0$ to any desired time $t=T$ (with $T>0$). While many individual results on such time-evolutions are known in the literature~\cite{mcquarrie1967stochastic,jahnke2007solving}, we hope that our concise formalism may help to consolidate the knowledge on the mathematical methods involved in deriving such exact results, and in some cases even extend the current knowledge. Amongst the novel technical results in this paper we present an exact solution to single-species binary chemical reaction systems. This result is obtained via a necessary modification of an Ansatz by McQuarrie~\cite{mcquarrie1964kinetics} via so-called Sobolev-Jacobi orthogonal polynomials, which have been introduced in the work of Kwon and Littlejohn~\cite{kwon1994characterizations,kwon1996new,kwon1997classification,kwon1998sobolev}.\\

For clarity of exposition, we will first present the general theory of so-called continuous-time Markov chains (CTMCs)~\cite{norris} in our framework of stochastic mechanics. This approach to the theory of CTMCs has the advantage of giving clear intuitions as familiar from the general theory of statistical physics to the types of computations that arise. It moreover admits an entry point to non-specialists that we think is somewhat more accessible than the vast standard literature on the topic. We will then present a derivation of a framework to describe chemical reaction systems, and illustrate how to obtain exact closed-form solutions to the evolution equations of certain types of reaction systems. 

\section{The stochastic mechanics formalism}\label{sec:SMform}

As a preparation for the main part of this paper, we begin with a short survey of some of the standard notions of \emph{continuous time Markov chain (CTMC)} theory~\cite{norris} and related notions of probability theory. Due to a certain (albeit superficial) similarity of many of the formulae of interest to those of quantum mechanics, following~\cite{baez2012course} we will refer to this framework as  \emph{stochastic mechanics}. In particular, some of the notations we will use are inspired by those common in quantum mechanics.\\

We will be concerned with the study of systems whose space of \emph{pure states} (i.e.\ of concrete configurations) is described in the following manner:

\begin{defn}\label{defn:CTMCstates}
For a stochastic dynamical system, let $S$ denote the vector space (over $\bR$) spanned by the \emph{pure states} (basis vectors) $\ket{s}\in S$. Typically, $S$ is countably infinite dimensional. A \emph{mixed state} $\ket{\Psi}\in Prob(S)$ of the system is a \emph{probability distribution over $S$}, with $Prob(S)$ defined as
\begin{equation}
\begin{aligned}
	Prob(S)&:=\bigg\{
		\ket{\Psi}=\sum_s \psi_s\ket{s}\,,\\
		&\qquad \qquad\psi_s\in \bR_{\geq0}\; \text{and}\; \sum_s \psi_s=1
	\bigg\}\,.
\end{aligned}
\end{equation}
\end{defn}
Elements of $Prob(S)$ may thus in particular be seen as special elements of the space $\ell^1_{\bR}(S)$ (absolutely $\ell^1_{\bR}$-summable series over $S$), 
\begin{equation}
\begin{aligned}
	\ell^1_{\bR}(S)&:=\bigg\{
	\ket{\Psi}=\sum_s \psi_s\ket{s}\,,\\
	&\qquad \quad \psi_s\in \bR\;\text{and}\;
	\sum_s |\psi_s|<\infty
	\bigg\}\,.
\end{aligned}
\end{equation}
For later convenience, we introduce an operation $\bra{\cR}$ via its action on basis vectors,
\begin{equation}
	\braket{\cR}{s}:=1 \quad \text{for all } \ket{s}\in S\,,
\end{equation}
which extends by linearity to
\begin{equation}
\braket{\cR}{\Psi}=\sum_{s}\psi_s \braket{\cR}{s}=\sum_{s}\psi_s\,.
\end{equation}
In other words, $\bra{\cR}$ implements the \emph{summation of coefficients} 
when applied to a distribution $\ket{\Psi}$ over $S$.\\

The reader is cautioned \emph{not} to take the analogy to the notations of quantum mechanics too far in interpreting the notation $\braket{\cR}{s}$ as a form of scalar product. In fact, since this is a rather crucial point, we provide in Appendix~\ref{app:QMvsSM} a direct comparison of the notions of stochastic mechanics as opposed to the notions of quantum mechanics. For the purposes of stochastic mechanics, we thus merely treat $\bra{\cR}$ as a convenient (and standard\footnote{Some alternative notations frequently encountered in the literature~\cite{doi1976second,doi1976stochastic,peliti1985path} include \[
\bra{\cR}\equiv\bra{}\equiv\bra{P}\equiv\bra{-}\,.
\]}) shorthand notation.\\

We now turn to one of the core concepts of the stochastic mechanics framework:
\begin{defn}[cf.\ e.g.~\cite{norris}]
	Let $S$ denote a vector space of pure states, and let $Prob(S)$ denote the space of probability distributions over $S$ (cf.\ Definition~\ref{defn:CTMCstates}). Then a \emph{continuous time Markov chain (CTMC)} over $S$ is defined in terms of providing an \emph{input state} 
\begin{equation}
\ket{\Psi(0)}\in Prob(S)
\end{equation}
and a \emph{one-parameter evolution semi-group} $\cE(t)$, where $t\in [0,\infty)$. $\cE(t)$ possesses the following defining properties (for all non-negative real parameters $s,t\in \bR_{\geq0}$):
\begin{subequations}
		\begin{align}
			\cE(0)&=\mathbb{1}\\
			\cE(s+t)&=\cE(s)\cE(t)=\cE(t)\cE(s)\\
			\frac{d}{dt}\cE(t)&=H\cE(t)=\cE(t)H\label{eq:CTMCdiff}\,.
		\end{align}
\end{subequations}
Here, the first two equations express the semi-group property of $\cE$, while~\eqref{eq:CTMCdiff} expresses the \emph{Markov property} (i.e.\ \emph{memorylessness}) itself, with $H$ the  \emph{infinitesimal generator} of $\cE$. Then the time-dependent state $\ket{\Psi(t)}\in Prob(S)$ of the $CTMC$ is given by
	\begin{equation}
		\ket{\Psi(t)}=\cE(t)\ket{\Psi(0)}\,.
	\end{equation}
	Owing to~\eqref{eq:CTMCdiff}, $\ket{\Psi(t)}$ satisfies the so-called \emph{Master equation} or \emph{Kolmogorov backwards equation},
	\begin{equation}
		\label{eq:Master}
		\frac{d}{dt}\ket{\Psi(t)}=H\ket{\Psi(t)}\,.
	\end{equation}
\end{defn}
Since the infinitesimal generator $H$ thus bears a certain formal resemblance in its role to the Hamiltonian in the quantum mechanics setting, we will (by an abuse of terminology as illustrated in Appendix~\ref{app:QMvsSM}) refer to it as Hamiltonian also in the stochastic mechanics setting.\\

The definition of a  continuous-time Markov chain thus comprises the information on what are the \emph{pure states} of the stochastic dynamical system, and on what are the possible \emph{transitions} of the system. The transitions in turn would typically implement the notions of, say, chemical reactions or other changes of configurations of states. In practice, the precise nature of these transitions will determine the structure of the infinitesimal generator $H$. We will provide an in-detail discussion of the infinitesimal generators in the case of chemical reaction systems in the main part of this paper.\\

It might be worth pointing out that the systems thus defined are in a certain sense \emph{dimensionless}, i.e.\ there is no notion of physical space implied by the definitions. However, as for example realized in the famous Doi-Peliti formalism~\cite{doi1976second,peliti1985path}, it is possible to extend the concept of pure states to include a notion of space e.g.\ in the form of configurations over a lattice of positions in physical space, and even to consider taking a certain limit of zero lattice spacing to recover a notion of continuous space.\\

As we will see in the sequel, the Hamiltonian $H$ of a CTMC can in general turn out to be an \emph{unbounded} operator acting on a countably infinite-dimensional space, whence considerable care is necessary in order to understand the precise mathematical meaning of the CTMC definitions and of the evolution semigroup. While in the case of a finite state-space we simply find the relation (see e.g.\ \cite{gauckler2014regularity} and references therein)
\begin{equation}
	\cE(t)=e^{tH}\,,
\end{equation}
for unbounded operators $H$ such an expression would quite possibly not be mathematically well-posed. In the general cases, one therefore needs to resort to more intricate concepts such as the \emph{Hille-Yosida theory} of functional analysis. We refer the interested readers to the standard literature on the subject (see e.g.~\cite{reuter1957denumerable,dynkin1965markov,engel1999one}) for the full mathematical details. Here, we content ourselves with stating a number of abstract properties of the infinitesimal generators $H$ that will prove useful to our framework in the following.

\begin{lem}
Given a CTMC over some state space $S$ with initial state $\ket{\Psi(0)}\in Prob(S)$, evolution semi-group $\cE$ and infinitesimal generator $H$, we have that
\begin{equation}
	\label{eq:keyLemH}
	\bra{\cR}H=0\,.
\end{equation}	
\begin{proof}
	Recall that since for all $t\geq0$ the state $\ket{\Psi(t)}$ of the CTMC is a probability distribution, it satisfies
	\[
		\braket{\cR}{\Psi(t)}=\sum_s \psi_s=1\,.
	\]
	This entails in particular that
	\begin{align*}
		0=\frac{d}{dt}\braket{\cR}{\Psi(t)}&=\frac{d}{dt}\bra{\cR}\cE(t)\ket{\Psi(0)}\\
		&\overset{\eqref{eq:CTMCdiff}}{=} 
		\bra{\cR}H\cE(t)\ket{\Psi(0)}\,.
	\end{align*}%
	But since this must hold for arbitrary times $t\geq0$, the claim~\eqref{eq:keyLemH} follows.
\end{proof}
\end{lem}

\subsection{Advantages of the stochastic mechanics formulation}

We will now present some of the key strengths of the formulation of CTMCs in terms of the stochastic mechanics formalism, a point which we will elaborate further when entering the discussion of chemical reaction systems.\\
 
One of the main advantages of the formalism is the ability to compute naturally (and in many cases even in closed form) quantities of interest in practical applications, namely the \emph{moments} of observables. Since part of the computations for chemical reactions will involve additional mathematical structures, we choose to first present some of the general concepts here.\\

Suppose then that we are studying a continuous-time Markov chain (CTMC) with space of pure states $S$, initial state $\ket{\Psi(0)}\in Prob(S)$, evolution semi-group $\cE(t)$ and infinitesimal generator $H$ as before. 
\begin{defn}\label{def:obs}
	An \emph{observable} $O$ for a CTMC is a \emph{diagonal linear operator on $S$}, whence for all pure states $\ket{s}\in S$ 
	\begin{equation}\label{eq:ObsDef}
		O\ket{s}=\omega_s(O)\ket{s}\quad (\omega_s(O)\in \bR)\,.
	\end{equation}
	We denote the \emph{space of observables on $S$} by $\cO(S)$. 
\end{defn}
Since the observables are thus by definition \emph{diagonal} in the basis of pure states, any two observables $O_1,O_2\in \cO(S)$ \emph{commute}:
\begin{equation}
	[O_1,O_2]=O_1O_2-O_2O_1=0\,.
\end{equation}

Consider now a fixed (finite) set of observables $\{O_1,\dotsc,O_n\}$, with $O_i\in \cO(S)$. Let us introduce the convenient \emph{multi-index notation}
\[
	\vec{x}:=(x_1,\dotsc,x_n)
\]
for vectors and
\begin{equation}
	\vec{\lambda}\cdot \vec{O}:=\sum_{i=1}^n\lambda_i O_i
\end{equation}
 for linear combinations of observables. Here, we will either consider real  coefficients $\lambda_i\in \bR$ or \emph{formal} coefficients (such as e.g.\ in Definition~\ref{def:EGF}). We will also make use of the shorthand notation
\begin{equation}
	\vec{O}^{\vec{m}}:=\prod_{i=1}^n O_i^{m_i}\,,
\end{equation}
with (finite) non-negative integer exponents $m_i$.\\

Denoting by $\ket{\Psi(t)}=\cE(t)\ket{\Psi(0)}$ the \emph{time-dependent state of the system}, we may now define in complete analogy to standard probability theory the concept of \emph{expectation values} of observables, commonly referred to as \emph{moments}:
\begin{defn}
	For a given choice of exponents $m_1,\dotsc,m_n$, we define
	\begin{equation}
		\left\langle \vec{O}^{\vec{m}}\right\rangle(t):=
		\bra{\cR}O_1^{m_1}\cdots \dotsc\cdots O_n^{m_n}\ket{\Psi(t)}\,.
	\end{equation}
\end{defn}

For example, the expectation value of a single observable $O_i$ in state $\ket{\Psi(t)}$, in the traditional probability theory literature sometimes denoted $\bE_{\ket{\Psi(t)}}(O_i)$, would thus read
\begin{align*}
\bE_{\ket{\Psi(t)}}(O_i)&\equiv \langle O_i\rangle(t)=\bra{\cR}\cO_i\ket{\Psi(t)}\\
&=\sum_s \psi_s(t) \bra{\cR}O_i\ket{s}\\
&\overset{\eqref{eq:ObsDef}}{=}\sum_s \psi_s(t) \omega_s(O_i)\braket{\cR}{s}\\
&=\sum_s\psi_s(t)\omega_s(O_i)\,.
\end{align*}

\begin{defn}\label{def:EGF}
	With notations as above, the \emph{exponential generating function (EGF)} $\cM(t;\vec{\lambda})$ for the moments of the set of observables $\{O_1,\dotsc,O_m\}$ of the system with time-dependent state $\ket{\Psi(t)}$ is defined as
	\begin{equation}
	\begin{aligned}
		\cM(t;\vec{\lambda})&:=
		\left\langle e^{\vec{\lambda}\cdot \vec{O}}\right\rangle(t)\\
		&=\sum_{m\geq 0}\frac{1}{m!}
		\bra{\cR}(\vec{\lambda}\cdot\vec{O})^m\ket{\Psi(t)}\,.
	\end{aligned}
	\end{equation}
	Arbitrary moments may be computed from $\cM(t;\vec{\lambda})$ according to
	\begin{equation}
	\begin{aligned}
		&\left\langle O_1^{m_1}\cdots\dotsc\cdots O_n^{m_n}\right\rangle(t)\\
		&\qquad =\left[\frac{\partial^{m_1}}{\partial \lambda_1^{m_1}}
		\dotsc \frac{\partial^{m_n}}{\partial \lambda_1^{m_n}}
		\cM(t;\vec{\lambda})\right]\bigg\vert_{\vec{\lambda}\to\vec{0}}\,.
	\end{aligned}
	\end{equation}
\end{defn}
In other words, knowing the exponential moment-generating function $\cM(t;\vec{\lambda})$ for the observables of interest in a given system amounts typically to knowing all relevant stochastic dynamical properties of the system.\\

Finally, we are now in a position to assemble all concepts introduced thus far into one of the cornerstones of this paper (see~\cite{bdg2016} for the original version of this result, formulated in the framework of rule algebras invented by the first named author, and~\cite{bCCDD2017} for its most general variant):
\begin{restatable}{thm}{thmGenEGFev}
\label{thm:GenEGFev}
	For a CTMC with space of pure states $S$, initial state $\ket{\Psi(0)}\in Prob(S)$, evolution semi-group $\cE(t)$ and infinitesimal generator $H$, the exponential moment-generating function $\cM(t;\vec{\lambda})$ for a given (finite) set of observables $\{O_1,\dotsc,O_n\}$ (with $O_i\in \cO(S)$) fulfills the \emph{evolution equation}
	\begin{equation}
		\tfrac{\partial}{\partial t}\cM(t;\vec{\lambda})=\sum_{q\geq 1}\frac{1}{q!} \,\bra{\cR}\left(ad_{(\vec{\lambda}\cdot\vec{O})}^{\circ\:q}\:H\right)e^{\vec{\lambda}\cdot\vec{O}}\ket{\Psi(t)}\,.
	\end{equation}
	Here, for two linear operators $A$ and $B$ the notation $ad_A B$ stands for the \emph{adjoint action} (\emph{commutator}),
	\[
		ad_A B:=[A,B]=AB-BA\,,
	\]
	while $ad_A^{\circ\:q} \:B$ denotes the $q$-fold commutator ($q>1$),
	\begin{align*}
		ad_A^{\circ\:q}\:B&:=\underbrace{(ad_A\circ\dotsc \circ ad_A)}_{\text{$q$ times}}(B)\\
		&=\underbrace{[A,[A,\dotsc,[A,B]\dotsc]]}_{\text{$q$ times}}\,.
	\end{align*}%
\end{restatable}
\begin{proof}
	See Appendix~\ref{app:proofGenMGFev}.
\end{proof}

\section{Chemical reaction systems of one species}\label{sec:CRS}

According to the standard literature on chemical reaction systems~\cite{bartholomay1962enzymatic,mcquarrie1967stochastic,delbruck1940statistical,ishida1960stochastic,mcquarrie1963kinetics,mcquarrie1964kinetics,singer1953application} (see in particular~\cite{mcquarrie1967stochastic} for a historical review of the early work on the subject), a chemical reaction system of one species of particles is a stochastic transition system whose basic (``one-step'') transitions are of the form
\begin{equation}\label{eq:CR1}
	i\: A \xrightharpoonup{\;r_{i,o}\;}\: o A\,.
\end{equation}
The notation encodes a transition in which $i$ copies of a particle of species $A$ are transformed into $o$ copies of the same species, while\footnote{The SI unit of base rates is in fact $[s^{-1}]$, but we will work throughout with dimensionless units due to the fact that rates are always multiplied by the system's time parameter $t$ in our applications.} $r_{i,o}\in \bR_{\geq 0}$ denotes the transition's \emph{base rate}. As sketched in Figure~\ref{fig:HWillustr}, one should think of the physical system to consist of a finite number of indistinguishable particles of type $A$, with the transitions implementing operations of removing and adding particles to the system. In this standard description, the positions and impulses of the individual particles are completely neglected, and all physical parameters such as temperature, pressure, volume, chemical potentials etc.\ are collected into the values of the base rates.\\

\begin{figure}[h]
\centering
  \includegraphics[height=4cm]{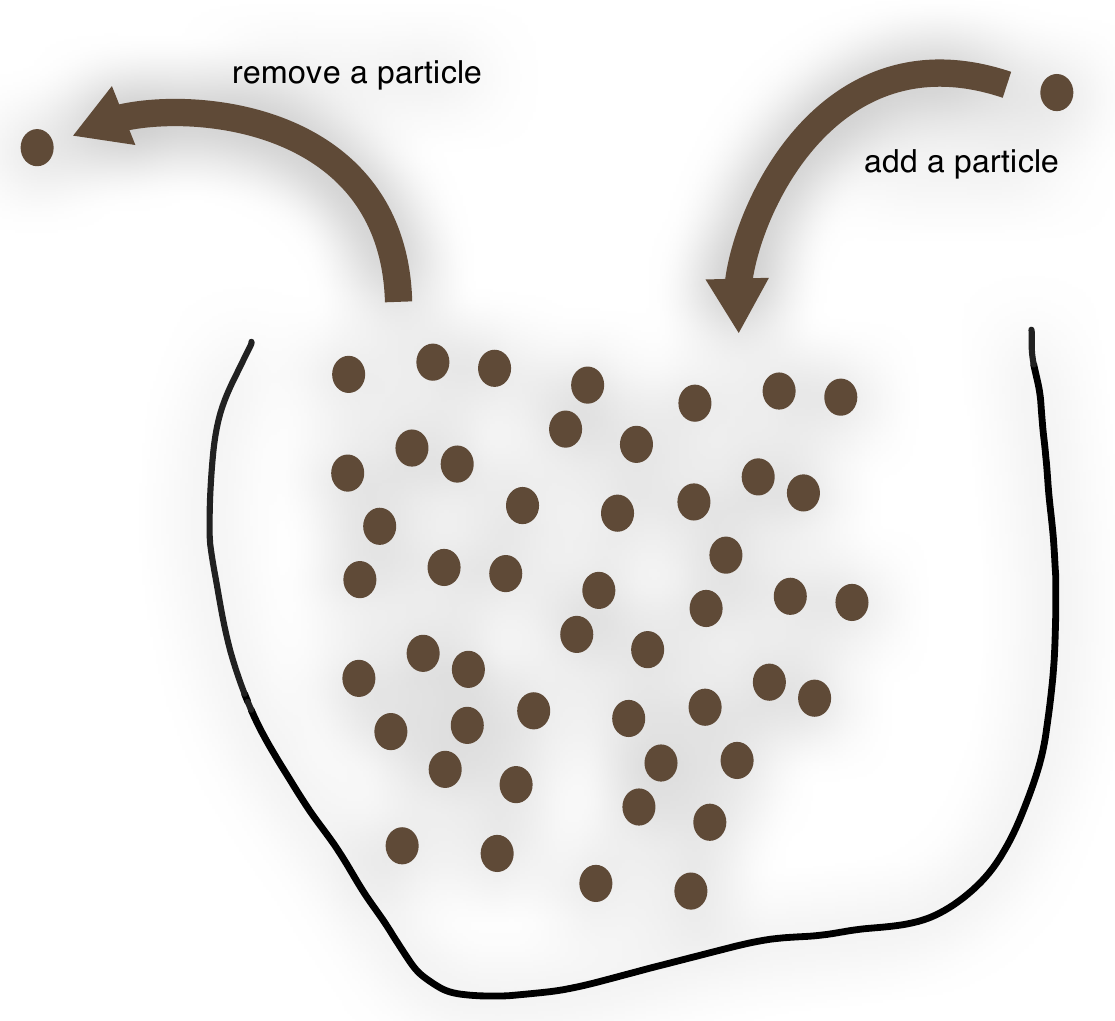}
  \caption{{\small \ Sketch of the action of the \emph{creation operator} $a^{\dag}$ (adding a particle) and of the \emph{annihilation operator} $a$ (removing a particle, in all possible ways) on an urn with a precise number $n$ of particles (i.e.\ a \emph{pure state} $\ket{n}$). See equation~\eqref{eq:HWvar} below for the concrete definitions.}}
  \label{fig:HWillustr}
\end{figure}

The precise details of the dynamics of the system include of course the way in which transitions occur. Following the general principles of stochastic mechanics introduced in the previous section, we would like to know how a transition influences the time-dependent state $\ket{\Psi(t)}$ of the system. We identify in the single-species case as \emph{pure states} configurations with a precise (and finite) number of particles $n$, which we denote by $\ket{n}$. Then according to the general formalism, $\ket{\Psi(t)}$ is a probability distribution over these pure states with time-dependent coefficients $\psi_n(t)$,
\begin{equation}\label{eq:CR1state}
	\ket{\Psi(t)}=\sum_{n\geq 0}\psi_n(t)\ket{n}\,.
\end{equation}
We quote from the literature (see e.g.\ \cite{mcquarrie1967stochastic}) the following result: for a reaction such as in~\eqref{eq:CR1} that consumes $i$ particles and produces $o$ particles at rate $r_{i,o}$, the change to the coefficient $\psi_n(t)$ in $\ket{\Psi(t)}$ obeys
\begin{equation}
\label{eq:CR2}
\frac{d}{dt}\psi_n(t)=r_{i,o}\bigg((n+i-o)_i\psi_{n+i-o}(t)-(n)_i\psi_n(t)\bigg)\,,	
\end{equation}
where we use the notational convention\footnote{We define $\Theta(x)$ as (a variant of) the Heaviside step function
\[
\Theta(x):=\begin{cases}
1\quad &\text{if } x\geq 0\\
0\quad &\text{else.}
\end{cases}
\]}
\begin{equation}
(n)_i:=i!\binom{n}{i}\equiv\Theta(n-i)\frac{n!}{(n-i)!}\,,
\end{equation}
whence $(n)_i=0$ for $n<i$.\\

Before proceeding with the analysis, it is interesting to note that the coefficient $(n)_i$ has in fact a natural \emph{combinatorial} interpretation: since we consider all particles as indistinguishable, $(n)_i$ may be interpreted as the number of ways to pick $i$ particles out of a pool of $n$ particles without replacement.\\

As standard in the literature and pioneered by M.~Delbr\"uck~\cite{delbruck1940statistical}, \eqref{eq:CR2} may be expressed in a more tractable form via the introduction of the concept of \emph{probability generating functions}: given the formula for the time-dependent state $\ket{\Psi(t)}$ of the system as in~\eqref{eq:CR1state}, one may construct the \emph{formal power series} $P(t;x)$ (with formal variable $x$):
\begin{equation}
	P(t;x):=\sum_{n\geq 0}\psi_n(t) x^n\,.
\end{equation}
We further introduce the linear operators $\partial_x$ and $\hat{x}$, defined to act on arbitrary formal power series $f(x)$ according to
\begin{equation}
	\partial_x f(x):= \tfrac{\partial}{\partial x} f(x)\,,\quad \hat{x}f(x):=xf(x)\,.
\end{equation}
Since in particular
\[
	\left(\partial_x\right)^i x^n=(n)_i\: x^{n-i}\,,
\]
this allows us to recast the information encoded in~\eqref{eq:CR2} in the form of an \emph{evolution equation} for $P(t;x)$:
\begin{equation}
	\label{eq:PGev}
	\frac{\partial}{\partial t} P(t;x)=r_{i,o}\big(\hat{x}^o-\hat{x}^i\big)\left(\partial_x\right)^i\: P(t;x)\,.
\end{equation}

The readers familiar with the theory of the Heisenberg-Weyl (HW) algebra (see e.g.~\cite{blasiak2011combinatorial}) will recognize the operators $\hat{x}$ and $\partial_x$ as the generators of the HW algebra in the \emph{Bargmann-Fock representation}, whose defining relation is computed (for $f(x)$ an arbitrary formal power series) from
\begin{equation}
	[\partial_x,\hat{x}]f(x)=f(x)\,.
\end{equation}
Contact with the standard notations of stochastic mechanics is made by employing instead the isomorphic so-called \emph{canonical representation} of the HW algebra, via the one-to-one correspondence
\begin{equation}
\begin{aligned}\label{eq:HWvar}
	x^n&\quad \leftrightarrow\quad \ket{n}\\
	\hat{x}&\quad \leftrightarrow\quad a^{\dag}\\
	\partial_x&\quad \leftrightarrow\quad a\,,
\end{aligned}	
\end{equation}
with $\ket{n}$ the pure state of $n$ particles as before. The operator $a^{\dag}$ is called the \emph{creation operator} and $a$ the \emph{annihilation operator}. One obtains via the one-to-one correspondence~\eqref{eq:HWvar} the following concrete properties of this representation\footnote{Note that this particular normalization is sometimes referred to in the literature as the ``mathematics variant''~\cite{solomon1993multi,katriel1995ordering} of bosons, as opposed to the ``physics variant'', in which the normalization is chosen more symmetrically.}
\begin{equation}
	a^{\dag}\ket{n}=\ket{n+1}\,,\quad
	a\ket{n}=(n)_1\ket{n-1}\,.
\end{equation}

We are now finally in a position to formulate the Hamiltonian $H$ of a chemical reaction system in the form proposed by M.~Doi~\cite{doi1976second}:

\begin{defn}\label{def:oneSpeciesCRS}
Consider a chemical reaction system of one species, specified in terms of an \emph{initial state},
\[
	\ket{\Psi(0)}=\sum_{n\geq 0} \psi_n(0)\ket{n},
\]
and a set of one-species reactions,
\[
	i\: A \stackrel{r_{i,o}}{\rightharpoonup} o\: A\,.
\]
Then this data induces a notion of stochastic transition system, namely a continuous-time Markov chain on the space of probability distributions over the number vectors and with infinitesimal generator
\begin{equation}\label{eq:HinfOS}
	H=\sum_{i,o\geq 0}r_{i,o}\left(\left(a^{\dag}\right)^o-\left(a^{\dag}\right)^i\right)a^i\,.
\end{equation}
We emphasize that $i$ and $o$ in the above formula take non-negative integer values, such that for example
\[
\left(a^{\dag}\right)^o=\underbrace{a^{\dag}\cdots a^{\dag}}_{\text{$o$ times}}\,.
\]
\end{defn}
Referring to the standard chemistry literature (see e.g.\ \cite{mcquarrie1967stochastic}) for the precise details, in practice only reactions with at most \emph{two} input or output particles and a total number of at most \emph{three} particles have to be considered, whence
\[
	0\leq i,o\leq 2\,,\; 1\leq i+o\leq 3\,.
\]
The corresponding formulae are presented in the first column of Table~\ref{tab:CRSone}.\\

Following the stochastic mechanics formalism, it then remains to construct a \emph{dual reference state} $\bra{\cR}$ and to identify \emph{observables}.

\begin{defn}
The \emph{dual number vectors}\footnote{Note that ``dual'' is to be understood as ``formal dual'', i.e.\ via action on basis vectors; we have to resort to this technicality since $\ell_1^{\bR}(S)$ does not admit a proper topological dual.} $\bra{m}$ (for non-negative integers $m$) are defined via their actions on basis vectors,
\begin{equation}
	\braket{m}{n}:=\delta_{m,n}\:n!\,.
\end{equation}
One may then define the \emph{resolution of the identity},
\begin{equation}
	\mathbb{1}=\sum_{n\geq0} \frac{1}{n!}\ket{n}\bra{n}\,,
\end{equation}
as well as the concrete realization of the dual reference vector $\bra{\cR}$ on the state space spanned by the number vectors,
\begin{equation}
	\bra{\cR}\equiv \sum_{n\geq 0}\frac{1}{n!}\bra{n}=\bra{0}e^a\,.
\end{equation}
\end{defn}
We refer the reader to Appendix~\ref{app:conv} for a comparison of these conventions to the analogous concepts in quantum mechanics (and in particular for a comparison of the definition of $\bra{\cR}$ to that of coherent states).\\

The dual reference vector $\bra{\cR}$ is thus in particular a \emph{left eigenvector with eigenvalue $1$ of the creation operator} $a^{\dag}$\,,
\begin{equation}
	\label{eq:refEV}
	\bra{\cR}a^{\dag}=\bra{\cR}\,,
\end{equation}
which may be seen as follows:
\begin{align*}
\bra{\cR}a^{\dag}&=\bra{0}e^a a^{\dag}=\sum_{m\geq0}\frac{1}{m!}\bra{0}a^ma^{\dag}\\
&=\sum_{m\geq 1}\frac{(m)_1}{m!}\bra{0}a^{m-1}=\bra{\cR}\,.	
\end{align*}%

It remains finally to identify a realization of \emph{observables} for single-species chemical reaction systems. It is straightforward to identify the \emph{number operator}
\begin{equation}
	\hat{n}:=a^{\dag}a
\end{equation}
as an observable, since due to
\begin{equation}
	\hat{n}\ket{n}=n\ket{n}
\end{equation}
every pure state $\ket{n}$ is an eigenvector of $\hat{n}$ as required for $\hat{n}$ to qualify as an observable according to Definition~\ref{def:obs}. Referring to equation~\eqref{eq:HWno} in Appendix~\ref{app:HWnorm} for the precise details, every element of the Heisenberg-Weyl algebra may be expressed as a linear combination of so-called \emph{normal-ordered} terms $a^{\dag\:r}a^s$. Amongst such terms, there exists a whole family of operators of the form $a^{\dag\:k}a^k$ (with $k$ a non-negative integer) that are diagonal in the basis spanned by the vectors $\ket{n}$:
\begin{equation}
a^{\dag\:k}a^k\ket{n}=(n)_k \ket{n}\,.
\end{equation}
However, there is an important relationship well-known from the combinatorics literature between powers of $\hat{n}$ and members of the above family:
\begin{thm}\label{thm:StirlingOne}
For arbitrary $k,\ell\in \bZ_{\geq0}$, we have that
\begin{equation}
	a^{\dag\:k}a^k=\sum_{\ell=0}^k s_1(k,\ell) \hat{n}^{\ell}\,,\quad \hat{n}^{\ell}=\sum_{k=0}^{\ell}S_2(\ell,k)a^{\dag\:k}a^k\,,
\end{equation}
where $s_1(k,\ell)$ are the \emph{Stirling numbers of the first kind}, while $S_2(\ell,k)$ denotes the \emph{Stirling numbers of the second kind} (see Appendix~\ref{app:Stirling} for the explicit definitions).
\begin{proof}
See for example~\cite{blasiak2005combinatorics}.
\end{proof}
\end{thm}

What this theorem entails is that for single-species reaction systems we may describe the entire statistics in terms of either the expectation value and higher order moments of the number operator $\hat{n}$, or alternatively in terms of the expectation values of the operators $a^{\dag\:k}a^k$. The latter are referred to in the statistics literature as \emph{factorial moments}.\\

With these preparations, it becomes evident that the combinatorics of normal ordering and related techniques can be of use in the study of chemical reaction systems in two main ways: either to explicitly compute the time-dependent state $\ket{\Psi(t)}$ (see Section~\ref{sec:ana}), or alternatively in order to formulate and solve evolution equations for observable generating functions. The former can be accomplished in a concise way:

\begin{thm}
	Given a one-species chemical reaction system according to Definition~\ref{def:oneSpeciesCRS}, the \emph{exponential moment generating function} $\cM(t;\lambda)$ of the number operator $\hat{n}$,
	\begin{equation}
		\cM(t;\lambda):= \left.\left\langle \cR\left\vert e^{\lambda \hat{n}}\right\vert \Psi(t)\right\rangle\right.\,,
	\end{equation}
	fulfills the \emph{evolution equation}
	\begin{equation}
		\label{eq:EGFoneEv}
		\begin{aligned}
		\tfrac{\partial}{\partial t}\cM(t;\lambda)&=\bD(\lambda,\partial_{\lambda})\cM(t;\lambda)\\
		\bD(\lambda,\partial_{\lambda})&=\sum_{i,o}r_{i,o} \bD_{(i,o)}\\
		 \bD_{(i,o)}&=\left(e^{\lambda(o-i)}-1\right)\sum_{\ell=0}^{i}s_1(i,\ell)\left(\tfrac{\partial}{\partial\lambda}\right)^{\ell}\,.
		\end{aligned}
	\end{equation}
\begin{proof}
See Appendix~\ref{app:proofThmEGFone}
\end{proof}
\end{thm}

Another interesting variant of generating function evolution equations is provided upon considering the \emph{factorial moment generating function} $\cF(t;\nu)$, defined as
\begin{equation}
	\cF(t;\nu):=\sum_{n\geq0}\frac{\nu^n}{n!}\left.\left\langle \cR\left\vert a^{\dag\:n}a^n\right\vert \Psi(t)\right\rangle\right.\,.
\end{equation}
Because the operation $\bra{\cR}\dotsc\ket{\Psi(t)}$ of taking the expectation value in the state $\ket{\Psi(t)}$ is \emph{linear} in its argument, $\cF(t;\nu)$ may be seen as a kind of ``change of basis'' for the observables of the stochastic system, i.e.\ from moments of the number operator to factorial moments. 
\begin{thm}
The generating function $\cF(t;\nu)$ obeys the following evolution equation:
\begin{equation}
	\label{eq:GFoneEv}
    \begin{aligned}
	\tfrac{\partial}{\partial t}\cF(t;\nu)&=\mathbb{d}(\nu,\partial_{\nu})\cF(t;\nu)\\
	\mathbb{d}(\nu,\partial_{\nu})&=
	\sum_{i,o}r_{i,o} \mathbb{d}_{(i,o)}\\
	\mathbb{d}_{(i,o)}&=\left((\nu+1)^o-(\nu+1)^i\right)\left(\tfrac{\partial}{\partial \nu}\right)^i\,.
    \end{aligned}
\end{equation}
\begin{proof}
See Appendix~\ref{app:proofThm3}
\end{proof}
\end{thm}

In summary, chemical reaction systems of one species may be described equivalently in terms of either of the three formulae: 
\begin{itemize}
	\item[(i)] the explicit formula for the \emph{probability distribution}, 
	\[
	\ket{\Psi(t)}=\cE(t)\ket{\Psi(0)}\,,
	\]
	\item[(ii)] the evolution equation for the \emph{exponential moment generating function $\cM(t;\lambda)$ of the number operator $\hat{n}$}, 
	\[\tfrac{\partial}{\partial t}\cM(t;\lambda)=\bD(\lambda,\partial_{\lambda})\cM(t;\lambda)\,,
	\]
	\item[(iii)] the evolution equation for the \emph{factorial moment generating function $\cF(t;\nu)$}, 
	\[\tfrac{\partial}{\partial t}\cF(t;\nu)=\mathbb{d}(\nu,\partial_{\nu})\cF(t;\nu)\,.
	\]
\end{itemize}
We summarize the explicit formulae for individual transitions $i\: A\stackrel{r_{i,o}}{\rightharpoonup}o\: A$ in Table~\ref{tab:CRSone}, both for generic transitions and for the ones of empirical relevance in the theory of chemical reaction systems (cf.\ the remarks after equation~\eqref{eq:HinfOS}).

\begin{table*}
\small
  \caption{\ Contributions $H_{(i,o)}$ to the evolution operator $H$ (cf.\ \eqref{eq:HinfOS}), $\bD_{(i,o)}$ to the differential operator $\bD\equiv\bD(\lambda,\partial_{\lambda})$ (cf.\ \eqref{eq:EGFoneEv}) and $\mathbb{d}_{(i,o)}$ to $\mathbb{d}\equiv\mathbb{d}(\nu,\partial_{\nu})$ (cf.\ \eqref{eq:GFoneEv}) of individual chemical reactions $i\: A\xrightharpoonup{\;r_{i,o}\;}o\:A$}
  \label{tab:CRSone}
  \renewcommand{\arraystretch}{1.7}
  \begin{tabular*}{\textwidth}{@{\extracolsep{\fill}}lllllll}
    \hline
    Parameters & $H_{(i,o)}$ & $\mathbb{D}_{(i,o)}$ & $\mathbb{d}_{(i,o)}$\\
    \hline
    $(i,o)$ & 
    $r_{i,o}\left(a^{\dag\: o}-a^{\dag\:i}\right)a^i$ 
    & $r_{i,o}\left(e^{\lambda(o-i)}-1\right)\sum_{\ell=0}^{i}s_1(i,\ell)\left(\tfrac{\partial}{\partial\lambda}\right)^{\ell}$
    &
    $r_{i,o}\left((\nu+1)^o-(\nu+1)^i\right)\left(\tfrac{\partial}{\partial \nu}\right)^i$
    \\
    \hline  
    $(0,2)$ & 
    $r_{0,2}\left(a^{\dag\: 2}-1\right)$ 
    & $r_{0,2}\left(e^{2\lambda}-1\right)$
    &
    $r_{0,2}\left(\nu^2+2\nu\right)$\\
    $(0,1)$ &
    $r_{0,1}(a^{\dag}-1)$
    & $r_{0,1}\left(e^{\lambda}-1\right)$
    &
    $r_{0,1}\nu$\\
    $(1,2)$ &
    $r_{1,2}\left(a^{\dag\:2}-a^{\dag}\right)a$ 
    &
    $r_{1,2}\left(e^{\lambda}-1\right)\frac{\partial}{\partial \lambda}$
    &
    $r_{1,2}(\nu^2+\nu)\frac{\partial}{\partial\nu}$\\
    $(1,0)$ &
    $r_{1,0}\left(1-a^{\dag}\right)a$ 
    &
    $r_{1,0}\left(e^{-\lambda}-1\right)\frac{\lambda}{\partial\lambda}$
    &
    $-r_{1,0}\nu\frac{\partial}{\partial \nu}$\\
   $(2,1)$ &
   $r_{2,1} \left(a^{\dag}-a^{\dag\:2}\right)a^2$
   &
   $r_{2,1}\left(e^{-\lambda}-1\right)\left(\frac{\partial^2}{\partial\lambda^2}-\frac{\partial}{\partial\lambda}\right)$
   &
   $-r_{2,1}(\nu^2+\nu)\frac{\partial^2}{\partial\nu^2}$\\
   $(2,0)$ &
   $r_{2,0}\left(1-a^{\dag\:2}\right)a^2$
   &
   $r_{2,0}\left(e^{-2\lambda}-1\right)\left(\frac{\partial^2}{\partial\lambda^2}-\frac{\partial}{\partial\lambda}\right)$
   &
   $-r_{2,0}(\nu^2+2\nu)\frac{\partial^2}{\partial\nu^2}$\\
   \hline
  \end{tabular*}
\end{table*}

\section{Multi-species reaction systems}\label{sec:CRmulti}

The key ingredient will be the canonical representation of the multi-species  Heisenberg-Weyl (HW) algebra (over the real numbers):

\begin{defn}\label{def:HWmulti}
Let $\mathbf{S}$ be a set of species (i.e.\ in general a denumerable set). Define a set of basis vectors
\begin{equation}
\ket{\vec{n}}\equiv \ket{n_{i_1},n_{i_2},\dotsc}\,,
\end{equation}
called the \emph{multi-species number vectors}. Since each of these vectors will be interpreted to represent a pure state, whence a realizable configuration of the system, we require in addition that all $n_i$ must be non-negative integers, and that only finitely many of the $n_i$ are non-zero. Then the \emph{canonical representation of the multi-species Heisenberg-Weyl algebra} is defined via
\begin{equation}\label{eq:HWmulti}
\begin{aligned}
a_i^{\dag}\ket{\vec{n}}&:=\ket{\vec{n}+\vec{\Delta}_i}\\
a_j\ket{\vec{n}}&:=(n_j)_1\ket{\vec{n}-\vec{\Delta}_j}\,,
\end{aligned}
\end{equation}
with $(n_j)_1$ denoting the \emph{falling factorials}
\[
(n_j)_1=\Theta(n_j-1)\frac{(n_j)!}{(n_j-1)!}\,.
\]
We have moreover introduced the convenient shorthand notation $\vec{\Delta}_i$ for a vector of non-negative integers whose only non-zero coordinate is $n_i=1$.
\end{defn}
The fact that these linear operators foster a representation of the multi-species HW algebra is confirmed by noting that the following \emph{canonical commutation relations}\footnote{More precisely, each of the relations in~\eqref{eq:HWmultiCCR} is an equation on linear operators, whose validity may be verified by acting on generic basis vectors $\ket{\vec{n}}$.}  hold true:
\begin{equation}\label{eq:HWmultiCCR}
[a_i,a^{\dag}_j]=\delta_{i,j}\,,\quad [a_i,a_j]=0=[a_i^{\dag},a_j^{\dag}]\,.
\end{equation}

The following additional multi-index notations will prove convenient throughout this section:
\begin{equation}
\begin{aligned}
\vec{n}!&\equiv \prod_{i\in \mathbf{S}} (n_i)!\,,\;
\vec{x}^{\vec{y}}\equiv \prod_{i\in \mathbf{S}}\left(x_i^{y_i}\right)\,,\; e^{\vec{x}}\equiv\prod_{i\in \mathbf{S}}e^{x_i}\\
(\vec{m})_{\vec{n}}&\equiv\prod_{i\in \mathbf{S}}\left((m_i)_{n_i}\right)
:= \prod_{i\in\mathbf{S}}\left(\frac{(m_i)!}{(m_i-n_i)!}\right)\\
\binom{\vec{m}}{\vec{n}}&=\prod_{i\in \mathbf{S}}\binom{m_i}{n_i}\,,\;
\delta_{\vec{m},\vec{n}}\equiv \prod_{i\in \mathbf{S}}\delta_{m_i,n_i}\,.
\end{aligned}
\end{equation}

An immediate consequence of~\eqref{eq:HWmultiCCR} is that we can express any element of the multi-species HW algebra in terms of linear combinations of \emph{normal-ordered} expressions of the form $\vec{a}^{\dag\:\vec{r}}\vec{a}^{\vec{s}}$ (with $\vec{r},\vec{s}$ vectors of non-negative integers), which is a consequence of the well-known \emph{normal ordering formula} (cf.\ e.g.\ \cite{blasiak2005combinatorics})
\begin{equation}\label{eq:HWmultiNO}
\begin{aligned}
&\vec{a}^{\dag\:\vec{m}}\vec{a}^{\vec{n}}\vec{a}^{\dag\:\vec{r}}\vec{a}^{\vec{s}}=
\sum_{\vec{k}\geq\vec{0}}\vec{k}!\binom{\vec{n}}{\vec{k}}\binom{\vec{r}}{\vec{k}}\vec{a}^{\dag\:(\vec{m}+\vec{r}-\vec{k})}\vec{a}^{(\vec{n}+\vec{s}-\vec{k})}\,.
\end{aligned}
\end{equation}
Here, we have made use of the standard convention $\binom{x}{y}=0$ whenever $y>x$.\\

We will use the following normalization for the (formal) dual vectors $\bra{\vec{m}}$:
\begin{equation}
\braket{\vec{m}}{\vec{n}}:=\vec{n}!\delta{\vec{m},\vec{n}}\,.
\end{equation}
The multi-species version $\bra{\vec{\cR}}$ of the \emph{reference dual vector} is defined as
\begin{equation}
\bra{\vec{\cR}}:=\sum_{\vec{m}\geq\vec{0}}\frac{1}{\vec{m}!}\bra{\vec{m}}\,.
\end{equation}
With the relations
\begin{equation}
\ket{\vec{n}}=\vec{a}^{\dag\:\vec{n}}\ket{\vec{0}}\,,\; \bra{\vec{m}}=\left(\vec{a}^{\vec{m}}\ket{\vec{0}}\right)^{\dag}=\bra{\vec{0}}\vec{a}^{\vec{m}}\,,
\end{equation}
which are compatible with $\braket{\vec{m}}{\vec{n}}=\vec{n}!\:\delta_{\vec{m},\vec{n}}$, we may identify $\bra{\vec{\cR}}$ as a \emph{simultaneous left eigenvector} of the creation operators $a^{\dag}_i$ ($i\in \mathbf{S}$) of eigenvalue $1$:
\begin{equation}
\bra{\vec{\cR}}a_i^{\dag}=\bra{\vec{0}}e^{\vec{a}}a_i^{\dag}
=\bra{\vec{0}}(a_i^{\dag}+1)e^{\vec{a}}=\bra{\vec{\cR}}\,.
\end{equation}
Here, we have made use of the following very useful general formula (which played a pivotal role in the derivation of several analytical combinatorics formulae, see Section~\ref{sec:ana}):
\begin{prop}
For all real vectors $\vec{\alpha},\vec{\beta}$, and for all Taylor-expandable functions $f(\vec{a}^{\dag},\vec{a})$, we have that
\begin{equation}\label{eq:HWmAux}
e^{\vec{\alpha}\cdot \vec{a}}f(\vec{a}^{\dag},\vec{a})e^{\vec{\beta}\cdot\vec{a}^{\dag}}=
e^{\vec{\alpha}\cdot\vec{\beta}}e^{\vec{\beta}\cdot\vec{a}^{\dag}}f\left(\vec{a}^{\dag}+\vec{\alpha},\vec{a}+\vec{\beta}\right)e^{\vec{\alpha}\cdot\vec{a}}\,.
\end{equation}
\begin{proof}
See Appendix~\ref{app:HWmAuxproof}.
\end{proof}
\end{prop}
Noticing that the multi-species number vectors are simultaneous eigenstates of the \emph{number operators} $\hat{n}_i:=a_i^{\dag}a_i$ ($i\in \mathbf{S}$),
\begin{equation}
\hat{n}_i\ket{\vec{n}}=(n_i)_1\ket{\vec{n}}\,,
\end{equation}
we have all preparations in place to state the multi-species generalization of the stochastic mechanics framework for chemical reaction systems:
\begin{defn}\label{def:multiSpeciesCRS}
Consider a chemical reaction system with reactions of the form
\begin{equation}
\vec{i}\cdot \vec{A} \xrightharpoonup{\;r_{\vec{i},\vec{o}}\;}\vec{o}\cdot\vec{A}\,,
\end{equation}
where $\vec{i},\vec{o}$ are vectors of non-negative integers (encoding the numbers of particles of each type on input and output of the transition, respectively), and with the \emph{reaction rates}\footnote{We use this convention in order to be able to specify a compact formula for $H$, i.e.\ only the reactions for parameter pairs $(\vec{i},\vec{o})$ for $r_{\vec{i},\vec{o}}>0$ contribute to $H$.} $r_{\vec{i},\vec{o}}\in \bR_{\geq0}$. Then together with an \emph{input state}
\begin{equation}
\ket{\Psi(0)}\in Prob(\cN)
\end{equation}
over the \emph{state space $\cN$ of multi-species number vectors},
\begin{equation}
Prob(\cN)\ni \ket{\Psi(0)}\equiv\sum_{\vec{n}\geq\vec{0}}\psi_{\vec{n}}(0)\ket{\vec{n}}\,,
\end{equation}
this data specifies a continuous-time Markov chain with \emph{evolution operator}
\begin{equation}\label{eq:HmultiDef}
H:=\sum_{\vec{i},\vec{o}\geq \vec{0}}r_{\vec{i}, \vec{o}}\left(\left(\vec{a}^{\dag}\right)^{\vec{o}}-\left(\vec{a}^{\dag}\right)^{\vec{i}}\right)\vec{a}^{\vec{i}}\,.
\end{equation}
\end{defn}

The main general result on the dynamics of generic chemical reaction systems are the following two generalizations of the moment and factorial moment generating function evolution equations from the previous section. As in the single-species setting, knowing all eigenvalues $n_i$ of a given pure state $\ket{\vec{n}}$ completely characterizes the eigenstate, such that the moments of the number operators $\hat{n}_i$ or the respective factorial moments are providing equivalent knowledge about a system state $\ket{\Psi(t)}$ as the probability distribution itself:
\begin{thm}\label{thm:MultiEv}
For a multi-species chemical reaction system as specified in Definition~\ref{def:multiSpeciesCRS}, let the \emph{exponential moment generating function} $\cM(t;\vec{\lambda})$ of the number operators $\hat{n}_i$ and the \emph{factorial moment generating function} $\cF(t;\vec{\nu})$ be defined as
\begin{equation}\label{eq:defEGFmultiMF}
\begin{aligned}
\cM(t;\vec{\lambda})&:= \left.
\left\langle \vec{\cR}\left\vert e^{\vec{\lambda}\cdot\vec{\hat{n}}}\right\vert\Psi(t)\right\rangle
\right.\\
\cF(t;\vec{\nu})&:=\sum_{\vec{k}\geq\vec{0}}
\frac{\vec{\nu}^{\vec{k}}}{\vec{n}!}\left.\left\langle \vec{\cR}\left\vert \vec{a}^{\dag\:\vec{k}}\vec{a}^{\vec{k}}\right\vert\Psi(t)\right\rangle\right.\,.
\end{aligned}
\end{equation}
Then the generating functions satisfy the \emph{evolution equations}
\begin{equation}\label{eq:EvMultiM}
\begin{aligned}
\tfrac{\partial}{\partial t}\cM(t;\vec{\lambda})&=\bD(\vec{\lambda},\partial_{\vec{\lambda}})\cM(t;\vec{\lambda})\\
\bD(\vec{\lambda},\partial_{\vec{\lambda}})&=\sum_{\vec{i},\vec{o}}
\bD_{(\vec{i},\vec{o})}\\
\bD_{(\vec{i},\vec{o})}&=r_{\vec{i},\vec{o}}\left(e^{\vec{\lambda}\cdot(\vec{o}-\vec{i})}-1\right)\cdot\\
&\quad\qquad \cdot\sum_{\vec{k}=\vec{0}}^{\vec{i}}\vec{s}_1(\vec{i},\vec{k})\left(\tfrac{\partial}{\partial\vec{\lambda}}\right)^{\vec{k}}\\
\vec{s}_1(\vec{i},\vec{k})&:=\prod_{j\in \mathbf{S}}s_1(i_j,k_j)\,,
\end{aligned}
\end{equation}
with $s_1(i,j)$ denoting the Stirling numbers of the first kind, and
\begin{equation}\label{eq:FmultiEv}
\begin{aligned}
\tfrac{\partial}{\partial t}\cF(t;\vec{\nu})&=\mathbb{d}(\vec{\nu},\partial_{\vec{\nu}})\cF(t;\vec{\nu})\\
\mathbb{d}(\vec{\nu},\partial_{\vec{\nu}})&=
\sum_{\vec{i},\vec{o}}\mathbb{d}_{(\vec{i},\vec{o})}\\
\mathbb{d}_{(\vec{i},\vec{o})}&=
r_{\vec{i},\vec{o}}\left((\vec{\nu}+\vec{1})^{\vec{o}}-(\vec{\nu}+\vec{1})^{\vec{i}}\right)\left(\tfrac{\partial}{\partial\vec{\nu}}\right)^{\vec{i}}\,.
\end{aligned}
\end{equation}
\begin{proof}
See Appendix~\ref{app:proofsHWMev}.
\end{proof}
\end{thm}
We summarize both these statements as well as explicit formulae for all  chemical reactions of practical interest (i.e.\ for transitions with $1\leq \sum_{j\in \mathbf{S}}(i_j+o_j)\leq 3$) in Table~\ref{tab:CCRmulti}.\\

\begin{table*}
\small
  \caption{\ Contributions $H_{(\vec{i},\vec{o})}$ to the evolution operator $H$ (cf.\ \eqref{eq:HmultiDef}), $\bD_{(\vec{i},\vec{o})}$ to the differential operator $\bD\equiv\bD(\vec{\lambda},\partial_{\vec{\lambda}})$ (cf.\ \eqref{eq:EvMultiM}) and $\mathbb{d}_{(\vec{i},\vec{o})}$ to $\mathbb{d}\equiv\mathbb{d}(\vec{\nu},\partial_{\vec{\nu}})$ (cf.\ \eqref{eq:FmultiEv}) of individual multi-species chemical reactions $\vec{i}\cdot\vec{A}\xrightharpoonup{\;r_{\vec{i},\vec{o}}\;}\vec{o}\cdot\vec{A}$; concrete values for $\vec{i}$ and $\vec{o}$ are encoded using the notation $\vec{\Delta}_{\alpha}$ for the vector with all entries equal to zero, except for the $\alpha$-th entry, which is $1$ (i.e.\ $\vec{\Delta}_{\alpha}$ is the $\alpha$-th unit basis vector)}
  \label{tab:CCRmulti}
  \renewcommand{\arraystretch}{1.7}
  \begin{tabular*}{\textwidth}{@{\extracolsep{\fill}}lllllll}
    \hline
    Parameters & $H_{(\vec{i},\vec{o})}/r_{\vec{i},\vec{o}}$ & $\mathbb{D}_{(\vec{i},\vec{o})}/r_{\vec{i},\vec{o}}$ & $\mathbb{d}_{(\vec{i},\vec{o})}/r_{\vec{i},\vec{o}}$\\
    \hline
    $(\vec{o},\vec{i})$ & 
    $\left(\vec{a}^{\dag\: \vec{o}}-\vec{a}^{\dag\:\vec{i}}\right)\vec{a}^{\vec{i}}$ 
    & $ 
    \left(e^{\vec{\lambda}\cdot(\vec{o}-\vec{i})}-1\right)  \sum_{\vec{\ell}=\vec{0}}^{\vec{i}}
    \vec{s}_1(\vec{i},\vec{\ell})
    \left(\frac{\partial}{\partial\vec{\lambda}}\right)^{\vec{\ell}}$
    &
    $\left((\vec{\nu}+\vec{1})^{\vec{o}}-(\vec{\nu}+\vec{1})^{\vec{i}}\right)
    \left(\frac{\partial}{\partial \vec{\nu}}\right)^{\vec{i}}$
    \\
    \hline  
    $(\vec{\Delta}_{\alpha}+\vec{\Delta}_{\beta},\vec{0})$ 
    & 
    $\left(a_{\alpha}^{\dag}a_{\beta}^{\dag}-1\right)$ &
    $\left(e^{\lambda_{\alpha}+\lambda_{\beta}}-1\right)$
    &
    $ \nu_{\alpha}\nu_{\beta}+\nu_{\alpha}+\nu_{\beta}$\\
    $(\vec{\Delta}_{\alpha},\vec{0})$ &
    $a_{\alpha}^{\dag}-1$ &
    $e^{\lambda_{\alpha}}-1$ &
    $\nu_{\alpha}$\\
    $(\vec{0},\vec{\Delta}_{\gamma})$ &
    $\left(1-a_{\gamma}^{\dag}\right)a_{\gamma}$ &
    $\left(e^{-\lambda_{\gamma}}-1\right)\frac{\partial}{\partial\lambda_{\gamma}}$ &
    $-\nu_{\gamma}\frac{\partial}{\partial\nu_{\gamma}}$\\
      $(\vec{\Delta}_{\alpha},\vec{\Delta}_{\beta})$ &
    $\left(a^{\dag}_{\alpha}-a^{\dag}_{\beta}\right)a_{\beta}$ &
    $\left(e^{\lambda_{\alpha}-\lambda_{\beta}}-1\right)\frac{\partial}{\partial \lambda_{\beta}}$ &
    $ (\nu_{\alpha}-\nu_{\beta})\frac{\partial}{\partial \nu_{\beta}}$ \\
    $(\vec{\Delta}_{\alpha}+\vec{\Delta}_{\beta},\vec{\Delta}_{\gamma})$ &
    $\left(a_{\alpha}^{\dag}a_{\beta}^{\dag}-a_{\gamma}^{\dag}\right)a_{\gamma}$ &
    $\left(e^{\lambda_{\alpha}+\lambda_{\beta}-\lambda_{\gamma}}-1\right)\frac{\partial}{\partial\lambda_{\gamma}}$ &
    $(\nu_{\alpha}\nu_{\beta}+\nu_{\alpha}+\nu_{\beta}-\nu_{\gamma})\frac{\partial}{\partial \nu_{\gamma}}$\\
    $(\vec{0},\vec{\Delta}_{\beta}+\vec{\Delta}_{\gamma})$ &
    $\left(1-a_{\beta}^{\dag}a_{\gamma}^{\dag}\right)a_{\beta}a_{\gamma}$ &
    $\left(e^{-\lambda_{\beta}-\lambda_{\gamma}}-1\right)
    \left(
    \frac{\partial^2}{\partial\lambda_{\beta}\partial\lambda_{\gamma}}-\delta_{\beta,\gamma}\frac{\partial}{\partial\lambda_{\gamma}}
    \right)$ &
    $-\left(\nu_{\beta}\nu_{\gamma}+\nu_{\beta}+\nu_{\gamma}\right)\frac{\partial^2}{\partial\nu_{\beta}\partial\nu_{\gamma}}$\\
    $(\vec{\Delta}_{\alpha},\vec{\Delta}_{\beta}+\vec{\Delta}_{\gamma})$ &
    $\left(a_{\alpha}^{\dag}-a_{\beta}^{\dag}a_{\gamma}^{\dag}\right)a_{\beta}a_{\gamma}$ &
    $\left(e^{\lambda_{\alpha}-\lambda_{\beta}-\lambda_{\gamma}}-1\right)\left(\frac{\partial^2}{\partial\lambda_{\beta}\partial\lambda_{\gamma}}
    -\delta_{\beta,\gamma}\frac{\partial}{\partial\lambda_{\gamma}}
    \right)$ &
 $\left(\nu_{\alpha}-\nu_{\beta}\nu_{\gamma}-\nu_{\beta}-\nu_{\gamma}\right)\frac{\partial^2}{\partial\nu_{\beta}\partial\nu_{\gamma}}$\\
 \hline
  \end{tabular*}
\end{table*}

As an aside and to conclude the discussion of this general framework, it may be worthwhile noting that while as in the single-species cases only a very limited number of reaction parameters $(\vec{i},\vec{o})$ is realized in realistic chemical reaction systems, the multi-species formulae presented thus far in fact can describe \emph{generic discrete transition systems}, whence \emph{discrete graph rewriting systems} (compare~\cite{bdg2016,bCCDD2017}). Interesting special cases would for example include \emph{branching processes} such as 
\begin{equation}
A_i\xrightharpoonup{\;r_{i,\vec{o}}\;}\vec{o}\cdot\vec{A}\,.
\end{equation}
Since branching processes constitute a class of \emph{semi-linear transition systems}, one could in principle envision to provide exact closed-form formulae for the respective time-dependent states $\ket{\Psi(t)}$ as well as for the generating functions $\cM(t;\vec{\lambda})$ and $\cF(t;\vec{\nu})$ along the lines presented in Section~\ref{sec:ana}.

\section{Reaction systems with first order moment closure}
\label{sec:SCCR}

The stochastic mechanics framework in its ``threefold view'' of chemical reaction systems in terms of time-dependent state $\ket{\Psi(t)}$, exponential moment generating function $\cM(t;\vec{\lambda})$ and factorial moment generating function $\cF(t;\vec{\nu})$, offers a number of very interesting insights into the combinatorics that govern chemical transitions. While in Section~\ref{sec:ana} techniques will be presented to derive explicit formulae for $\ket{\Psi(t)}$, in this section we will focus on some general properties of moment evolution equations.\\

In practical applications, amongst the simplest to analyze reaction systems are those that possess the property of \emph{first order moment closure}, whence for which the evolution equations for the first moments of the number operators $\hat{n}_i$ read
\begin{equation}\label{eq:FOMC}
\tfrac{d}{dt}\bra{\vec{\cR}} \hat{n}_i\ket{\Psi(t)}=\sum_{j\in\mathbf{S}} \mu_{ij}\bra{\vec{\cR}}  \hat{n}_j\ket{\Psi(t)}\,,
\end{equation}
with $(\mu_{ij})$ a matrix of coefficients. In other words, the special property encoded in~\eqref{eq:FOMC} is that the time-derivatives of the expectation values of the number operators only depend on the expectation values of the number operators, and not on any higher moments thereof. It is straightforward to derive the class of reaction systems with this property from the evolution equation for the factorial moments (noting that $\hat{n}_i=a^{\dag}_ia_i$, whence that the first moments and first factorial moments coincide). To this end, let us expand the factorial moment generating function $\cF(t;\vec{\nu})$ more explicitly (compare~\eqref{eq:defEGFmultiMF}):
\begin{equation}
\begin{aligned}
\cF(t;\vec{\nu})&=\sum_{\vec{n}\geq\vec{0}}
\frac{\vec{\nu}^{\vec{n}}}{\vec{n}!}\bra{\vec{\cR}} \vec{a}^{\dag\:\vec{n}}\vec{a}^{\vec{n}}\ket{\Psi(t)}\\
&\equiv\sum_{\vec{n}\geq\vec{0}}\frac{\vec{\nu}^{\vec{n}}}{\vec{n}!}f_{\vec{n}}(t)\,.
\end{aligned}
\end{equation}
Then by expanding the evolution equation for $\cF(t;\vec{n})$ and collecting terms contributing to $\frac{d}{d t}f_{\vec{n}}(t)$, we obtain the following:
\begin{thm}
Consider a multi-species chemical reaction system as specified in Definition~\ref{def:multiSpeciesCRS}. Then the evolution equations for the factorial moments $f_{\vec{n}}(t)$ read
\begin{equation}
\label{eq:FMev}
\begin{aligned}
\tfrac{d}{dt}f_{\vec{n}}(t)&=\sum_{\vec{i},\vec{o}}r_{\vec{i},\vec{o}}
\sum_{\vec{k}}\phi_{\vec{k}}(\vec{n};\vec{i},\vec{o})f_{\vec{n}+\vec{i}-\vec{k}}(t)\\
\phi_{\vec{k}}(\vec{n};\vec{i},\vec{o})&=
\bigg[
(\vec{o})_{\vec{k}}-(\vec{i})_{\vec{k}}
\bigg]\binom{\vec{n}}{\vec{k}}\,.
\end{aligned}
\end{equation}
Here, we have made use of the standard conventions $\binom{x}{y}=0$ and $(x)_y=0$ whenever $y>x$.\\

Specializing to $\vec{n}=\vec{\Delta}_{\alpha}$, i.e.\ to $f_{\vec{\Delta}_{\alpha}}$ denoting the first moment of the number vector $\hat{n}_{\alpha}$, one obtains the evolution equations for the first moments as
\begin{equation}
\frac{d}{dt}f_{\vec{\Delta}_{\alpha}}(t)
=\sum_{\vec{i},\vec{o}}r_{\vec{i},\vec{o}}(o_{\alpha}-i_{\alpha})f_{\vec{i}}(t)\,.
\end{equation}
Therefore, the only reaction systems for which we have \emph{first order moment closure} are \emph{semi-linear} reaction systems, i.e.\ systems for which for all $(\vec{i},\vec{o})$ with rates $r_{\vec{i},\vec{o}}>0$ we have that, for $\vec{i}\equiv(i_1,i_2,\dotsc)$,
\begin{equation}
\sum_{j\in \mathbf{S}} i_j\leq 1\,.
\end{equation}
In other words, the condition is satisfied for choices of $\vec{i}$ where either all entries $i_1,i_2,\dotsc$ are zero, or where only one of the $i_j$ is non-zero and equal to $1$.
\end{thm}
The fact that semi-linear reaction (and branching) systems are the only systems with first order moment closure already indicates that they are somewhat simpler in their dynamical structure than generic systems. In Section~\ref{sec:ana} we will demonstrate in fact that all of them can be solved analytically. In view of practical applications, the following no-go corollary is important to note:
\begin{cor}
The factorial moment evolution equation~\eqref{eq:FMev} entails that there are \emph{no factorial moment closures} beyond semi-linear reaction systems. Because of the relationship between moments of the number operators and the factorial moments as presented in Theorem~\ref{thm:StirlingOne}, the same statement holds true for the moments of the number operators.
\end{cor}
In other words, we will in particular not be able to rely on a simple matrix exponential in order to solve the evolution equations for generic reaction systems. Nonetheless, there still exist interesting cases in which one encounters a form of a higher-order factorial moment closure, namely when we consider catalytic reaction systems (or other systems in which the numbers of certain types of particles are conserved): suppose that the reaction system leaves the number $n_i$ of a certain species of particles constrained in the range $0\leq n_i\leq N_i$ when initialized on a state $\ket{\Psi(0)}$ that satisfies this constraint. Then the factorial moments of $\hat{n}_i$ of orders greater than $N_i$ will vanish identically, which may in certain cases provide some interesting computation strategies. We refer the interested readers to the standard literature (such as the review paper~\cite{mcquarrie1967stochastic}) for explicit examples and further technical details.

\section{Analytical solution strategies for probability generating functions}
\label{sec:ana}

For the material presented in this section, it will prove most convenient to work in the Bargmann-Fock representation of the Heisenberg-Weyl algebra. This entails to focus on the dynamical evolution properties of the \emph{probability generating function} 
\begin{equation}
	P(t;x)=\sum_{\vec{n}\geq\vec{0}} p_{\vec{n}}(t)\vec{x}^{\vec{n}}
\end{equation}
of the state
\begin{equation}
	\ket{\Psi(t)}=\sum_{\vec{n}\geq \vec{0}} p_{\vec{n}}(t)\ket{\vec{n}}
\end{equation}
of the chemical reaction system. We will employ two rather different strategies to make progress: the first one involves the use of a result on boson normal-ordering~\cite{blasiak2005boson,blasiak2011combinatorial,Dattoli:1997iz} for the special case of reactions that are not binary (i.e.\ that consume at most one particle per elementary transition), while the second method, applicable precisely to said binary reactions, involves the use of \emph{(Sobolev-) orthogonal polynomials}.

\subsection{Solving non-binary reactions via semilinear boson normal-ordering}
\label{sec:nBbnor}

Our main technical tool will be a normal-ordering formula developed by some of the authors~\cite{blasiak2005boson} (presented here in the multi-variate version as e.g.\ described in~\cite{blasiak2011combinatorial}; see also Dattoli et al.~\cite{Dattoli:1997iz} for an account of the original versions and uses of this approach in the realm of formal power series):
\begin{thm}\label{thm:seminal}
Consider an $N$-species Heisenberg-Weyl algebra, whose generators read in the Bargmann-Fock representation $\hat{x}_i$ (multiplication by $x_i$) and $\partial_{x_j}$ (partial derivation by $x_j$) for $1\leq i,j\leq N$. Let $\cH$ (i.e.\ the Hamiltonian $H\equiv H(\vec{a}^{\dag},\vec{a})$ according to~\eqref{eq:HmultiDef} expressed in the Bargmann-Fock basis following~\eqref{eq:HWvar}) be a \emph{semi-linear} expression in the generators,
\begin{equation}
	\cH=v(\hat{\vec{x}})+\sum_{i=0}^N q_i(\vec{\hat{x}})\:\partial_{x_i}\,,
\end{equation}
with $q_i(\vec{\hat{x}})$ and $v(\hat{\vec{x}})$ some functions in the operators $\hat{x}_i$. Let $F(0;\vec{x})$ be an entire function in the indeterminates $x_i$. Define the formal power series 
\begin{equation}
F(\lambda;\vec{x}):=e^{\lambda \cH}\:F(0;\vec{x})
\end{equation}
with formal variable $\lambda$, which is the \emph{exponential generating function} for terms of the form
\[
\cH^{n}F(0;\vec{x})\,,
\]
thus describing the $n$-fold application of $\cH$ to $F(0;\vec{x})$. Then $F(\lambda;\vec{x})$ may be expressed in closed form as follows:\footnote{Equation~\eqref{eq:gEq} is sometimes expressed (see e.g.~\cite{blasiak2005boson}) in the alternative form
\[
	\tfrac{\partial}{\partial \lambda} \ln(g(\lambda;\vec{x}))=v(\vec{T}(\lambda;\vec{x}))\,,\quad g(0;\vec{x})=1\,.
\]}
\begin{subequations}
\begin{align}\label{eq:seminal}
	F(\lambda;\vec{x})&=g(\lambda;\vec{x})F\big(0;\vec{T}(\lambda;\vec{x})\big)\\
	\tfrac{\partial}{\partial\lambda}T_i(\lambda;\vec{x})&=q_i(\vec{T}(\lambda;\vec{x}))\,,\quad T_i(0;\vec{x})=x_i\\
	\ln g(\lambda;\vec{x})&=\int_0^{\lambda} v(\vec{T}(\kappa;\vec{x}))d\kappa\,.\label{eq:gEq}
\end{align}
\end{subequations}
Moreover, we have the relations
\begin{equation}
	\label{eq:groupSeminalTG}
	\begin{aligned}
		\vec{T}(\lambda+\mu;\vec{x})&=\vec{T}(\mu;\vec{T}(\lambda;\vec{x}))\\
		g(\lambda+\mu;\vec{x})&=g(\lambda;\vec{x})g(\mu;\vec{T}(\lambda;\vec{x}))\,,
	\end{aligned}
\end{equation}
which in turn implies that $e^{\lambda \cH}$ induces a \emph{one-parameter group} of transformations.
\end{thm}
Note that in the above formulation, we understand the notion of integration as a notion of \emph{formal} integration, whence as the operation on formal power series induced by the action of formal integration on monomials in the formal variables,
\[
\int dx_i \; x_i^{n_i}=\frac{x_i^{n_i+1}}{(n_i+1)}\,.
\]

As might not be immediately evident to the non-combinatorists, this result in fact provides a closed-form solution to a vast class of normal-ordering problems. Recall that in the HW algebra every element may be presented as a linear combination in normal-ordered terms, i.e.\ in terms of the form $\hat{\vec{x}}^{\vec{r}}\partial_{\vec{x}}^{\vec{s}}$ (for multi-indices $\vec{r},\vec{s}\geq \vec{0}$). To illustrate the practical value of Equation~\eqref{eq:seminal}, consider for example computing the term $\cH^{2}$ explicitly:
\begin{align*}
	\cH^{2}&=\left(v(\vec{\hat{x}})+\sum_{i=1}^n q_i(\vec{\hat{x}})\partial_{x_i}\right)^{2}\\
	&=\colon\left(v+\vec{q}\cdot \partial_{\vec{x}}\right)^2\colon\\
	&\quad +\sum_{i=1}^n q_i\:\partial_{x_i}
	+\sum_{i,j=1}^n q_i\left(\tfrac{\partial q_j}{\partial x_i}\right)\partial_{x_j}
	\,.
\end{align*}%
Here, the notation $\colon\dotsc\colon$ stands for the \emph{forgetful normal-ordering} operation (c.f.\ e.g.\ \cite{blasiak2007combinatorics}) , defined as bringing the expression $\dotsc$ into normal-ordered form \emph{ignoring} the commutation relations
\[
[\hat{x}_i,\partial_{x_j}]=\delta_{i,j}\,.
\]
Thus, in the example above we find concretely
\[
\colon\left(v+\vec{q}\cdot \partial_{\vec{x}}\right)^2\colon
=v^2+2v \vec{q}\cdot \partial_{\vec{x}}+\sum_{i,j=1}^n q_iq_j\partial_{x_i}\partial_{x_j}\,.
\]
As this example computation demonstrates, finding closed-form solutions to the normal-ordering problem is a highly intricate task, whence it is convenient that for the semi-linear operators Theorem~\ref{thm:seminal} provides such a practicable solution.\\

Returning to the main topic of the present paper, we are now in a position to formulate one of our key results. The argument employs some standard notions and results from the theory of probability distributions and of probability generating functions (see e.g.\ the modern textbooks~\cite{gut2012probability,klenke2013probability}), presented for convenience in the following proposition:
\begin{prop}
For a \emph{discrete random variable} $X$, i.e.\ for a random variable taking values in the non-negative integers, denote by $F_X(x)$ its \emph{probability generating function} (PGF),
\begin{equation}
F_X(x):=\sum_{n\geq 0} Pr(X=n)x^n\,.
\end{equation}
The following two operations on PGFs of discrete random variables yield well-defined PGFs of random variables:
\begin{itemize}
	\item[(i)] (\emph{Multiplication theorem}) Let $X_1,\dotsc,X_N$ be independent discrete random variables, with PGFs $F_1(x),\dotsc,F_N(x)$, respectively. Then 
	\[
	Z=X_1+\dotsc+X_N
	\]
is a discrete random variable with PGF
\begin{equation}
	\label{eq:PGFconv}
	F_Z(x)=\prod_{i=1}^N F_i(x)\,.	
	\end{equation}
On the level of discrete probability distributions, this amounts to computing the \emph{convolution} of the distributions of the random variables $X_i$.
	\item[(ii)] Let $M$ be a discrete random variable with PGF $F_M(x)$, and denote by $\{X_i\}_{i\in \mathbb{N}}$ a family of \emph{independent identically distributed} discrete random variables, i.e.\ with common PGF $F_X(x)$. Then the \emph{compound distribution} of the discrete random variable
	\[
		Y=X_1+\dotsc +X_M
	\]
		has the probability generating function
	\begin{equation}
		\label{eq:PGFcompound}
		F_Y(x)=F_M(F_X(x))\,.
	\end{equation}
\end{itemize}
Moreover, both concepts generalize in a straightforward fashion to multivariate generating functions.
\end{prop}

With these preparations, we find:

\begin{thm}
\label{thm:sem}
Consider an $N$-species chemical reaction system specified according to Definition~\ref{def:multiSpeciesCRS}. Let $P(0;\vec{x})$ be the generating function of a probability distribution of a random variable taking values in $N$-tuples of non-negative integers, whence
\begin{equation}
	P(0;\vec{x})=\sum_{\vec{n}\geq \vec{0}}p_{\vec{n}}(0)\vec{x}^{\vec{n}}\,,\;
	P(0;\vec{1})=\sum_{\vec{n}\geq\vec{0}}p_{\vec{n}}(0)=1\,,
\end{equation}
and with all coefficients $p_{\vec{n}}(0)$ non-negative real numbers. Let $\cH$ be the infinitesimal generator of the system (whence the operator $H$ of Definition~\ref{def:multiSpeciesCRS} expressed in the Bargmann-Fock basis). Then for each of the non-binary elementary chemical reactions, acting with $e^{\lambda \cH}$ on the initial state's probability generating function $P(0;\vec{x})$ results in a new generating function $P(\lambda;\vec{x})$,
\begin{equation}
\begin{aligned}
	P(\lambda;\vec{x})&=e^{\lambda \cH} P(0;\vec{x})\\
	&=g(\lambda;\vec{x})P\big(0;\vec{T}(\lambda;\vec{x})\big)\,,
\end{aligned}
\end{equation}
with $g(\lambda;\vec{x})$ and $\vec{T}(\lambda;\vec{x})$ computed according to~\eqref{eq:seminal} of Theorem~\ref{thm:seminal}. The resulting formal power series $P(\lambda;\vec{x})$ is a \emph{well-posed probability distribution} for all (finite) values of $\lambda\in \bR_{\geq 0}$. 
\begin{proof}
Referring to the concrete results of applying Theorem~\ref{thm:seminal} as  listed in Table~\ref{tab:nonBinaryCRs}, in all cases we have that $g(\lambda;\vec{x})$ as well as $\vec{T}(\lambda;\vec{x})$ are in fact themselves probability generating functions of certain discrete $N$-multivariate random variables. Therefore, by a combination of equations~\eqref{eq:PGFconv} and~\eqref{eq:PGFcompound}, the formula for $P(\lambda;\vec{x})$ as stated identifies $P(\lambda;\vec{x})$ as the PGF of the convolution of the distribution encoded in $g(\lambda;\vec{x})$ with the compound distribution encoded in $P(0;\vec{T}(\lambda;\vec{x}))$ whenever $\lambda\in \bR_{\geq0}$. 	
\end{proof}
\end{thm}

Interestingly, according to the results as presented in Table~\ref{tab:nonBinaryCRs}, the probability generating function $P(t;\vec{x})$ for a given reaction system with an individual non-binary chemical reaction may in all cases be expressed in terms of convolutions and formation of compound distributions involving some standard probability laws (Poisson, Bernoulli, and geometric). However, as we will present in Section~\ref{sec:compNBCR}, the situation is somewhat more involved when considering reaction systems of more than one non-binary reaction.\\

It is important to note that while for arbitrary generating functions $F(\lambda;\vec{x})$ such as those mentioned in Theorem~\ref{thm:seminal} we find one-parameter \emph{groups} of transformations\footnote{Note that this statement refers to a statement on formal parameters, for which addition and subtraction are well-defined operations.}, specializing to generating functions that are PGFs we have to also comply with the requirement that all coefficients of the generating functions must be \emph{non-negative} real numbers. Upon closer inspection of the concrete formulae provided for  $P(\lambda;\vec{x})$, one can verify that $P(\lambda;\vec{x})$ with the formal parameter $\lambda$ set to a real number value will yield a proper PGF if and only if $\lambda$ is a non-negative real number. Therefore, truncating the parameter space to $\lambda\in \bR_{\geq0}$ entails that the one-parameter \emph{group} described by Theorem~\ref{thm:seminal} specializes to a \emph{one-parameter semi-group} (of real parameter $\lambda$): denote by $\cE(\lambda)$ the operator
\begin{equation}
	\cE(\lambda):=e^{\lambda \cH}\,.
\end{equation}
Then we have for all $\lambda,\mu\in \bR_{\geq0}$
\begin{equation}
\begin{aligned}
	&\cE(\lambda)P(\mu;\vec{x})\\
	&\;=\cE(\lambda)\cE(\mu)P(0;\vec{x})\\
	&\;=\cE(\lambda)g(\mu,\vec{x})P(0;\vec{T}(\mu,\vec{x}))\\
	&\;=g(\lambda,\vec{x})g(\mu,\vec{T}(\lambda,\vec{x}))P(0;\vec{T}(\mu;\vec{T}(\lambda;\vec{x})))\\
	&\overset{\eqref{eq:groupSeminalTG}}{=}
	g(\lambda+\mu,\vec{x})P(0;\vec{T}(\lambda+\mu;\vec{x}))\\
	&=\cE(\lambda+\mu)P(0;\vec{x})=
	P(\lambda+\mu;\vec{x})\,.
\end{aligned}
\end{equation}  
This finally allows us to make the following statement:
\begin{cor}
Consider an $N$-species chemical reaction system with a \emph{single} non-binary chemical reaction. Then the action of the \emph{evolution semi-group} $\cE(t)$ on the initial state's PGF $P(0;\vec{x})$ is completely described in terms of the action of the evolution semi-group at the level of probability generating functions, whence the PGF of $\ket{\Psi(t)}$ is given for all $t\in \bR_{\geq 0}$ by
\begin{equation}
	P(t;\vec{x})=\cE(t)P(0;\vec{x})=g(t;\vec{x})P(0;\vec{T}(t;\vec{x}))\,.
\end{equation}	
\end{cor}

\begin{sidewaystable*}[p]
\small
  \caption{\ Closed-form results for the time-dependent probability generating functions $P(t;\vec{x})$ for reaction systems of $N$ species with a single non-binary elementary reaction; here, $S_1,\dotsc,S_N$ denote the $N$ different species, while $\vec{\Delta}_i$ ($i\in \{1,\dotsc,N\}$) denotes the $N$-vector with coordinates $(\vec{\Delta}_i)_j=\delta_{i,j}$.}
  \label{tab:nonBinaryCRs}
  \renewcommand{\arraystretch}{1.9}
  \begin{tabular*}{\textwidth}{@{\extracolsep{\fill}}LLLC}
    \hline
    \text{reaction} & 
    \cH=\vec{q}(\vec{x})\cdot \partial_{\vec{x}}+v(\vec{x}) & 
    P(t;\vec{x})=g(t;\vec{x})P(0;\vec{T}(t;\vec{x}) & 
    \text{comments}\\
    \hline
    \emptyset\stackrel{\alpha}{\rightharpoonup} S_i &
    \alpha\left(\hat{x}_i-1\right) &
    Pois(\alpha t;x_i)\cdot P(0;\vec{x}) &
    Pois(\mu;x):=e^{\mu(x-1)}\\
    \emptyset\stackrel{\alpha}{\rightharpoonup} S_i+S_j &
    \alpha\left(\hat{x}_i\hat{x}_j-1\right) &
    \left(e^{\alpha t(x_ix_j-1)}\right)\cdot P(0;\vec{x}) &
    \text{(Poisson distribution, $0\leq \mu <\infty$)}
    \\
    \hline
     S_i\stackrel{\alpha}{\rightharpoonup} \emptyset &
    \alpha\left(1-\hat{x}_i\right)\frac{\partial}{\partial x_i} &
    P(0;\vec{x}+\left(-x_i+Bern(e^{-\alpha t};x_i)\right)\vec{\Delta}_i) &
    Bern(\mu;x):=(1-\mu)+x \mu\\
     S_i\stackrel{\alpha}{\rightharpoonup} S_j\quad (i\neq j) &
    \alpha\left(\hat{x}_j-\hat{x}_i\right)\frac{\partial}{\partial x_i} &
    P(0;\vec{x}+\left(-x_i+(x_j(1-e^{-\alpha t})+x_i e^{-\alpha t}\right)\vec{\Delta}_i) &
    \text{(Bernoulli distribution, $0\leq \mu\leq 1$)}
    \\
    \hline
    S_i\stackrel{\alpha}{\rightharpoonup} 2S_i&
    \alpha\left(\hat{x}_i^2-\hat{x}_i\right)\frac{\partial}{\partial x_i} &
    P(0;\vec{x}+\left(-x_i+Geom(e^{-\alpha t};x_i)\right)\vec{\Delta}_i) &
       Geom(\mu;x):=\frac{x \mu}{1-x(1-\mu)}\\
      S_i\stackrel{\alpha}{\rightharpoonup} S_i+S_j \quad (i\neq j)&
    \alpha\left(\hat{x}_i\hat{x}_j-\hat{x}_i\right)\frac{\partial}{\partial x_i} &
    P(0;\vec{x}+\left(-x_i+x_i Pois(\alpha t; x_j))\right)\vec{\Delta}_i) &
    \text{(Geometric distribution, $0< \mu\leq 1$)}\\
    S_i\stackrel{\alpha}{\rightharpoonup} S_j+S_k \quad (i\neq j\neq k)&
    \alpha\left(\hat{x}_j\hat{x}_k-\hat{x}_i\right)\frac{\partial}{\partial x_i} &
    P(0;\vec{x}+\left(-x_i+x_jx_k(1-e^{-\alpha t})+x_ie^{-\alpha t}\right)\vec{\Delta}_i) &\\
    \hline
    \end{tabular*}
 \end{sidewaystable*}

\subsection{Comments on elementary non-binary reactions}

In principle, for the cases covered in Table~\ref{tab:nonBinaryCRs}, one could also compute the exponential or the factorial moment generating functions $\cM(t;\vec{\lambda})$ and $\cF(t;\vec{nu})$, employing the techniques of Theorem~\ref{thm:seminal} to solve the  evolution equations~\eqref{eq:EvMultiM} and~\eqref{eq:FmultiEv}, respectively. However, since we have already computed the probability generating functions, we may alternatively make use of the well-known relationship
\begin{equation}
	\cM(t;\lambda)
=P(t;e^{\lambda})	
\end{equation}
as well as of the \emph{Stirling transform} (cf.\ Appendix~\ref{app:Stirling})
\begin{equation}
	\cF(t;\nu)=\cM(t;\ln(\nu+1))=P(t;\nu+1)\,.
\end{equation}
For practical applications, it is often also of interest to consider the \emph{cumulant generating function} $\cC(t;\lambda)$, defined as
\begin{equation}
	\cC(t;\lambda):=\ln(\cM(t;\lambda))=\ln(P(t;e^{\lambda}))\,.
\end{equation}

In order to illustrate the results as presented in Theorems~\ref{thm:seminal} and~\ref{thm:sem} graphically, one may extract from the probability generating functions $P(t;x)$ the time-dependent probabilities $p_n(t)$ according to
\begin{equation}
p_n(t)=\tfrac{1}{n!}\left[\tfrac{\partial^n}{\partial x^n}P(t;x)\right]\vert_{x\to 0}\,.
\end{equation}
We present in Figures~\ref{fig:0Ato1A}, \ref{fig:0Ato2A},
 \ref{fig:1Ato0A} and~\ref{fig:1Ato2A} the probabilities $p_n(t)$ for the four types of non-binary reactions as a function of $n$ and for a number of different times. The initial state was chosen\footnote{Note in particular that the probability distributions for systems consisting of an individual reaction $0A\xrightharpoonup{\;\gamma\:}2A$ or $2A\xrightharpoonup{\;\kappa\:}0A$, respectively, can evidently only evolve taking non-zero values on the even integer particle numbers when initialized on a pure state $\ket{\Psi(0)}=\ket{100}$.} for all four cases as $\ket{\Psi(0)}=\ket{100}$ (i.e,\ $P(0;x)=x^{100}$ or $p_n(0)=\delta_{n,100}$). The choices for the reaction rates $\alpha,\beta,\gamma,\tau$ was taken purely for aesthetic reasons, such as to keep the time evolutions for the chosen time steps within the window of presentation for all four types of reactions. Note that Figures~\ref{fig:2Ato0A} and~~\ref{fig:2Ato1A} present distributions for elementary binary reactions, the computation of which will be presented in Sections~\ref{sec:bCR}. Reaction systems consisting of more than one individual reaction will be treated in the next section, where we will provide an exact result on the probability generating functions for systems consisting of all four types of non-binary chemical reactions (cf.\ Theorem~\ref{thm:NBCRN}).

\begin{figure*}[h]
\centering
\subfloat[birth reaction $0A\xrightharpoonup{\;\beta=50\;}1A$\label{fig:0Ato1A}]{%
        \includegraphics[width=0.425\linewidth]{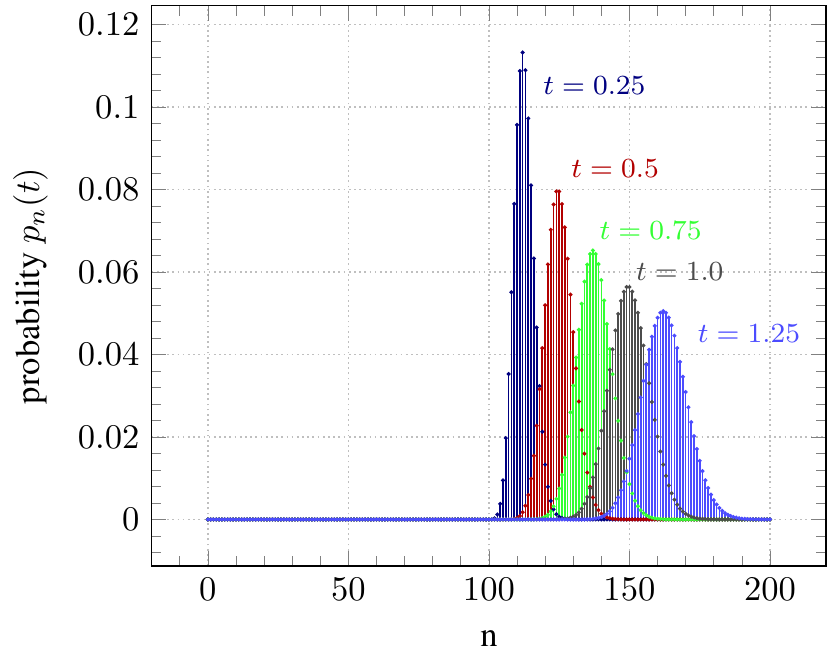}
               }
        \qquad\qquad
        \subfloat[pair creation reaction $0A\xrightharpoonup{\;\gamma=25\;}2A$\label{fig:0Ato2A}]{%
                \includegraphics[width=0.425\linewidth]{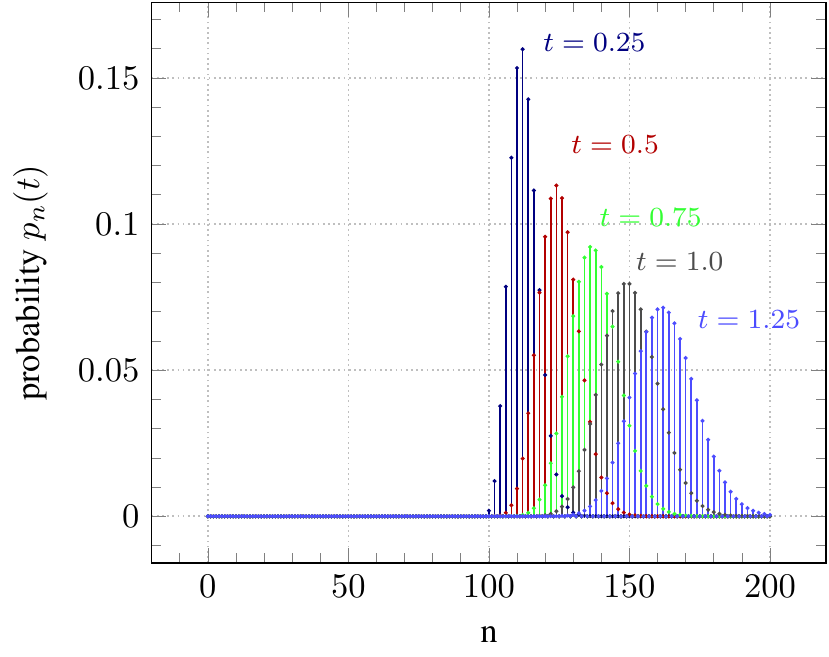}
                }
                
        \subfloat[decay reaction $1A\xrightharpoonup{\;\tau=4\;}0A$\label{fig:1Ato0A}]{%
                \includegraphics[width=0.425\linewidth]{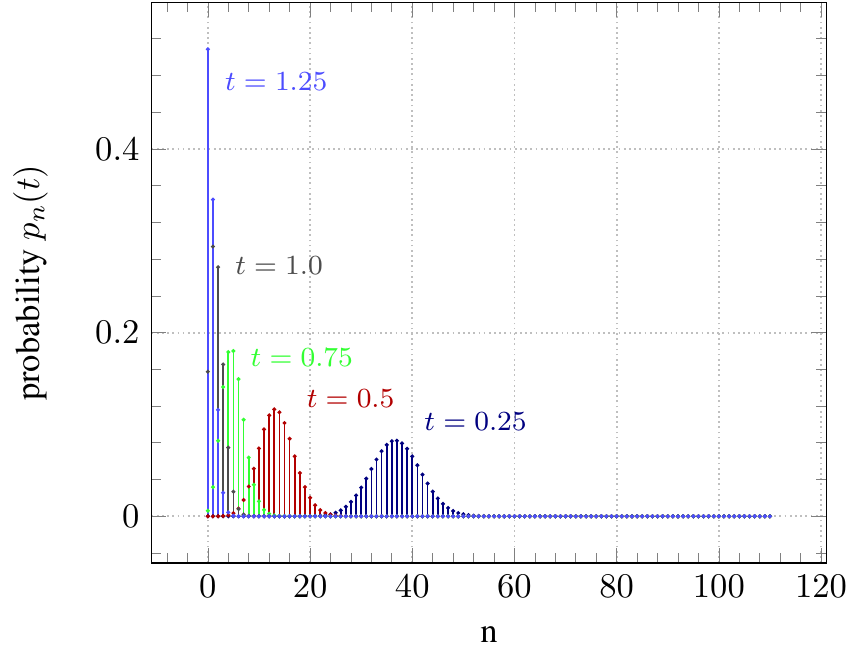}
                }
        \qquad\qquad
         \subfloat[autocatalysis reaction $1A\xrightharpoonup{\;\alpha=\tfrac{1}{2}\;}2A$\label{fig:1Ato2A}]{%
                \includegraphics[width=0.425\linewidth]{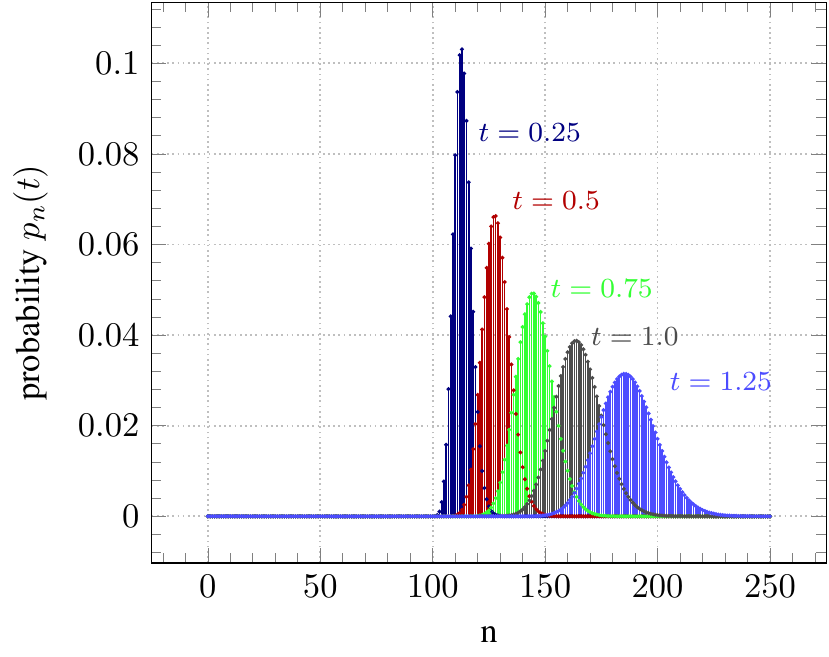}
                }
                
         \subfloat[pair annihilation reaction $2A\xrightharpoonup{\;\kappa=\tfrac{1}{40}\;}0A$\label{fig:2Ato0A}]{%
                \includegraphics[width=0.425\linewidth]{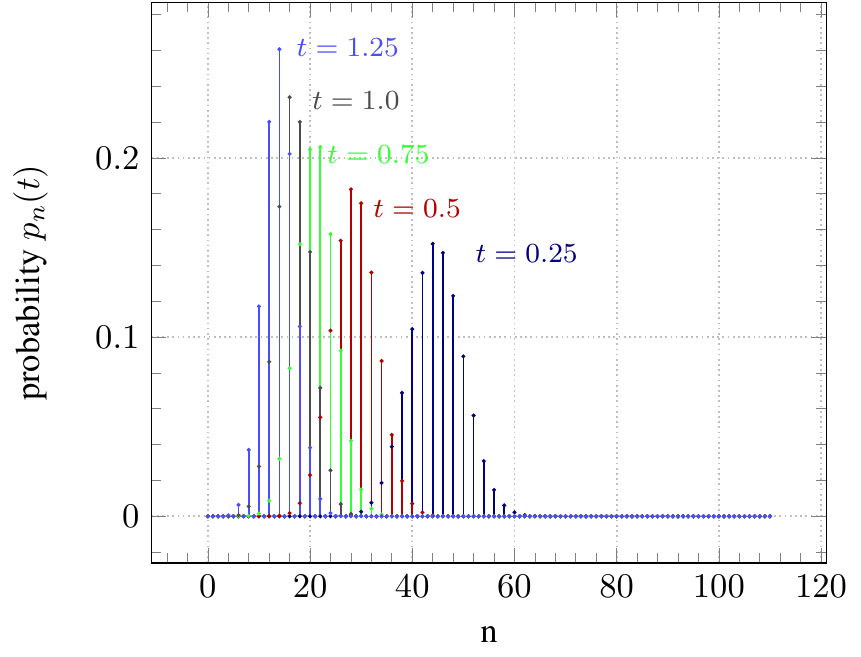}
                }
        \qquad\qquad
         \subfloat[catalytic decay reaction $2A\xrightharpoonup{\;\lambda=\tfrac{1}{10}\;}1A$\label{fig:2Ato1A}]{%
                \includegraphics[width=0.425\linewidth]{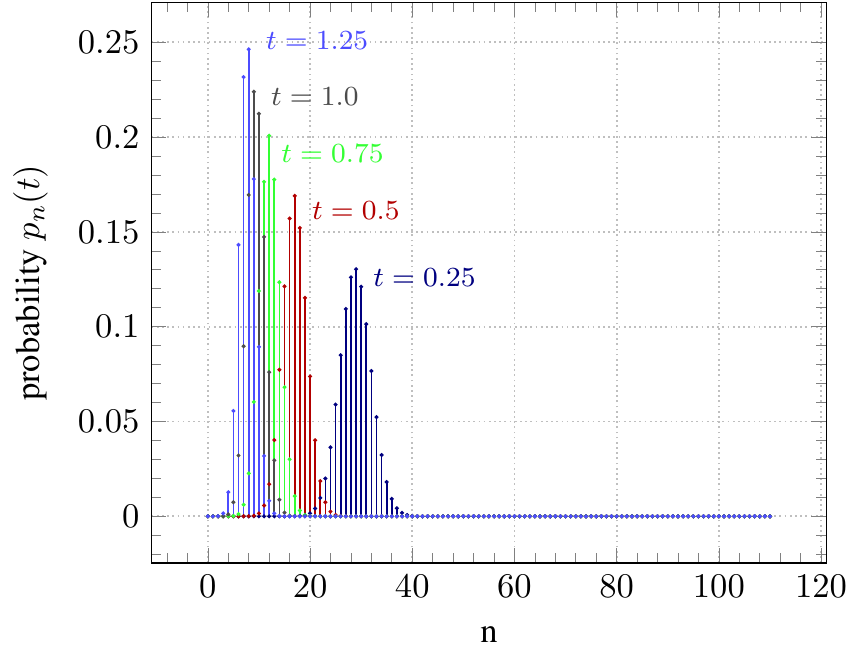}
                }
         \caption{Discrete probability distributions for initial state $\ket{\Psi(0)}=\ket{100}$ and for individual elementary reactions.}\label{fig:discrNB}
\end{figure*}

\subsection{On compositionality of elementary non-binary reactions}
\label{sec:compNBCR}
Even though the results presented in Table~\ref{tab:nonBinaryCRs} constitute a complete list of evolution formulae for \emph{individual} non-binary chemical reactions, this does not yet provide a full answer to the question of which form the evolution of reaction systems composed of \emph{multiple} non-binary chemical reactions would take. While the techniques of Theorem~\ref{thm:seminal} remain applicable for such a generic non-binary reaction system, one would a priori have to compute the functions $g(\lambda;\vec{x})$ and $\vec{T}(\lambda;\vec{x})$ anew for each possible such reaction system. It would thus be desirable to find a technique that allows to simplify this computation via providing a form of \emph{composition law} for solutions to chemical reaction systems assembled from elementary reactions. As a first step in this direction, we present a result for a reaction system composed of all four possible one-species non-binary elementary reactions:
\begin{thm}\label{thm:NBCRN}
Consider a $1$-species reaction system with reactions
\begin{align*}
\emptyset&\xrightharpoonup{\;\beta\;} S\,,\quad	
\emptyset\xrightharpoonup{\;\gamma\;} 2S\\
S &\xrightharpoonup{\;\tau\;} \emptyset\,,\quad
S\xrightharpoonup{\;\alpha\;} 2S\,,
\end{align*}%
with parameters $\alpha,\beta,\gamma,\tau \in \bR_{\geq 0}$. The reactions describe birth $(B)$, pair creation $(C)$, death $(D)$ and autocatalysis $(A)$ transitions. Then the probability generating function $P(t;x)$, given $P(0;x)$, reads
\begin{equation}
	\label{eq:allNBone}
	\begin{aligned}
		P(t;x)&=g_{BDA}(t;x)g_{CDA}(t;x)P(0;T_{DA}(t;x))\,,
	\end{aligned}
\end{equation}
where the concrete formulae differ slightly for the various possible choices of parameters. Let us for convenience introduce the notations (cf.\ Table~\ref{tab:nonBinaryCRs})
\begin{equation}
	\begin{aligned}
		T_D(t;x)&:=Bern(e^{-t\tau};x)\\
		T_A(t;x)&:=Geom(e^{-t\alpha};x)\,,
	\end{aligned}
\end{equation}
and the following formulae for some standard probability generating functions (i.e.\ the alternative geometric, the Poisson and the $2$-aerated Poisson distributions' PGFs)~\cite{gut2012probability,klenke2013probability,olver2010nist}:
\begin{equation}
	\begin{aligned}
	Geom_2(p;x)&:=\frac{p}{1-x(1-p)}&&(0<p\leq 1)\\
	Pois(p,x)&:=e^{p(x-1)} \quad&&(0\leq p)\\
	A_2Pois(p,x)&:=e^{p(x^2-1)} \quad&&(0\leq p)\,.	
	\end{aligned}
\end{equation}

\noindent\underline{Case $\alpha=\tau=0$:}
\begin{equation}
	\begin{aligned}
		T_{DA}^{\alpha=\tau=0}(t;x)&=x\\
		g_{BDA}^{\alpha=\tau=0}(t;x)&=
		Pois(\beta t;x)\\
		g^{\alpha=\tau=0}_{CDA}(t;x)&=A_2Pois(\gamma t;x)\,.
	\end{aligned}
\end{equation}

\noindent\underline{Case $\alpha=\tau\neq 0$:} 
\begin{equation}
\begin{aligned}
	T_{DA}^{\alpha=\tau>0}(t;x)&=
	T_D(s(t);T_A(s(t);x))\\
	g_{BDA}^{\alpha=\tau>0}(t;x)&=\left(Geom_2\left(e^{-\alpha s(t)};x\right)\right)^{\beta/\alpha}\\
	g_{CDA}^{\alpha=\tau>0}(t;x)&=\left(Geom_2\left(e^{-\alpha s(t)};x\right)\right)^{2\gamma/\alpha}\\
	&\qquad \cdot e^{\gamma t(x-1)^2Geom_2(e^{-\alpha s(t)};x)}\\
	s(t)&:=\frac{1}{\alpha}\ln(1+\alpha t)\,.
\end{aligned}
\end{equation} 

\noindent\underline{Case $\alpha\neq\tau$ and $\alpha>0,\tau\geq 0$:} 
\begin{equation}
\begin{aligned}
	T_{DA}^{\alpha\neq\tau}(t;x)&=
	T_D(s_D(t);T_A(s_A(t);x))\\
	g_{BDA}^{\alpha\neq\tau}(t;x)&=\left(Geom_2\left(e^{-\alpha s_A(t)};x\right)\right)^{\beta/\alpha}\\
	g_{CDA}^{\alpha\neq\tau}(t;x)&=\left(Geom_2\left(e^{-\alpha s_A(t)};x\right)\right)^{\gamma(\alpha+\tau)/\alpha^2}\\
	&\qquad \cdot e^{ f_{CDA}(t;x)Geom_2(e^{-\alpha s_A(t)};x)}\\
	f_{CDA}(t;x)&=
	\frac{\gamma}{\alpha+\tau}\bigg(
	(x-1)^2\left(e^{t(\alpha-\tau)}-1\right)\\
	&\qquad +(x-1)(1+\tfrac{\tau}{\alpha})\left(1-e^{t(\alpha-\tau)}\right)
	\bigg)
	\\
	\tau s_D(t)&:=\ln\left(
	\frac{\alpha-\tau e^{(\tau-\alpha)t}}{\alpha-\tau}
	\right)\\
	\alpha s_A(t)&:=\ln\left(
	\frac{\alpha e^{(\alpha-\tau)t}-\tau }{\alpha-\tau}
	\right)\,.
\end{aligned}
\end{equation}

\noindent\underline{Case $\alpha=0$ and $\tau> 0$:}
\begin{equation}\label{eq:BCTsystem}
	\begin{aligned}
		T_{DA}^{\alpha=0,\tau>0}(t;x)&=T_D(t;x)\\
		g_{BDA}^{\alpha=0,\tau>0}(t;x)&=Pois(\tfrac{\beta}{\tau}\left(1-e^{-t\tau});x\right)\\
		g_{CDA}^{\alpha=0,\tau>0}(t;x)&=Pois(\tfrac{\gamma}{\tau}\left(1-e^{-t\tau})^2;x\right)\cdot\\
		&\;\cdot A_2Pois(\tfrac{\gamma}{2\tau}\left(1-e^{-2t\tau});x\right)\,.
	\end{aligned}
\end{equation}
\begin{proof}
Via applying Theorem~\ref{thm:seminal}, one obtains formulae for $P(t;x)$ in the various cases of interest that may then be brought into the forms as presented by an elementary, if somewhat tedious ``factorization'' of the formulae into compounds and convolutions of probability generating functions. See Appendix~\ref{app:comp} for further details.
\end{proof}
\end{thm}

It might be worthwhile to consider the following auxiliary formula, which expresses the fact that taking a power of the $Geom_2(p;x)$ PGF with a positive real exponent $\mu$ results in a well-posed PGF\footnote{For the case $0<\mu<1$, we find
\[
\binom{k+\mu-1}{k}=\frac{\Gamma(\mu+k)}{k!\Gamma(\mu)}=\frac{1}{k!}(\mu)^{(k)}\,,
\]
with $(x)^{(k)}$ the \emph{rising factorial}, i.e.\ $(x)^{(0)}=1$ and $x^{(k)}=x(x+1)\dotsc (x+k-1)$ for $k>0$.}:
\begin{equation}
\begin{aligned}
	&\left(Geom_2(p;x)\right)^{\mu}=\left(\frac{p}{1-x(1-p)}\right)^{\mu}\\
	&\quad=\delta_{\mu,0}+\delta_{\mu>0}p^{\mu}\sum_{n\geq 0}\binom{n+\mu-1}{n}(1-p)^n x^n\,.
\end{aligned}
\end{equation}
For the permissible paramter range of the $Geom_2(p;x)$ ($0<p\leq 1$), we thus observe that all coefficients of $x^n$ for all $n\geq 0$ are non-negative real numbers, and evidently $Geom_2(p;1)^{\mu}=1$.\\

Since admittedly the results of Theorem~\ref{thm:NBCRN} are somewhat hard to interpret from the presentation in terms of generating functions, we exhibit in Figure~\ref{fig:NBCRN1} a set of illustrative examples of parameter choices and their respective effects on the first three cumulants and on the probability distributions. Moreover, in Figure~\ref{fig:NBCRN2} we provide ternary data plots for the effects of relative parameter choices for reaction systems of three of the four possible semi-linear reactions for illustration.\\

A particularly interesting effect concerns the example of the reaction system
\begin{equation}\label{eq:ACT}
1A\xrightharpoonup{\;\alpha=1/3\;}2A\,,\quad 0A\xrightharpoonup{\;\gamma=1/3\;}2A\,,\quad 1A\xrightharpoonup{\;\tau=1/3\;}0A\,,
\end{equation}
for which we depict the evolution of the first three cumulants in Figure~\ref{fig:ACTa} as well as the evolution of the probability distribution in Figure~\ref{fig:ACTb}. Notably, the mean (i.e.\ the first cumulant) remains fixed at the initial value (in this case $c_1(0)=100$), and throughout the evolution the probability distribution effectively undergoes a progressive broadening of the distribution around this mean value.\\

In Figures~\ref{fig:BCTc} and~\ref{fig:BCTd}, we depict cumulant and probability distribution evolutions for a variant of a birth-death system, namely a system with also a pair creation reaction,
\begin{equation}\label{eq:BCT}
0A\xrightharpoonup{\;\beta=1/5\;}1A\,,\quad
0A\xrightharpoonup{\;\gamma=3/5\;}2A\,,\quad 1A\xrightharpoonup{\;\tau=1/5\;}0A\,.
\end{equation}
It is instructive to consider the exact formula for the probability generating function of this system (cf.\ \eqref{eq:BCTsystem}), which for an initial state $\ket{\Psi(0)}=\ket{M}$ of precisely $M$ particles (i.e.\ for $P(0;x)=x^M$) reads
\begin{equation}
	\begin{aligned}
	P(t;x)&=e^{\tfrac{\gamma}{2\tau}\left(1-e^{-2t\tau}\right)(x^2-1)}\\
		&\quad \cdot
		e^{\left[
		\tfrac{\gamma}{\tau}\left(1-e^{-t\tau}\right)^2
		+\tfrac{\beta}{\tau}\left(1-e^{-t\tau}\right)\right](x-1)}\\
		&\qquad\cdot(Bern(e^{-t\tau};x))^M\,.
	\end{aligned}
\end{equation}
Taking the limit $t\to\infty$ and noting that 
\[
\lim\limits_{t\to\infty}Bern(e^{-t\tau};x)=1\,,
\]
we find the limit distribution
\begin{equation}
\lim\limits_{t\to\infty}P(t;x)=
e^{\tfrac{\gamma}{2\tau}(x^2-1)+\tfrac{\beta+\gamma}{\tau}(x-1)}\,.
\end{equation}
We present in Figure~\ref{fig:NBCRN2} a number of ternary data plots the evolution of first and second cumulants for this reaction system for illustration.\\

As a final example, we depict in Figures~\ref{fig:ABCTe} and~\ref{fig:ABCTf} the evolution of a reaction system \emph{with all four types} of non-binary one-species reactions. Depending on the choices of the four reaction rates involved, we find that the system evolves in a fashion either similar to the reaction system~\eqref{eq:ACT} or to the reaction system~\eqref{eq:BCT}. A representative example for a choice of parameters that leads to dynamical behaviour akin to that of the reaction system~\eqref{eq:BCT} is depicted in Figures~\ref{fig:ABCTe} and~\ref{fig:ABCTf}.\\

To conclude our treatment of non-binary chemical reaction systems, it is worthwhile noting that there exists previous work in the mathematical chemistry literature~\cite{jahnke2007solving} on closed-form exact solutions for probability generating functions of multi-species non-binary reactions except for the autocatalytic reaction types. Since our formulae cover an arbitrary single-species non-binary reaction system, it appears worthwhile to consider the generalization of our Theorem~\ref{thm:NBCRN} to the multi-species cases. We plan to present such results in future work.

\begin{figure*}[h]
\centering
\subfloat[$1A\xrightharpoonup{\;\alpha=1/3\;}2A$, $0A\xrightharpoonup{\;\gamma=1/3\;}2A$, $1A\xrightharpoonup{\;\tau=1/3\;}0A$\label{fig:ACTa}]{%
        \includegraphics[width=0.425\linewidth]{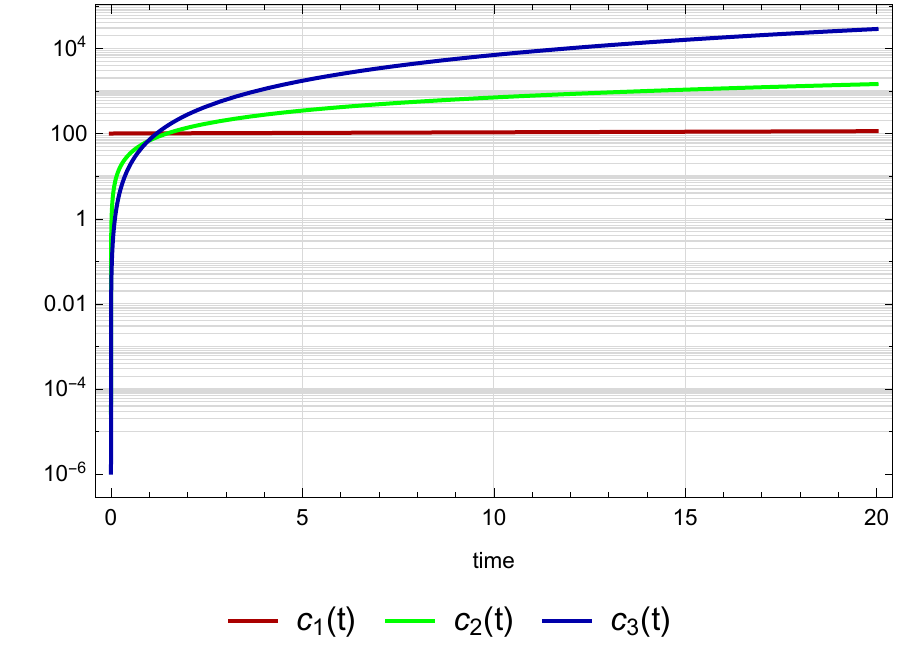}
               }
        \qquad\qquad
        \subfloat[distributions $\ket{\Psi(t)}$ for \textbf{a)}\label{fig:ACTb}]{%
                \includegraphics[width=0.425\linewidth]{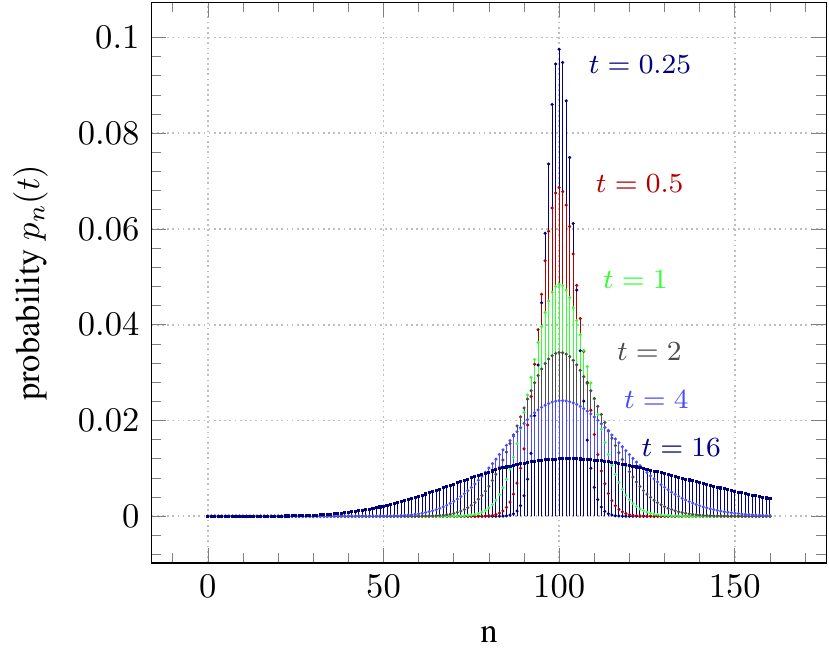}
                }
                
\subfloat[%
$0A\xrightharpoonup{\;\beta=0.2\;}1A$, $0A\xrightharpoonup{\;\gamma=0.6\;}2A$, $1A\xrightharpoonup{\;\tau=0.2\;}0A$
\label{fig:BCTc}]{%
        \includegraphics[width=0.425\linewidth]{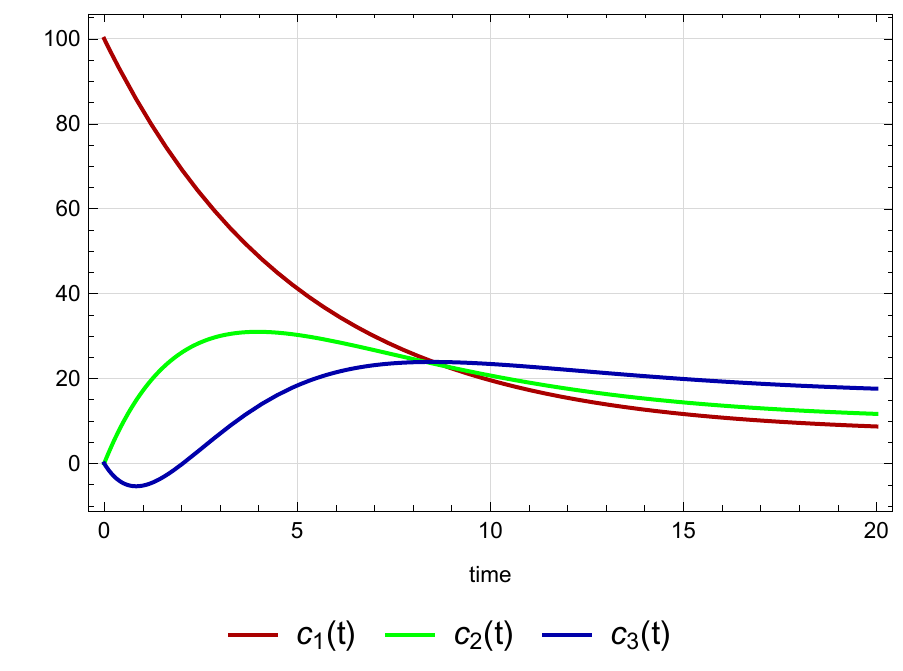}
               }
        \qquad\qquad
        \subfloat[distributions $\ket{\Psi(t)}$ for \textbf{c)}\label{fig:BCTd}]{%
                \includegraphics[width=0.425\linewidth]{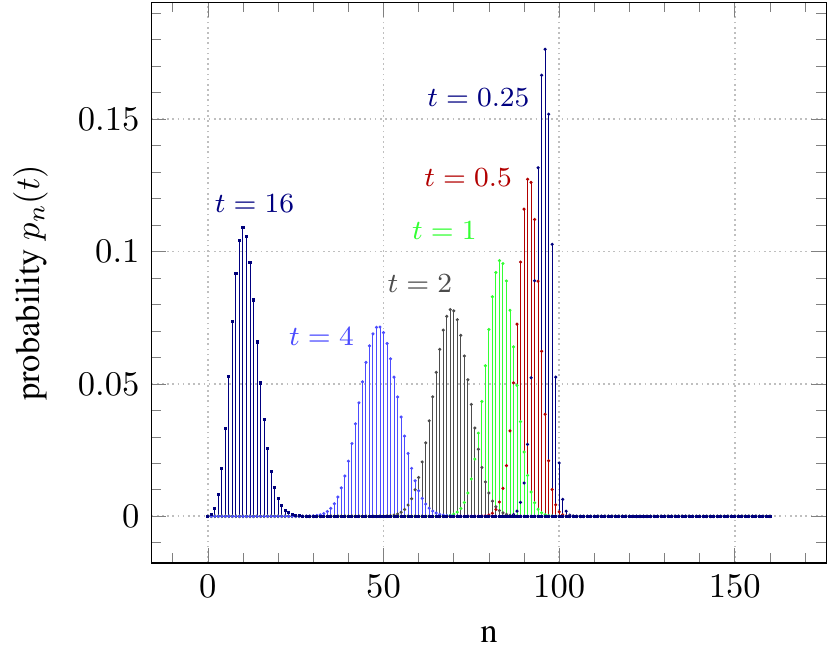}
                }
                
 \subfloat[%
 $0A\xrightharpoonup{0.4\;}1A$, $0A\xrightharpoonup{0.2\;}2A$, $1A\xrightharpoonup{0.3\;}0A$, $1A\xrightharpoonup{0.1\;}2A$
 \label{fig:ABCTe}]{%
        \includegraphics[width=0.425\linewidth]{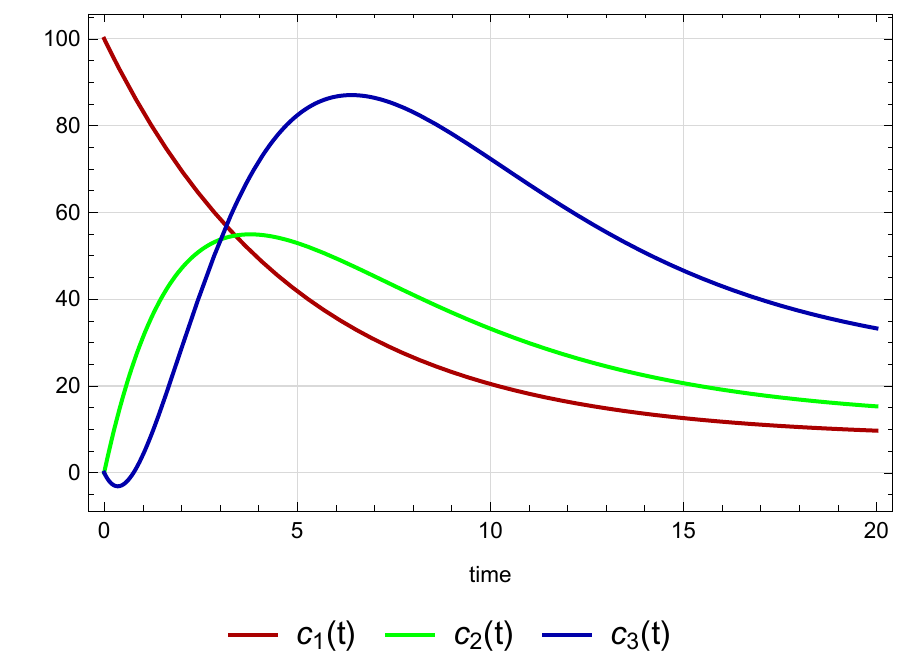}
               }
        \qquad\qquad
        \subfloat[distributions $\ket{\Psi(t)}$ for \textbf{e)}\label{fig:ABCTf}]{%
                \includegraphics[width=0.425\linewidth]{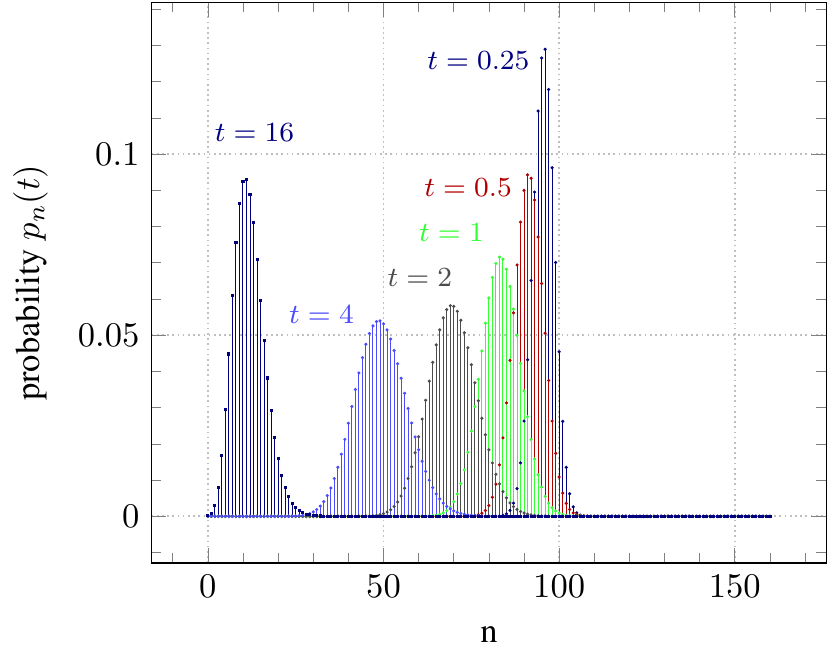}
                }
         \caption{First three cumulants $c_i(t)$ ($i=1,2,3$) and discrete probability distributions $\ket{\Psi(t)}=\sum_{n\geq 0} p_n(t)\ket{n}$ for systems of non-binary reactions with initial state $\ket{\Psi(0)}=\ket{100}$.}\label{fig:NBCRN1}
\end{figure*}

\begin{figure*}[h]
\centering
\subfloat[Mean number of particles at time $t=1$\label{fig:ternBDPc1-t1}]{%
        \includegraphics[width=0.425\linewidth]{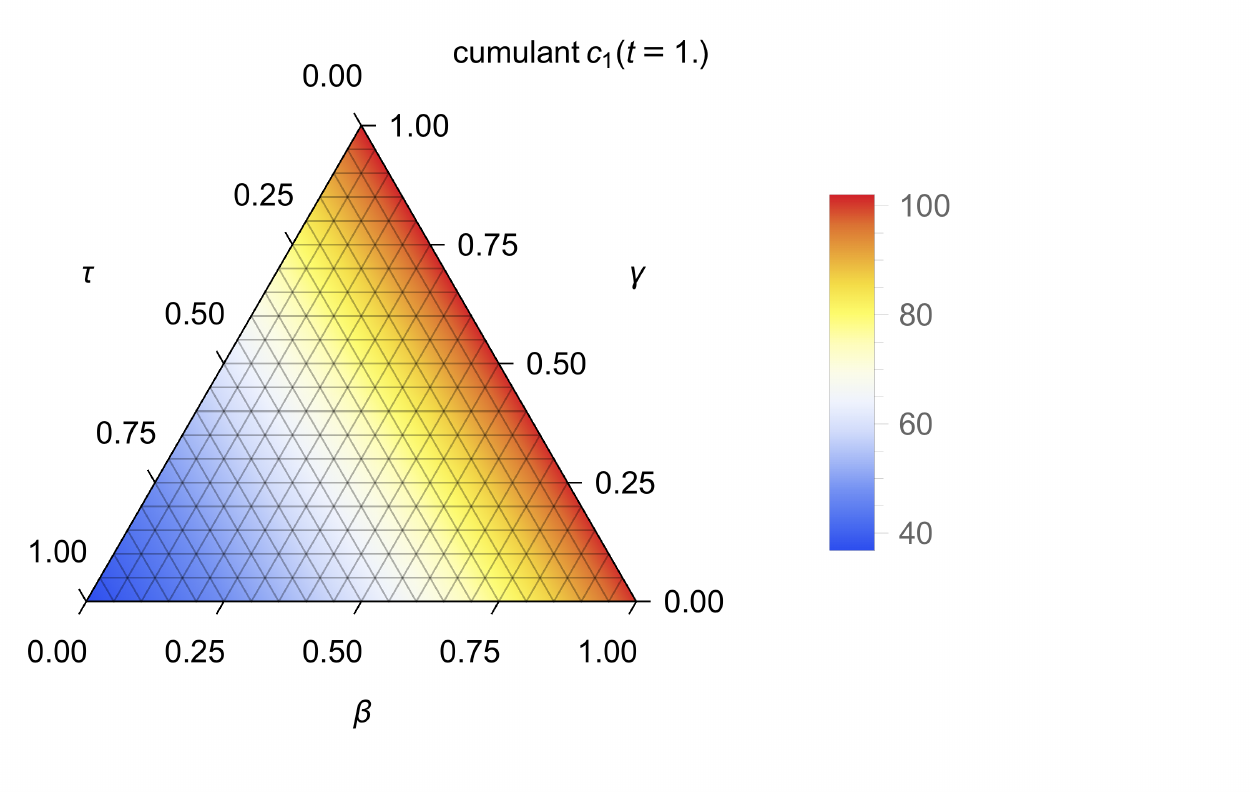}
               }
        \qquad\qquad
        \subfloat[Variance of number of particles at time $t=1$\label{fig:ternBDPc2-t1}]{%
        \includegraphics[width=0.425\linewidth]{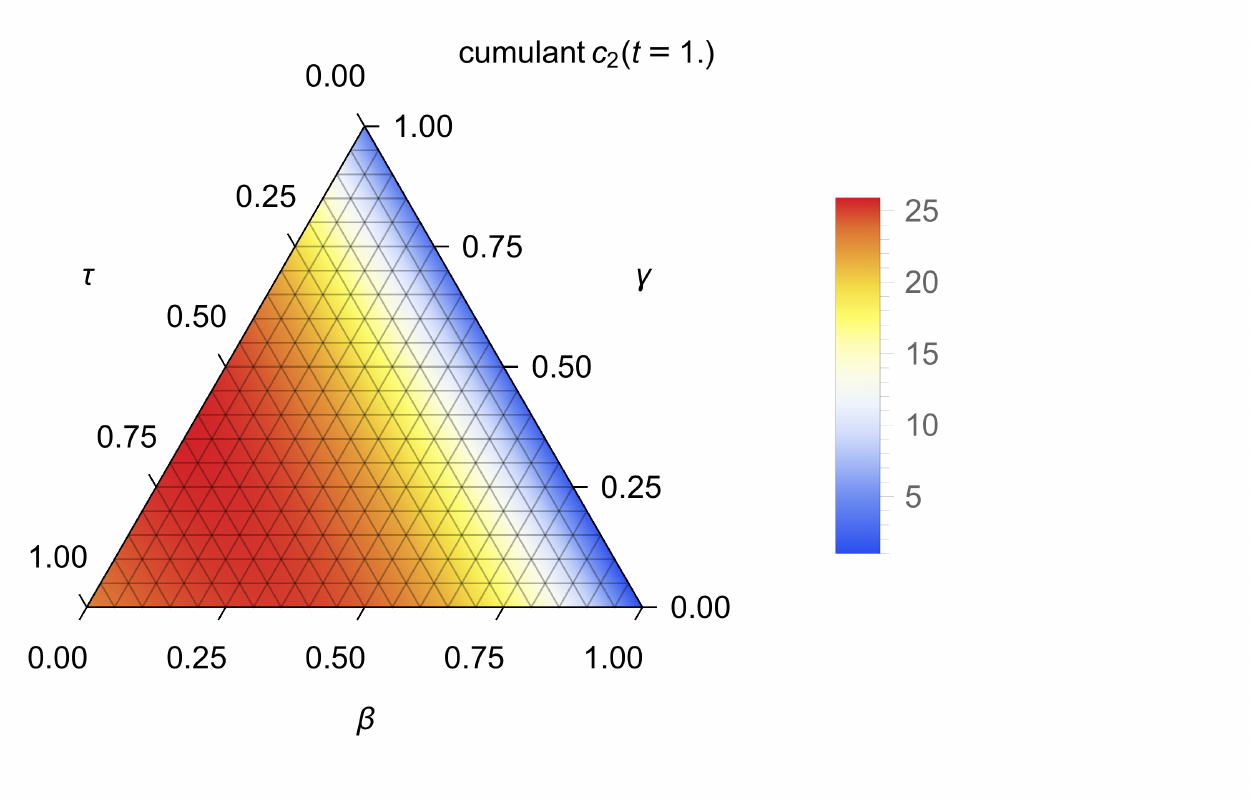}
               }
                
\subfloat[Mean number of particles at time $t=4$\label{fig:ternBDPc1-t4}]{%
        \includegraphics[width=0.425\linewidth]{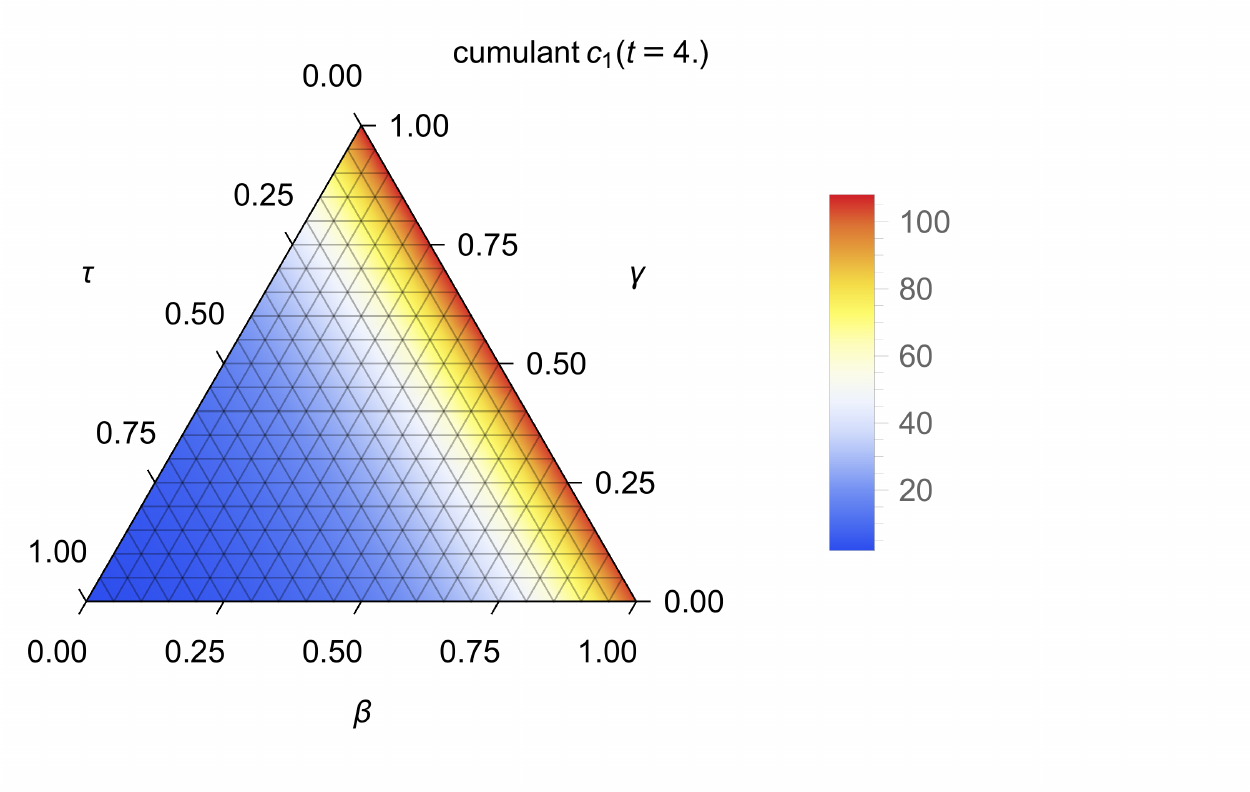}
               }
        \qquad\qquad
        \subfloat[Variance of number of particles at time $t=4$\label{fig:ternBDPc2-t4}]{%
        \includegraphics[width=0.425\linewidth]{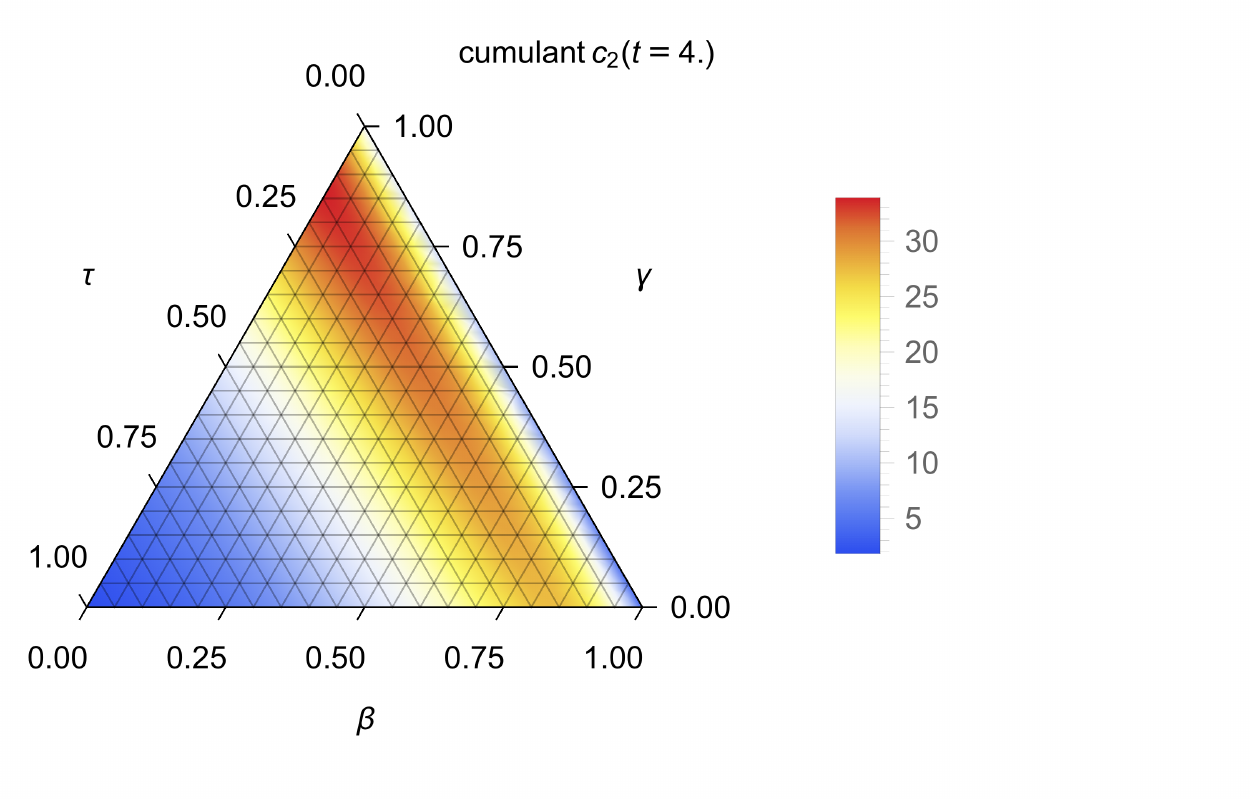}
               }
                
 \subfloat[Mean number of particles at time $t=16$\label{fig:ternBDPc1-t1}]{%
        \includegraphics[width=0.425\linewidth]{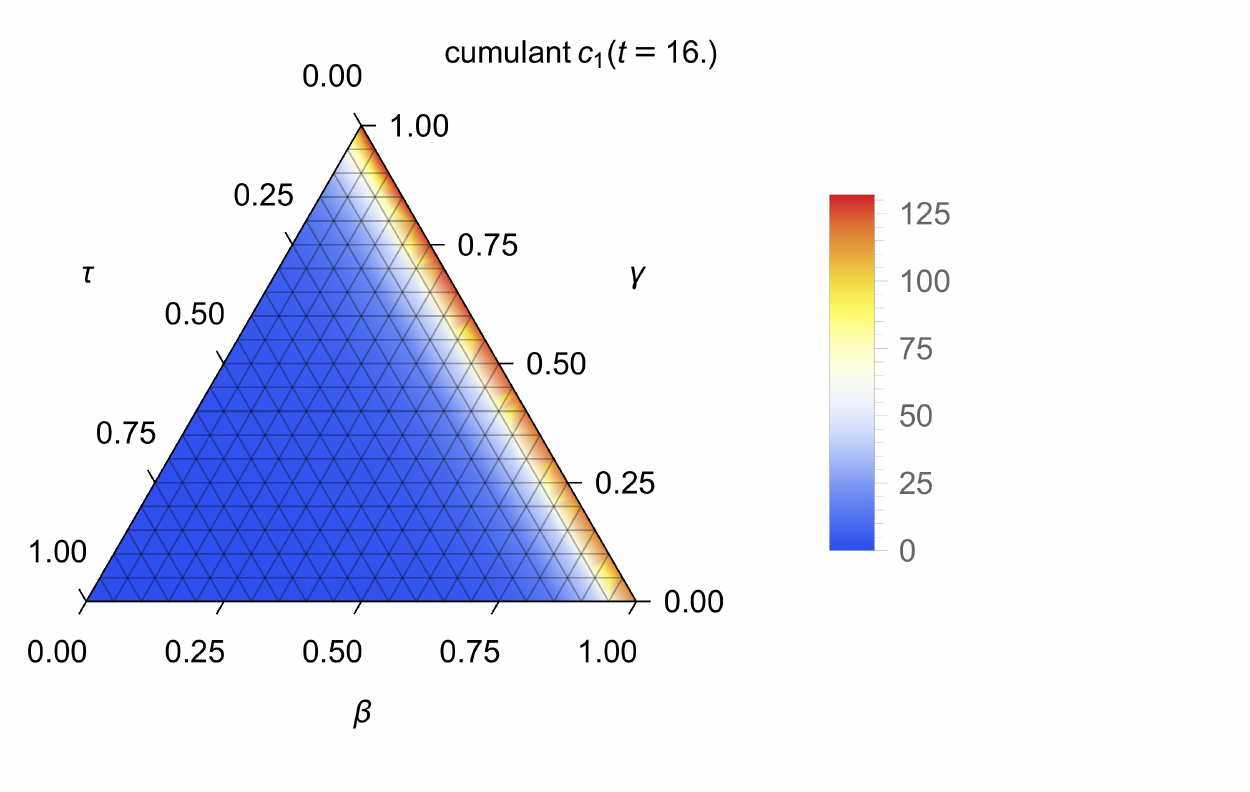}
               }
        \qquad\qquad
        \subfloat[Variance of number of particles at time $t=16$\label{fig:ternBDPc2-t16}]{%
        \includegraphics[width=0.425\linewidth]{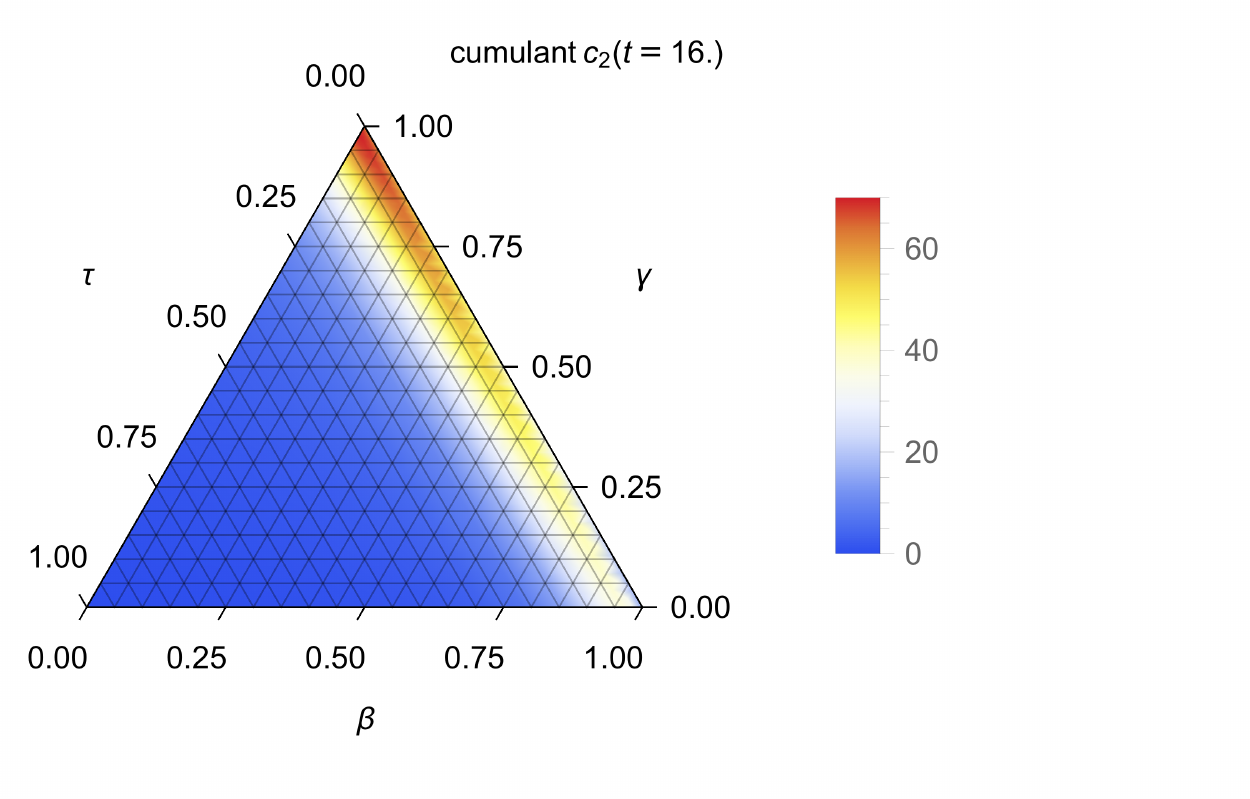}
               }
         \caption{Ternary parameter dependence plot ($\beta+\gamma+\tau=1$) of the first and second cumulants (mean $c_1(t)$ and variance $c_2(t)$) for a reaction system composed of birth, pair creation and decay reactions, with initial state $\ket{\Psi(0)}=\ket{100}$. In $a)$, $c)$ and $e)$, we observe that with increasing decay rate $\tau$ the mean particle number decreases over time, while the relative strengths of $\beta$ and $\gamma$ have little influence on the dynamics. In contrast, the time evolution of the variance  as depicted in $b)$, $d)$ and $f)$ is more sensitive to the ratio of $\beta$ and $\gamma$, as are all higher cumulants (not presented).}\label{fig:NBCRN2}
\end{figure*}

\subsection{Closed-form solutions for probability generating functions of single-species binary reactions}\label{sec:bCR}

In the 1960's, McQuarrie~\cite{mcquarrie1964kinetics,mcquarrie1967stochastic} suggested to solve the problem of finding the probability generating functions for single-species binary chemical reaction systems and for a deterministic initial state $\ket{\Psi(0)}=\ket{M}$, whence for $P(0;x)=x^M$, via a separation of variables Ansatz of the form
\begin{equation}\label{eq:MQA}
	P(t;x)=\sum_{n\geq0} A_{M;n}T_n(t)X_n(x)\,.
\end{equation}
Here, the functions $X_n(x)$ are \emph{polynomials}, with $degree(X_n)=n$. The time-dependent functions $T_n(t)$ read
\begin{equation}
	T_n(t)=e^{\lambda_n t}\,,
\end{equation}
while the constant coefficients $A_{M;n}\in \bR$ are determined by the initial condition $P(0;x)=x^M$,
\begin{equation}
	x^M=\sum_{n\geq0} A_{M;n} X_n(x)
	=\sum_{n=0}^{M} A_{M;n} X_n(x)\,.
\end{equation}
In order for McQuarrie's Ansatz to be consistent, it is thus strictly necessary that the polynomials $X_n(x)$ form a \emph{complete system of eigenfunctions} of the Hamiltonian $\cH\equiv \cH(\hat{x},\partial_x)$ in the Bargmann-Fock basis (compare~\eqref{eq:HmultiDef} for the one-species case upon invoking the one-to-one correspondence described in~\eqref{eq:HWvar}):
\begin{equation}
	\cH X_n(x)=\lambda_n X_n(x)\,.
\end{equation}
In other words, it must be the case that precisely one degree $n$ polynomial $X_n(x)$ is available for each non-negative integer value $n$.\\

McQuarrie suggests the following approach: the functions $X_n(x)$ should be found as certain \emph{orthogonal polynomials of the Jacobi type}. Let us briefly review some of the necessary background material on this mathematical structure.

\subsection{Classical Jacobi polynomials and orthogonality}
\label{sec:cJac}

The Jacobi-type second-order differential operator\footnote{In this section, for notational coherence with the cited standard reference we temporarily denote the variables by $z$ and drop the notation $\hat{z}$ in favor of just $z$ for the linear operator of multiplication by $z$.} (cf.\ the standard reference book~\cite{olver2010nist}, chapter 18, p.~445, Table~18.8.1)
\begin{equation}
	\label{eq:DEjac}
	\begin{aligned}
	D^{\alpha,\beta}_{Jac}&:=(1-z^2)\tfrac{\partial^2}{\partial z^2}
	+q^{(\alpha,\beta)}(z)\tfrac{\partial}{\partial z}\\
	q^{(\alpha,\beta)}(z)&=(\beta-\alpha-(\alpha+\beta+2)z)
	\end{aligned}
\end{equation} 
with \emph{real parameters}
\begin{equation}
	\label{eq:Jparam}
	\alpha,\beta>-1
\end{equation}
is known to possess a \emph{system of orthogonal polynomials} $\{P^{(\alpha,\beta)}_n(z)\}_{n\geq 0}$, the so-called \emph{classical Jacobi polynomials}, as its complete basis of eigenfunctions. More explicitly (cf.\ e.g.\ \cite{olver2010nist}, Equation~18.5.8),
\begin{equation}
	\begin{aligned}\label{eq:JacExpl}
		&P^{(\alpha,\beta)}_n(z)\\
		&\quad =2^{-n}\sum_{\ell=0}^n\binom{n+\alpha}{\ell}\binom{n+\beta}{n-\ell}(z-1)^{n-\ell}(z+1)^{\ell}\,,
	\end{aligned}
\end{equation} 
satisfying the eigenequation
\begin{equation}
	D^{\alpha,\beta}_{Jac}P_n^{(\alpha,\beta)}(z)=-n(n+\alpha +\beta+1)P_n^{(\alpha,\beta)}(z)\,.
\end{equation}
The orthogonality property is found by defining for each admissible choice of parameters $\alpha,\beta>-1$ an \emph{inner product} $\Phi_{\alpha,\beta}$ on the space of polynomials,
\begin{equation}
	\begin{aligned}
		\label{eq:innerJclass}
		\Phi_{\alpha,\beta}(p(z),q(z))&:=\int_{-1}^{+1}dz\; w_{\alpha,\beta}(z)p(z)q(z)\\
		w_{\alpha,\beta}(z)&:=(1-z)^{\alpha}(1+z)^{\beta}\,.
	\end{aligned}
\end{equation}
It is one of the classical results of the theory of orthogonal polynomials that for all $n\geq 0$
\begin{equation}
	P_{n}^{(\alpha,\beta)}(z)\in L^2_{\alpha,\beta}([-1,1])\,.
\end{equation}
Here, $L^2_{\alpha,\beta}([-1,1])$ denotes the space of functions on the interval $[-1,1]$ which are \emph{square-integrable} w.r.t.\ the integral against the weight function $w_{\alpha,\beta}(z)$. Moreover, orthogonality manifests itself as
\begin{equation}
	\Phi_{\alpha,\beta}\left(P_m^{(\alpha,\beta)}(z),P_n^{(\alpha,\beta)}(z)\right)=\delta_{m,n}\phi_{\alpha,\beta}(n)\,,
\end{equation}
with $\phi_{\alpha,\beta}(n)\in \bR_{>0}$ some non-zero real numbers. 

\subsection{The case of binary reactions and a technical problem}

As we shall see momentarily, the type of differential operators we are interested in when studying chemical reaction systems, while superficially of the structure of the Jacobi-type differential operators, is precisely \emph{not} fulfilling the parameter constraint~\eqref{eq:Jparam}. The problem presents itself as follows:
\begin{lem}\label{lem:CRbin}
Consider a chemical reaction system of one species of particles and the reactions
\begin{equation}
	\begin{aligned}
	&A\xrightharpoonup{\;r_d\;}\emptyset\,,\quad2\:A\xrightharpoonup{\;r_k\;}\emptyset\,,\quad
	2\:A\xrightharpoonup{\;r_{\ell}\;}A\,,
	\end{aligned}
\end{equation}
where $r_k+r_{\ell}>0$ (i.e.\ the system involves at least one binary reaction), and where we assume for simplicity and without loss of generality\footnote{Note that generic values of $r_k$ and $r_{\ell}$ may be achieved via appropriately rescaling $t$. See Proposition~\ref{prop:binGen} for the precise details.}
\[
	r_k+r_{\ell}=1\,.
\]
The corresponding Hamiltonian reads in the Bargmann-Fock basis
\begin{equation}
\begin{aligned}
	\cH&=r_d(1-\hat{x})\tfrac{\partial}{\partial x}\\
	&\quad 
		r_k(1-\hat{x}^2)\left(\tfrac{\partial}{\partial x}\right)^2
		+r_{\ell}(\hat{x}-\hat{x}^2)\left(\tfrac{\partial}{\partial x}\right)^2\,.
\end{aligned}
\end{equation}
Then one can bring the Jacobi-type differential operator $D_{Jac}^{\alpha,\beta}$ as presented in~\eqref{eq:DEjac} via the affine transformation
\begin{equation}
	x(z)=\frac{r_{\ell}}{2}+ \left(1-\frac{r_{\ell}}{2}\right)z
\end{equation}
into the form of $\cH$, which entails that
\begin{equation}\label{eq:binCRans}
		\left(\cH +n(n+r_d-1)\right)P^{(-1;r_d-1)}_n\left(\frac{x-\frac{r_{\ell}}{2}}{1-\frac{r_{\ell}}{2}}\right)=0\,.
\end{equation}	
\end{lem}

Unfortunately, we thus encounter for \emph{all} of the interesting reaction rate choices yielding binary chemical reaction systems the ``illegal'' parameter ranges, since we are interested specifically in the cases
\begin{equation}
\begin{aligned}
		\text{(i)}\quad &\alpha=-1, \;\beta=r_d-1>-1\\
		\text{(ii)}\quad &\alpha=-1,\;\beta=-1\quad (\text{i.e.\ } r_d=0)\,.
\end{aligned}
\end{equation}
We encounter the following technical problems for these two types of parameter choices:
\begin{itemize}
\item[$(i)$] \emph{Case $\alpha=-1,\beta>-1$:} With respect to the inner product $\Phi_{-1,\beta}$ as defined in~\eqref{eq:innerJclass} (with the variable named $x$ from hereon), we find that the degree zero polynomial
\[
	P^{(-1,\beta)}_0(x)=1
\]
	is \emph{not} square-integrable,
\[
	P^{(-1,\beta)}_0(x)\not\in L^2_{-1,\beta}([-1,1])\,.
\]
Therefore, the set $\{P^{(-1,\beta)}_n(x)\}_{n\geq 0}$ fails to constitute a complete set of eigenfunctions of the operator $\cH$.\\

\item[$(ii)$] \emph{Case $\alpha=-1,\beta=-1$:} In addition to the failure of $P^{(-1,-1)}_0(x)$ to be square-integrable w.r.t.\ the inner product $\Phi_{-1,-1}$, we moreover find from the explicit formula~\eqref{eq:JacExpl} for the Jacobi polynomials that
\begin{equation}
	P^{(-1,-1)}_1(x)=0\,.
\end{equation}
In other words, in this case we lack the degree zero \emph{and} the degree one eigenfunctions, and thus again do not obtain a complete system of eigenfunctions.
\end{itemize}
As it stands, one would thus have to conclude that McQuarrie's approach to binary chemical reaction systems is ill-conceived.\\

A resolution of the problem of incompleteness of the eigenbasis has been provided by Kwon and Littlejohn~\cite{kwon1994characterizations,kwon1997classification,kwon1996new,kwon1998sobolev} (see also their survey article~\cite{everitt2001orthogonal}). They observe first that any constant function and any linear function are solutions to the Jacobi type differential equation, and that moreover the family $\{P^{(\alpha,\beta)}_n(x)\}_{n\geq 2}$ even at the problematic parameter values $\alpha=-1$ and $\beta\geq-1$ remains an orthogonal family of  eigenfunctions. Their key insight is then to consider an \emph{alternative notion of orthogonality}, so-called \emph{Sobolev-orthogonality}, which admits to construct a family $\{\tilde{P}^{(\alpha,\beta)}_n\}_{n\geq 0}$ that is a complete system of eigenfunctions and orthogonal w.r.t.\ a Sobolev-type inner product.\\

We quote from~\cite{kwon1998sobolev} the relevant material (cf.\ Propositions 4.2 and 4.3 of loc cit):
\begin{prop}\label{prop:KL}
Let the \emph{Sobolev inner products}   $\Phi^{(-1,-1)}_{A,B}$ and $\Phi_C^{(-1,\beta)}$ (for $\beta>-1$) be defined as follows (for $p,q\in \bR[x]$ some polynomials with real coefficients)
\begin{equation}
\begin{aligned}
	\Phi^{(-1,-1)}_{A,B}(p,q)&:=A p(1)q(1)+B p(-1)q(-1)\\
	&\quad +\int_{-1}^{+1}dx\; p'(x)q'(x)\\
	\Phi^{(-1,\beta)}_C(p,q)&:= C p(1)q(1)\quad\qquad (\beta>-1)\\
	&\quad+\int_{-1}^{+1}dx\; (x+1)^{\beta+1}p'(x)q'(x)\,.
\end{aligned}
\end{equation}
	Here, $p'(x)$ and $q'(x)$ denote the first derivatives of the polynomials, and $A,B,C\in \bR$ are some parameters, which have to satisfy the following conditions: for $\Phi^{(-1,-1)}_{A,B}$, $A$ and $B$ must verify
\begin{equation}
	\begin{aligned}
		&A+B>0\,,\quad A(\gamma+1)^2+B(\gamma-1)^2+2\neq 0\\
		&\quad A(\gamma+1)+B(\gamma-1)=0\\
		&\quad \gamma:=(B-A)/(A+B)\,,
	\end{aligned}
	\end{equation}	
	while for $\Phi_C^{(-1,\beta)}$ the parameter $C$ must verify
	\begin{equation}
		C>0\,.
	\end{equation}
Then the \emph{Sobolev-Jacobi orthogonal polynomials} $\tilde{P}^{(\alpha,\beta)}_n(x)$ are defined for $\alpha=-1,\beta=-1$ as
\begin{equation}\label{eq:SJmm}
\begin{aligned}
	&\tilde{P}^{(-1,-1)}_0(x):=1\\
	&\tilde{P}^{(-1,-1)}_1(x):=x+\gamma\\
	&\tilde{P}^{(-1,-1)}_{n\geq 2}(x):=\binom{2n-2}{n}^{-1}\cdot\\
	&\; \cdot\sum_{k=1}^{n-1}\binom{n-1}{k}\binom{n-1}{n-k}(x-1)^{n-k}(x+1)^k\,,
\end{aligned}
\end{equation}
where $\gamma=(B-A)/(A+B)$, while for the parameters $\alpha=-1,\beta>-1$ one defines
\begin{equation}
	\begin{aligned}\label{eq:SJbeta}
	&\tilde{P}_0^{(-1,\beta)}(x):=1\\
	&\tilde{P}_1^{(-1,\beta)}(x):=x-1\\
	&\tilde{P}_{n\geq 2}^{(-1,\beta)}(x):=\binom{2n+\beta-1}{n}^{-1}\\
		&\; \cdot\sum_{k=0}^n\binom{n-1}{k}\binom{n+\beta}{n-k}(x-1)^{n-k}(x+1)^k\,.
	\end{aligned}
	\end{equation}
Then the $\tilde{P}^{(\alpha,\beta)}_n(x)$ form a complete orthogonal system of polynomial eigenfunctions of the Jacobi differential operator at parameters in the aforementioned parameter ranges.
\end{prop}

The readers may have noticed that indeed the Sobolev-Jacobi polynomials $\tilde{P}^{(\alpha,\beta)}_n(x)$ as defined above coincide for $n\geq 2$ with the monic versions of the classical Jacobi polynomials (i.e.\ a different normalization choice where in each polynomial of degree $n$ the coefficient of $x^n$ is normalized to be $1$).\\

In other words, while the strategy of using systems of \emph{classical} orthogonal polynomial eigenfunctions of the Jacobi-operator as originally suggested by McQuarrie fails, one can ``repair'' the Ansatz by working with the Sobolev-orthogonal polynomials of Kwon and Littlejohn instead. These polynomials are orthogonal over the \emph{weighted Hilbert spaces} defined via the inner products $\Phi^{(-1,-1)}_{A,B}$ and $\Phi^{(-1,\beta)}_{C}$, respectively. This application of their framework to the study of binary chemical reactions appears to be new.\\

\subsection{Intermezzo: how to fix the free parameters in the inner products}\label{sec:GP}

In a remarkable series of papers~\cite{gurappa1996new,gurappa1999equivalence,gurappa2001novel,gurappa2002linear,panigrahi2003coherent,gurappa2004polynomial}, N.~Gurappa and P.K.~Panigrahi developed a novel method to construct polynomial eigenfunctions of second and higher order differential operators. If such a differential operator  is given\footnote{In this subsection, in order to be compatible in notation with the material presented in~\cite{gurappa1996new,gurappa1999equivalence,gurappa2001novel,gurappa2002linear,panigrahi2003coherent,gurappa2004polynomial}, we will temporarily employ the notational conventions $x\equiv \hat{x}$ and $\tfrac{d}{dx}\equiv \partial_x$.} as a polynomial in $x$ and $\frac{d}{dx}$, they observe that one can always uniquely express the equation for the  eigenfunctions $y(x)$ of the operator in the form
\begin{equation}
	\left(F(D)+P(x,\tfrac{d}{dx})\right)y(x)=0\,.
\end{equation}
Here, $D=x\frac{d}{dx}$ is the \emph{Euler operator} (which the readers may recognize as the number operator in the Bargmann-Fock basis). For instance, given a generic second order differential operator (see~\cite{gurappa1996new,gurappa1999equivalence,gurappa2001novel,gurappa2002linear,panigrahi2003coherent,gurappa2004polynomial} for more general cases), the term $F(D)$ reads 
\begin{equation}
	F(D)=a_2 D^2+a_1D+a_0\,,
\end{equation}
with $a_i\in \bR$ some constant coefficients. Crucially, $F(D)$ is a \emph{diagonal operator} on the space of monomials, due to
\begin{equation}
	D x^n=nx^n\,.
\end{equation}
The operator $P\equiv P(x,\frac{d}{dx})$ on the other hand contains all terms not expressible via polynomials in $D$, whence $P$ contains all ``off-diagonal'' contributions to the differential operator. The authors of~\cite{gurappa1996new,gurappa1999equivalence,gurappa2001novel,gurappa2002linear,panigrahi2003coherent,gurappa2004polynomial} then construct the polynomial eigenfunctions via the Ansatz
\begin{equation}
\begin{aligned}\label{eq:BPans}
	F(D)x^{\lambda}&=0\;
	\Rightarrow\; y_{\lambda}(x)=c_{\lambda}\hat{G}_{\lambda}x^{\lambda}\\
	\hat{G}_{\lambda}&:=\sum_{m\geq 0}(-1)^m\left[\frac{1}{F(D)}P(x,\tfrac{d}{dx})\right]^m\,.
\end{aligned}
\end{equation}
Here, $\lambda\in \bZ$ is an integer parameter and $c_{\lambda}\in \bR$ a normalization constant.

Specializing the approach of~\cite{gurappa1996new,gurappa1999equivalence,gurappa2001novel,gurappa2002linear,panigrahi2003coherent,gurappa2004polynomial} to the case of the Jacobi-type differential operator $D_{Jac}^{\alpha,\beta}$ for parameters $\alpha=-1$ and $\beta\geq -1$, we first compute the decomposition of the operator into ``diagonal'' and ``off-diagonal'' parts. After some elementary manipulations, we find the following results:
\begin{prop}\label{prop:SJP}
	For $\beta\geq -1$, the \emph{Sobolev-Jacobi polynomials} $\tilde{P}^{(-1,\beta)}_n(x)$ have the presentation
	\begin{equation}\label{eq:GPSJ}
		\begin{aligned}
			\tilde{P}^{(-1,\beta)}_n(x)&=\hat{G}_{\beta,n}\:x^n\\
			\hat{G}_{\beta,n}&=\sum_{m\geq0}
			\bigg[\frac{1}{(D-n)(D+n+\beta)}\cdot\\
			&\qquad\quad\cdot \left(\tfrac{d^2}{dx^2}+(\beta+1)\tfrac{d}{dx}\right)\bigg]^m\,.
		\end{aligned}
	\end{equation}
	More explicitly, since by virtue of this definition
	\begin{equation}
		\tilde{P}^{(-1,\beta)}_1(x)
		=x-1+\delta_{\beta,-1}\,,
	\end{equation}
	the polynomials $\tilde{P}^{(-1,\beta)}_n(x)$ are identified as Sobolev-Jacobi polynomials w.r.t.\ the inner products $\Phi_{A,A}^{(-1,-1)}$ and $\Phi_C^{(-1,\beta)}$, respectively. Moreover, for the case $\beta=-1$ one may simplify the expression for $\tilde{P}^{(-1,-1)}_n(x)$ into the form
	\begin{equation}\label{eq:GPSJalt}
		\tilde{P}^{(-1,-1)}_n(x)
		=e^{-\tfrac{1}{2}\frac{1}{D+n-1}\tfrac{d^2}{dx^2}}x^n\,.
	\end{equation}
\end{prop}
\begin{proof}
See Appendix~\ref{app:GPSJ}.	
\end{proof}

It is via these novel techniques that one may hope to tackle the cases of multi-species binary chemical reaction systems, which we plan to elaborate upon in future work. For the present paper, we merely note that this approach suggests a particular choice for the parameters $A$ and $B$ of the inner product $\Phi^{(-1,-1)}_{A,B}$, which we will adopt from hereon in all of the ensuing computations: the ``canonical'' choice of parameters is according to Proposition~\ref{prop:SJP} to set $A=B$, which yields indeed
\[
	\gamma=(B-A)/(A+B)=0\,.
\]
For later convenience, we further choose to fix the remaining free parameters to $A=1$ and $C=1$.

\subsection{Solving binary chemical reaction systems via Sobolev-orthogonal polynomials}\label{sec:SJCR}

It remains to provide a formula for computing the linear combination coefficients $A_{M;n}$ in the separation of variables Ansatz~\eqref{eq:MQA} for the time-dependent probability generating function $P(t;x)$. Let us then assume as before that we are given the PGF of the initial state of the form
\[	
	P(0;x)=x^M\,,
\]
whence the PGF of a pure state with $M\geq0$ particles. We will follow the ideas of McQuarrie, modified via the usage of Sobolev-Jacobi polynomials. In order to clearly distinguish the coefficients $A_{M,n}$ from those in the present Ansatz, we choose to introduce the notations
\[
	A^{M,n}_{\beta}\;\text{and } B^{M,n}_{\beta}
\]
for the coefficients in the ensuing formulae. Adapting the Ansatz presented in Lemma~\ref{lem:CRbin} according to 
\begin{equation}
\begin{aligned}
	&P(t;x)=\\
	&\quad\sum_{n=0}^M A^{M,n}_{(r_d-1)}e^{-n(n+r_d-1)t}\tilde{P}_n^{(-1,r_d-1)}\left(\frac{x-\frac{r_{\ell}}{2}}{1-\frac{r_{\ell}}{2}}\right)\,,
\end{aligned}	
\end{equation}
it remains to solve the following equation:
\begin{equation}\label{eq:AMn}
x^M	=\sum_{n=0}^M A^{M,n}_{(r_d-1)}\tilde{P}_n^{(-1,r_d-1)}\left(\frac{x-\frac{r_{\ell}}{2}}{1-\frac{r_{\ell}}{2}}\right)\,.
\end{equation}
For the purpose of simplicity of the presentation, it is more economical to focus instead on the equation
\begin{equation}
	x^M=\sum_{n=0}^M B^{M,n}_{\beta} \tilde{P}^{(-1,\beta)}_n(x)\,,
\end{equation}
which then allows to extract the coefficients $A^{M,n}_{(r_d-1)}$ via performing a suitable change of variables in~\eqref{eq:AMn} (see Appendix~\ref{app:SJ} for further details):
\begin{equation}
	\begin{aligned}
		&A^{M,n}_{(r_d-1)}=\\
		&\quad\sum_{P=n}^M\binom{M}{P}(\tfrac{r_{\ell}}{2})^{M-P}(1-\tfrac{r_{\ell}}{2})^PB^{P,n}_{(r_d-1)}\,.
	\end{aligned}
\end{equation}
According to Proposition~\ref{prop:KL} and relying in addition on the results of the previous section, for both of the cases $\beta=-1$ and $\beta>-1$ we have appropriate inner products $\Phi_{1,1}^{(-1,-1)}$ and $\Phi_{1}^{(-1,\beta)}$ available with respect to which the Sobolev-Jacobi polynomials are orthogonal. Therefore, one can compute the coefficients $B^{M,n}_{\beta}$ as follows:
\begin{itemize}
\item \emph{Case $\beta=-1$:}
\begin{equation}
B^{M,n}_{-1}=\tfrac{\Phi_{1,1}^{(-1,-1)}\left(x^M,\tilde{P}^{(-1,-1)}_n(x)\right)}{\Phi_{1,1}^{(-1,-1)}\left(\tilde{P}^{(-1,-1)}_n(x),\tilde{P}^{(-1,-1)}_n(x)\right)}\,.
\end{equation}
\item \emph{Case $\beta>-1$:}
\begin{equation}
B^{M,n}_{\beta}=\tfrac{\Phi_{1}^{(-1,\beta)}\left(x^M,\tilde{P}^{(-1,\beta)}_n(x)\right)}{\Phi_{1}^{(-1,\beta)}\left(\tilde{P}^{(-1,\beta)}_n(x),\tilde{P}^{(-1,\beta)}_n(x)\right)}\,.
\end{equation}
\end{itemize}

We may yet again rely on the results of Kwon and Littlejohn for the values of the denominators (cf.\ Propositions 4.2 and 4.3 in~\cite{kwon1998sobolev}):
\begin{equation}
	\begin{aligned}
		&\Phi_{1,1}^{(-1,-1)}\left(\tilde{P}^{(-1,-1)}_n(x),\tilde{P}^{(-1,-1)}_n(x)\right)\\
		&\quad=\begin{cases}
2\quad &\text{if }n=0\\
4\quad &\text{if }n=1\\
n^2 K_{n-1} &\text{else.}	
\end{cases}\\
&K_{n-1} =\frac{2^{2n-1}((n-1)!)^4}{(2n-2)!(2n-1)!}\quad (n\geq 2)
	\end{aligned}
\end{equation}

\begin{equation}
	\begin{aligned}
		&\Phi_{1}^{(-1,\beta)}\left(\tilde{P}^{(-1,\beta)}_n(x),\tilde{P}^{(-1,\beta)}_n(x)\right)\\
		&\quad=\begin{cases}
1\quad &\text{if }n=0\\
n^2 \kappa_{n-1}(0,\beta+1) &\text{else.}	
\end{cases}\\
&\kappa_{n-1}(0,\beta+1) =\frac{2^{2n+\beta}((n-1)!(n+\beta+1)!)^2}{(2n+\beta-1)!(2n+\beta)!}\\
&\qquad\qquad (n\geq 1)\,.
	\end{aligned}
\end{equation}

However, the computation of the numerators of the formulae for $B^{M,n}_{\beta}$ appears to not have been performed in the literature before, and thus necessitates a considerable amount of additional work. Referring to Appendix~\ref{app:SJ} for some of the technical details of this derivation, suffice it to report here the final results for the coefficients $B^{M,n}_{\beta}$:

\begin{prop}
	For the case $\beta=-1$, the coefficients $B^{M,n}_{\beta}$ read:
	\begin{equation}\label{eq:SJPcoeffMoneMone}
		\begin{aligned}
			B^{M,n}_{-1}&=
			\begin{cases}
			\binom{2n}{n}
			\frac{m!\left(\frac{M+n}{2}\right)!}{(M+n)!\left(\frac{M-n}{2}\right)!},\quad &\text{if } M+n\in 2\bZ\\
			&\quad \text{and } n\leq M\\
			0,&\text{else.}
			\end{cases}
		\end{aligned}
	\end{equation}	
For the case $\beta>-1$, we find
\begin{equation}
	\begin{aligned}
		B^{M,n}_{\beta}&=\begin{cases}
	\delta_{n,0}\quad &\text{if } M=0\\
	1 &\text{if } n=0\\
	b^{M,n}_{\beta} &\text{if }	M,n\geq 1\text{ and } n\leq M\\
	0 &\text{else}
\end{cases}\,,
\end{aligned}
\end{equation}
where
\begin{equation}
\begin{aligned}
b^{M,n}_{\beta}&=\frac{M(2n+\beta)!}{2^{n-1}n!(n+\beta+1)!}\cdot\\
&\cdot \sum_{k=0}^{n-1}\bigg[
\binom{n-1}{k}(-1)^{n-1-k}\\
&\qquad\qquad {}_2F_1(1-M,n-k;\beta+n+2;2)\bigg]\,.
	\end{aligned}
\end{equation}
\begin{proof}
	See Appendix~\ref{app:SJ}.
\end{proof}
\end{prop}

While we thus do not find a closed-form solution for the coefficients $B^{M,n}_{\beta}$ for the case $\beta>-1$, the result may be easily and efficiently implemented with the help of a modern computer algebra software such as e.g.\ \textsc{Mathematica}, \textsc{Maple} or \textsc{Sage}.\\

In summary, we find the following analytic solutions to binary chemical reaction systems of the type related to Sobolev-Jacobi polynomials:

\begin{prop}
\label{prop:binGen}
Consider a chemical reaction system with Hamiltonian (for arbitrary $r_d,r_k,r_{\ell}\in \bR_{\geq0}$)
\begin{equation}
	\begin{aligned}
		\cH&=r_d\left(1-\hat{x}\right)\tfrac{\partial}{\partial x}\\
		&\quad
		+r_k\left(1-\hat{x}^2\right)\left(\tfrac{\partial}{\partial x}\right)^2
		+r_{\ell}\left(\hat{x}-\hat{x}^2\right)\left(\tfrac{\partial}{\partial x}\right)^2
	\end{aligned}
\end{equation}
Let $P_M(0;x):=x^M$ denote the PGF of an input state that contains exactly $M\geq0$ particles. Then the analytic solution to the PGF $P_M(t;x)$ for $t\geq 0$ reads for the case $r_{k}=r_{\ell}=0$
\begin{equation}
	P_M(t;x)=Bern(e^{-r_d t};x)\,,
\end{equation}
and for the case $r_k+r_{\ell}>0$
\begin{equation}
	\begin{aligned}
		P_M(t;x)&=\sum_{n=0}^M\bigg[ A^{M,n}_{(\bar{r}_d-1)} e^{- n(n+\bar{r}_d-1)\sigma t}\cdot\\
		&\qquad\qquad\cdot \tilde{P}^{(-1,\bar{r}_d-1)}_n\left(\tfrac{x-\tfrac{\bar{r}_{\ell}}{2}}{1-\tfrac{\bar{r}_{\ell}}{2}}\right)\bigg]\,,
\end{aligned}
\end{equation}
where
\begin{equation}
\begin{aligned}
		\sigma&:= r_k+r_{\ell}\,,\; \bar{r}_d=\tfrac{r_d}{\sigma}\,,\; \bar{r}_{\ell}=\tfrac{r_{\ell}}{\sigma}\\
		A^{M,n}_{(\bar{r}_d-1)}&=\sum_{P=n}^M\binom{M}{P}(\tfrac{\bar{r}_{\ell}}{2})^{M-P}(1-\tfrac{\bar{r}_{\ell}}{2})^PB^{P,n}_{(\bar{r}_d-1)}\,.
	\end{aligned}
\end{equation}
\begin{proof}
	For the case $r_k=r_{\ell}=0$, we simply recover our previous result for the pure decay reaction. For the case $r_k+r_{\ell}>0$, one may bring $\cH$ into the form compatible with the assumptions of Lemma~\ref{lem:CRbin} via
	\begin{equation}
		\cH=\sigma \tilde{\cH}\,,
	\end{equation}
	where $\sigma=r_k+r_{\ell}$. It then remains to apply the results presented above in order to derive the formula for $P_M(t;x)$. 
\end{proof}
\end{prop}

The present formulation permits thus to complete the tableau of analytically tractable elementary one-species reactions to \emph{all six} elementary types, see Figure~\ref{fig:discrNB} for graphical illustration. Despite the deceptive simplicity of the appearance of the distributional dynamics even for the binary reaction cases presented in Figures~\ref{fig:2Ato0A} and~\ref{fig:2Ato1A}, there is in fact an important conceptual divide present between these reactions and the non-binary reactions presented in Figures~\ref{fig:0Ato1A} to~\ref{fig:1Ato2A}. More precisely, while the examples presented in Figure~\ref{fig:discrNB} describe \emph{individual} elementary reactions, our analytical techniques only permit to describe \emph{composite} reaction systems of either of the two following forms: either the system is only composed of non-binary reactions (cf.\ Theorem~\ref{thm:NBCRN}), or alternatively it is composed of the decay reaction $1A\xrightharpoonup{\;\tau\;}0A$ and the two types of elementary binary reactions (cf.\ Proposition~\ref{prop:binGen}). Finding an analytical description of the dynamics of the other possible classes of  chemical reaction systems remains the subject of further ongoing research.

\section*{Conclusion and outlook}

Our goal in this work consisted in developing a framework within which the combinatorial properties and stochastic dynamics of chemical reaction systems may be discovered and analyzed. To this end, it proved fruitful to consider first a slightly more general theoretical setting, namely that of continuous-time Markov chains (CTMCs)~\cite{norris} over discrete state spaces, and to construct a suitable \emph{stochastic mechanics} framework for such theories (Section~\ref{sec:SMform}). It is interesting to note that while some concepts well-known from the study of quantum mechanical theories might provide a certain level of intuition in understanding the stochastic mechanics setup, the analogy is in fact not a particularly close one (see Appendix~\ref{app:QMvsSM}). We then specialized our framework to the main focus of interest in this paper, namely to chemical reaction systems (Sections~\ref{sec:CRS} and~\ref{sec:CRmulti}). Our approach highlights a rather fruitful synergy between the ideas of Delbr\"uck~\cite{delbruck1940statistical}, Doi~\cite{doi1976second} and the  techniques based on generating functions (for probability distributions and their various types of moments) known from the combinatorics literature. We provide a concise derivation of three different types of \emph{evolution equations} for analyzing chemical reaction systems, namely for the time-dependent states, for the exponential moment generating functions and for the factorial moment generating functions (Tables~\ref{tab:CRSone} and~\ref{tab:CCRmulti}). On the level of general results, our framework then permits to derive a precise statement specifying which reaction systems possess the special dynamical feature of first-order moment closure (Section~\ref{sec:SCCR}).\\

The second main contribution of our paper consists in a set of exact, closed-form results for the time-evolution of probability distributions of chemical reaction systems (Section~\ref{sec:ana}). For chemical reactions of the non-binary type, we identify the problem of finding a solution to the evolution equations as a special case of a so-called semi-linear normal-ordering problem (Section~\ref{sec:compNBCR}). Combining a powerful result from the combinatorics literature~\cite{blasiak2005boson,blasiak2011combinatorial} (cf.\ Theorem~\ref{thm:seminal}) with some elementary notions on constructions of convolutions and compound probability distributions, we are able to derive not only exact solutions for \emph{individual} elementary non-binary chemical reactions (Table~\ref{tab:nonBinaryCRs}), but also for \emph{arbitrarily composed} non-binary reaction systems (Section~\ref{sec:compNBCR}). While there exists a large body of work in the standard chemistry literature (see~\cite{mcquarrie1967stochastic} and in particular~\cite{jahnke2007solving}) on such non-binary reactions,  our approach itself and some of our exact results appear to be new. It would be interesting to also consider multi-species non-binary reaction systems (such as those already treated in~\cite{jahnke2007solving} for reaction systems not involving autocatalytic reactions) in our novel framework, but in this paper we focussed on the single-species systems for brevity. Some examples for the dynamics of non-binary reaction systems are illustrated in Figures~\ref{fig:NBCRN1} and~\ref{fig:NBCRN2}.

Finally, it has been an open problem in the literature to analytically describe the evolution of chemical reaction systems involving binary chemical reactions. We highlight in Sections~\ref{sec:bCR} to~\ref{sec:SJCR} how the traditionally employed Ansatz due to McQuarrie~\cite{mcquarrie1964kinetics} is ill-conceived due to an inconsistency in his choice of basis of eigenfunctions, and how the Ansatz may be repaired via modifying the relevant eigenbasis to so-called \emph{Sobolev-Jacobi orthogonal polynomials}, which were originally introduced by Kwon and Littlejohn in the 1990s~\cite{kwon1994characterizations,kwon1997classification,kwon1996new,kwon1998sobolev,everitt2001orthogonal}. We present in Section~\ref{sec:SJCR} an analytic solution to the time-evolution of the probability distributions for binary chemical reaction systems of one species (Proposition~\ref{prop:binGen}), which is illustrated for individual binary reactions in Figures~\ref{fig:2Ato0A} and~\ref{fig:2Ato1A}. We also present an interesting alternative definition of the Sobolev-Jacobi polynomials in terms of a remarkable framework due to Gurappa and Panigrahi~~\cite{gurappa1996new,gurappa1999equivalence,gurappa2001novel,gurappa2002linear,panigrahi2003coherent,gurappa2004polynomial} (Section~\ref{sec:GP}), which provides intriguing perspectives in terms of extending our analytical results from single- to multi-species binary chemical reaction systems.\\

Two main avenues of future research suggest themselves as a result of our work: firstly, as already realized by one of the authors in a long-standing research project~\cite{bdg2016,bdgh2016,bCCDD2017}, chemical reactions may alternatively be considered as a particular class of stochastic transition systems, so-called discrete transition systems. More precisely, the pure states of a chemical system are fully characterized in terms of the respective numbers of particles of each of the chemical species in the pure state. In contrast, it is possible to formulate a full stochastic mechanics framework (based on the concept of the so-called \emph{rule algebras}~\cite{bdg2016,bdgh2016}) for more general types of stochastic transition systems, whose state spaces consist of certain isomorphism classes of graphs, and whose transitions consist of \emph{graph rewriting rules}. Without going into much further detail on these systems here, it is worthwhile to mention that the results presented in this paper on the three types of evolution equations for the moments and factorial moments of generic multi-species chemical reaction systems (Table~\ref{tab:CCRmulti}) may not only be generalized to the setting of generic transition systems, but in some cases conversely even arise as the evolution equations for certain observables in the generic systems. These results will be presented in~\cite{bCCDD2017}.\\

Finally, while we are well aware that of course the examples presented do not reflect the complexity of reaction systems typically encountered in applications, we hope that our techniques may provide a starting point for the development of novel approximation schemes and algorithms for computing the time-dependent states of realistically complex reaction systems. We believe that it is of the utmost importance in this endeavor to take into account the combinatorics of the infinitesimal generators of chemical reaction systems, such as highlighted for example in the semi-linear normal-ordering techniques, and that it will prove paramount to understand in further detail why there appears to exist a clear structural divide between non-binary and binary reactions in terms of their evolution semi-groups. Referring to classification schemes for ODEs of second order such as in the work of Zhang~\cite{zhang2012exact}, already the appearance of Sobolev-Jacobi orthogonal polynomials as eigenbasis appears to be a new development, and it would be rather interesting to understand in which way these polynomials interact with the infinitesimal generators of generic reaction systems in order to potentially provide a general closed-form solution for at least the single-species cases.

\section*{Acknowledgements}
The work of NB was supported by the European Union's \emph{Horizon 2020 research and innovation programme} in form of a \emph{Marie Sk\l{}odowska-Curie Actions Individual Fellowship} (Grant Agreement No.~753750 -- RaSiR), and in an earlier phase by Vincent Danos via his \emph{ERC Advanced Fellowship} (Grant Agreement No.~320823 -- RULE). We would like to thank S.~Abbes, V.~Danos, I.~Garnier and the participants of the IH\'{E}S workshop \emph{``Combinatorics and Arithmetic for Physics: special days''} in November 2017 for fruitful discussions. NB would like to thank the \'{E}quipe Antique at DI-ENS Paris and the LPTMC (Paris 06) for warm hospitality.

\footnotesize{

}


\balance

\newpage
\onecolumn
\appendix

\section{Comparison of stochastic and quantum mechanics formalism}
\label{app:QMvsSM}

Since this constitutes a rather crucial aspect of our formalism, we review here the key differences between stochastic and quantum mechanics. We will focus on the scenarios most directly comparable, namely the cases where the space of pure states is some denumerable space of configurations (for simplicity considered in $0$ spacial dimensions), as e.g.\ in the case of a chemical reaction system vs.\ the case of a quantum mechanical harmonic oscillator. Our aim is not so much to perform an in-depth analysis, but rather allow the interested readers to sharpen their intuition about the mathematical and conceptual properties of stochastic mechanics. Finally, for the experts, we note that in the following comparison we only consider the case of \emph{time-independent} infinitesimal generators $H$ and Hamiltonians $\cH$ (where the latter is not to be confused with the expression of the infinitesimal generator $H$ in the Bargmann-Fock basis as used in the main text). We refer the interested readers to the excellent review~\cite{baez2012course} for references and further technical details.

\subsection{Explanation of conventions in definitions of stochastic mechanics}
\label{app:conv}

In quantum mechanics, one would tend to normalize the basis and dual vectors differently, namely instead of $\braket{m}{n}=\delta_{m,n}n!$ one would rather define the inner product such as to have orthonormal basis vectors. However, the present choice is motivated by the following combinatorial interpretation:
\begin{lem}
From the canonical commutation relation $[a,a^{\dag}]=\mathbb{1}$, it follows by applying the Hermitean conjugation $\dag$ to this equation that
\begin{equation}
	(a)^{\dag}=a^{\dag}\,,\quad \left(a^{\dag}\right)^{\dag}=a\,.
\end{equation}
Fixing the normalization of $\bra{0}$ to
	\begin{equation}
		\braket{0}{0}:=1
	\end{equation}
	and using the defining relations of the canonical representation of the HW algebra to infer that
	\begin{equation}
		\ket{n}=a^{\dag \:n}\ket{0}\,,\quad \bra{m}=\left(\ket{m}\right)^{\dag}=\bra{0}a^m\,,
	\end{equation}
	this implies the stochastic mechanics choice of inner product,
	\begin{equation}
		\braket{m}{n}=\delta_{m,n}n!\,.
	\end{equation}
\begin{proof}
The claim follows from a direct application of~\eqref{eq:HWno}. It follows from the definition of the canonical representation of the HW algebra that
\begin{align*}	
	a^m\ket{n}&=a^m a^{\dag\:n}\ket{0}=\sum_{k=0}^{min(m,n)}k!\binom{m}{k}\binom{n}{k}
	a^{\dag\:(n-k)}a^{(m-k)}\ket{0}=\Theta(n-m)(n)_m \ket{n-m}\,,
\end{align*}%
where 
\[
\Theta(x)=\begin{cases}
1\quad &\text{if }x\geq 0\\
0 	&\text{else}
\end{cases}
\]
is (a variant of) the Heaviside step function, and where $(n)_m:=n!/(n-m)!$ is a \emph{falling factorial}. Together with
\[
\left(a^r\ket{0}\right)^{\dag}=0=\bra{n}a^{\dag\:n}\,,
\]
the claim follows then from
\begin{align*}
	\braket{m}{n}&=\bra{0}a^{m}a^{\dag\:n}\ket{0}=(n)_m \bra{0}a^{\dag\:(n-m)}\ket{0}=n!\delta_{m,n}\,.
\end{align*}
\end{proof}
\end{lem}
Another consequence of the definition of the dual number vectors $\bra{m}$ is that the dual reference vector $\bra{\cR}$ is in fact analogous to a \emph{coherent state} in quantum mechanics:
\begin{lem}
The dual reference vector $\bra{\cR}$ as defined above may equivalently be defined as\footnote{In comparison to the standard definitions of (dual) coherent states,
\[
	\ket{z}:=e^{-\frac{|z|^2}{2}}e^{za^{\dag}}\ket{0}\,,\; \bra{z}:=e^{-\frac{|z|^2}{2}}\bra{0}e^{\bar{z}a}\quad (z\in \bC)\,,
\]
the dual reference vector $\bra{\cR}$ is proportional to a dual coherent state with $z=1$.}
\begin{equation}
	\bra{\cR}=\bra{0}e^a\,,
\end{equation}	
which implies that it is a \emph{left eigenvector with eigenvalue $1$ of the creation operator} $a^{\dag}$\,,
\begin{equation}
	\label{eq:refEV}
	\bra{\cR}a^{\dag}=\bra{\cR}\,.
\end{equation}
\begin{proof}
By direct computation:
\begin{align*}
\bra{\cR}a^{\dag}&=\bra{0}e^a a^{\dag}=\sum_{m=0}^{\infty}\frac{1}{m!}\bra{0}a^ma^{\dag}=\sum_{m=1}^{\infty}\frac{(m)_1}{m!}\bra{0}a^{m-1}=\bra{\cR}\,.	
\end{align*}%
\end{proof}

\end{lem} 

\renewcommand{\arraystretch}{1.9}

\begin{sidewaystable*}[p]
\small
  \caption{\ Comparison of mathematical concepts of stochastic and of quantum mechanics.}
  \label{tab:SMvsQM}
  \renewcommand{\arraystretch}{1.7}
  \begin{tabular*}{\textwidth}{@{\extracolsep{\fill}}lllllll}
    \hline
    concept & stochastic mechanics & quantum mechanics\\
    \hline
    pure states & $\ket{s}\in S$ & $\ket{s}\in S$\\
    state space & 
    \begin{tabular}[t]{@{}RCL@{}}
    Prob(S)&:=&\bigg\{\ket{\Psi}=\sum_{s}\psi_s\ket{s}\bigg\vert\\
    &&\psi_s\in \bR_{\geq0}\text{ and }\sum_s\psi_s=1\bigg\}\\
    &&\text{(probability distributions)}
    \end{tabular} &
    \begin{tabular}[t]{@{}RCL@{}}
    Wave(S)&:=&\bigg\{\ket{\Psi}=\sum_{s}\psi_s\ket{s}\bigg\vert\\
    &&\psi_s\in \bC\text{ and }\sum_s|\psi_s|^2=1\bigg\}\\
    &&\text{(wave functions)}
    \end{tabular}\\
    comments &
    $Prob(S)$ is \emph{not} a Hilbert space! &
    $\ell^2_{\bC}(S)$ \emph{is} a Hilbert space, with \emph{dual space} $\cong$ $\ell^2_{\bC}(S)$
    \\
    ambient space & $\ell^1_{\bR}(S)$ & $\ell^2_{\bC}(S)$\\
    \hline
    time-dependent state &
    \begin{tabular}[t]{@{}RCL@{}}
    	\ket{\Psi(t)}&=& \cE(t)\ket{\Psi(0)}\\
    	\cE(0)&=&\mathbb{1}\\
    	\cE(s+t)&=&\cE(s)\cE(t)=\cE(t)\cE(s)\\
    	\cE(t)&\text{--}& \text{evolution semi-group}
    \end{tabular} &
    \begin{tabular}[t]{@{}RCL@{}}
    	\ket{\Phi(t)}&=& \hat{U}(t)\ket{\Phi(0)}\\
    	\hat{U}(0)&=&\mathbb{1}\\
    	\hat{U}(s+t)&=&\hat{U}(s)\hat{U}(t)=\hat{U}(t)\hat{U}(s)\\
    	\hat{U}(t)&\text{--}& \text{unitary evolution operator aka \emph{propagator}}
    \end{tabular}\\
    evolution equations &
    \begin{tabular}[t]{@{}RCLL@{}}
    \frac{d}{dt}\cE(t)&=& H\cE(t) &\\
    \frac{d}{dt}\ket{\Psi(t)}&=& H\ket{\Psi(t)}&\text{(Master equation)}
    \end{tabular}
    &
    \begin{tabular}[t]{@{}RCLL@{}}
    i\hbar\frac{d}{dt}\hat{U}(t)&=&\cH\hat{U}(t) &\\
    i\hbar\frac{d}{dt}\ket{\Phi(t)}&=& \cH\ket{\Phi(t)}&\text{(Schr\"{o}dinger equation)}
    \end{tabular}\\
    properties of $H$ vs.\ $\cH$ &
    \begin{tabular}[t]{@{}RCL@{}}
    \bra{\cR}H&=& 0
    \end{tabular}    
    &
    \begin{tabular}[t]{@{}RCL@{}}
    \cH^{\dag}&=& \cH
    \end{tabular}\\
    \hline
    structure of $\cE(t)$ vs.\ $\hat{U}(t)$ &
    \begin{tabular}[t]{@{}RCL@{}}
    \cE(t)&\neq & e^{tH}
    \end{tabular} &
    \begin{tabular}[t]{@{}RCL@{}}
    \hat{U}(t)&=&e^{\frac{t}{i\hbar}\hat{\cH}}\\
    && \text{(Stone-von-Neumann theorem)}
    \end{tabular}\\
    \hline
    observables &
    \begin{tabular}[t]{@{}RCL@{}}
    O &\text{--}& \text{diagonal linear operators:}\\
    O\ket{s}&=& \omega_s\ket{s}\;\text{for all $\ket{s}\in S$}
    \end{tabular} &
    \begin{tabular}[t]{@{}RCL@{}}
    \hat{\cO} &\text{--}& \text{Hermitian operators:}\\
    \hat{\cO}^{\dag}&=& \hat{\cO}
    \end{tabular}\\
    \text{expectation values of observables}
    &
    \begin{tabular}[t]{@{}RCL@{}}
    \bE_{\ket{\Psi}}(O)&\equiv& \langle O\rangle_{\ket{\Psi}}\\
    &:=& \bra{\cR}O\ket{\Psi}
    \end{tabular} 
    &
    \begin{tabular}[t]{@{}RCL@{}}
    \langle \hat{O}\rangle_{\ket{\Phi}}
    &:=& \bra{\Phi}O\ket{\Phi}
    \end{tabular} 
    \\
    evolution of observables &
    \begin{tabular}[t]{@{}RCL@{}}
    \frac{d}{dt}\left\langle O\right\rangle(t)&=&\left\langle [O,H]\right\rangle(t)\\
    	&=&\bra{\cR}[O,H]\ket{\Psi(t)}
    \end{tabular}
    &
    \begin{tabular}[t]{@{}RCL@{}}
    \frac{d}{dt}\left\langle \hat{O}\right\rangle(t)&=&\left\langle [O,H]\right\rangle(t)\\
    	&=&\bra{\Phi(t)}[O,H]\ket{\Phi(t)}
    \end{tabular}\\
    \hline
    \end{tabular*}
 \end{sidewaystable*}

\subsection{Comparison of stochastic and quantum mechanics} 
 
Notably, one of the crucial differences concerns the nature of the ambient spaces the two types of theory are formulated in: for quantum mechanics, since the space of $\ell^2$-summable sequences with complex coefficients may be equipped with the structure of a \emph{Hilbert space} whose dual is isomorphic to the space itself, and in contrast to the stochastic mechanics case there exists a rather precise construction known as the \emph{Stone-von Neumann theorem} that uniquely relates Hermitean operators $\cH$ to unitary evolution semigroups $\hat{U}(t)$,
 \begin{equation}
 	\hat{U}(t)=e^{\frac{t}{i\hbar}\cH}\,.
 \end{equation}
 Moreover, the existence of a proper sesqui-linear inner product $\langle,\rangle$ renders expressions such as the expectation value
 \begin{equation}
 	\langle \hat{O}\rangle_{\ket{\Phi}}=\bra{\Phi}\hat{O}\ket{\Phi}
 \end{equation}
 of a quantum mechanical observable $\hat{O}$ well-posed.\\
 
 In stark contrast, no such mathematical elegance is available in the general case in the setting of stochastic mechanics in general. While in the case of a \emph{finite} state space one has the rather similar relationship
 \begin{equation}
 	\cE(t)=e^{tH}
 \end{equation}
 between infinitesimal generator $H$ and evolution semi-group $\cE(t)$, one in general is confronted with \emph{denumerable} aka countably infinite-dimensional spaces and \emph{unbounded} infinitesimal generators $H$. In this general setting, while the evolution equation
 \begin{equation}
 	\frac{d}{dt}\cE(t)=H\cE(t)
 \end{equation}
 remains valid, one needs to employ the intricate functional analysis techniques of the \emph{Hille-Yosida theory} in order to interpret the precise relationship between $H$ and $\cE(t)$ further. Note also that the interpretation of observables and expectation values is far more direct than in the quantum mechanical case, since they amount to the standard respective notions of probability theory. For example,
 \begin{equation}
 \begin{aligned}
 	\bE_{\ket{\Psi}}(O)&=\bra{\cR}O\ket{\Psi}
 	=\sum_s \psi_s\bra{\cR}O\ket{s}\\
 	&=\sum_s \psi_s \omega_s \braket{\cR}{s}=\sum_s \psi_s \omega_s\,,
 \end{aligned}
 \end{equation}
 which amounts to the conventional definition of the expectation value of an observable. It is therefore rather remarkable that the evolution equations for the expectation values of observables take such analogous looking forms (bearing in mind the entirely different meanings of expectation values in stochastic and quantum mechanics of course). In fact, these formal similarities motivated the moniker ``stochastic mechanics'' for this particular flavor of stochastic theories.
 
 \subsection{Normal-ordering in the HW algebra}
 \label{app:HWnorm}
 
The combinatorics encoded in the HW algebra is conveniently encoded in the well-known \emph{normal ordering formula}, giving the structure constants of the HW algebra:
\begin{equation}
	\label{eq:HWno}
	a^{\dag\:m}a^na^{\dag\:r}a^s=\sum_{k=0}^{min(n,r)}k!\binom{n}{k}\binom{r}{k}a^{\dag\:(m+r-k)}a^{(n+s-k)}\,.
\end{equation}
This formula may be derived in a number of ways, including inductive application of $[a,a^{\dag}]=1$, using the Bargmann-Fock representation $a^{\dag}\equiv z$, $a\equiv\partial_z$, or more intuitively via using diagram algebraic reasoning~\cite{blasiak2010combinatorial,bdgh2016}. It plays a crucial role in our framework, since it allows to represent any operator constructed from an arbitrary word in $a$ and $a^{\dag}$ in terms of its \emph{normal-ordered form}, either by inductively applying~\eqref{eq:HWno} or via more sophisticated techniques from the combinatorics literature.

\section{Definitions and identities involving Stirling numbers}\label{app:Stirling}

Denote by $s_1(k,\ell)$ the Stirling numbers of the first kind, and by $S_2(\ell,k)$ the Stirling numbers of the second kind (sequences \href{https://oeis.org/A008275}{A008275} and~\href{https://oeis.org/A008277}{A008277} in the OEIS~\cite{OEIS}, respectively; see also~\cite{comtet2012advanced}):
\renewcommand{\arraystretch}{1.0}
\begin{equation}
\begin{aligned}
	s_1(k,\ell)&:=\delta_{k,0}\delta_{0,\ell}+\Theta(k-1)(-1)^{k-\ell}\left[\begin{array}{c}
		k\\ \ell
	\end{array}\right]\,,\quad \left[\begin{array}{c}
		k\\ \ell
	\end{array}\right]&:=\sum_{1\leq i_1<\dotsc<i_{\ell-1}<k}\frac{(k-1)!}{i_1\cdots i_{\ell-1}}\\
	S_2(\ell,k)&=\frac{1}{k!}\sum_{j=0}^k(-1)^{k-j}\binom{k}{j}j^{\ell}	\,.
\end{aligned}
\end{equation}
Some useful identities involving the two kinds of Stirling numbers are ($\forall n>0$)
\begin{equation}
\begin{aligned}
s_1(n,0)&=s_1(0,n)=0\,,\; s_1(n,1)=(-1)^{n-1}(n-1)!\,\;  s_1(n,n)=1\,,\; s_1(n,n-1)=-\binom{n}{2}\\
S_2(0,k)&=\delta_{k,0}\,,\; S_2(\ell,1)=\Theta(\ell-1)\,,\; S_2(\ell,\ell)=1\,.
\end{aligned}
\end{equation}
and for all $n>k$ we have that $s_1(k,n)=S_2(k,n)=0$. We also note the well-known \emph{orthogonality relations}
\begin{equation}
\begin{aligned}\label{eq:StirlingOrth}
\sum_{n=\ell}^k s_1(k,n)S_2(n,\ell)&=\sum_{n=0}^k s_1(k,n)S_2(n,\ell)=\delta_{k,\ell}\\
\sum_{n=k}^{\ell} S_2(\ell,n)s_1(n,k)&=
\sum_{n=0}^{\ell} S_2(\ell,n)s_1(n,k)=\delta_{k,\ell}\,,
\end{aligned}
\end{equation}
where equality of the respective two variants follows from $s_1(k,\ell)=S_2(k,\ell)=0$ for $\ell>k$. The orthogonality relations are straightforward to derive from the exponential generating functions for the Stirling numbers,
\begin{equation}
\begin{aligned}
\sum_{\ell=0}^{\infty}S_2(\ell,k)\frac{x^{\ell}}{\ell!}&=\frac{(e^x-1)^k}{k!}\,,\quad 
\sum_{\ell=0}^{\infty}s_1(\ell,k)\frac{x^{\ell}}{\ell!}=\frac{(\ln(1+x))^k}{k!}\,.
\end{aligned}
\end{equation}
The essential step in the derivation of the orthogonality relations consists in noticing that the functions $f(x)=e^x-1$ and $g(x)=\ln(1+x)$ are compositional inverses of one another, i.e.\ $f(g(x))=g(f(x))=x$.\\

Of particular importance in the present paper is the concept of \emph{Stirling transforms}~\cite{bernstein1995some}. Given two sequences $\{a_n\}^{\infty}_{n\geq 0}$ and $\{b_k\}^{\infty}_{k\geq 0}$ that are related by
\begin{equation}
b_n=\sum_{k=0}^nS_2(n,k)a_k\,,\; a_k=\sum_{n=0}^k s_1(k,n)b_n\,,
\end{equation}
then their exponential generating functions $\mathcal{A}(x)=\sum_{k=0}^{\infty}a_k\frac{x^k}{k!}$ and $\mathcal{B}(x)=\sum_{k=0}^{\infty}b_k\frac{x^k}{k!}$ are related via
\begin{equation}
\mathcal{B}(x)=\mathcal{A}\left(e^{x}-1\right)\,,\quad 
\mathcal{A}(x)=\mathcal{B}\left(\ln(x+1)\right)\,.
\end{equation}

\section{Proofs}
\label{app:proofs}

\subsection{Proof of the EGF evolution equation in its most general form}
\label{app:proofGenMGFev}

In preparation of the proof, let us state the following useful corollary of the Baker-Campbell-Hausdorff formula:
\begin{prop}[Cf.~\cite{hall2015lie}, Prop.~3.35]\label{prop:conv}
For two linear operators $A,B$ and a formal variable $\lambda$, we obtain the following formal power-series result:
\begin{equation}
e^{\lambda A} B e^{-\lambda A}=e^{\lambda ad_{A}}B\,,
\end{equation}
where
\begin{equation}
ad_A B:=[A,B]=AB-BA\,,\; ad^0_A B :=B\,,\; ad^{q}_A B:=\underbrace{ad_A\circ\dotsc\circ ad_A}_{\text{$q>0$ times}}B\,.
\end{equation}
\begin{proof}
Noting that at this point we \emph{do not} make any statement whatsoever about the convergence of the power-series involved in the above formula, we simply proceed completely analogously to~\cite{hall2015lie} (p.~74, Ex.~14) and expand both exponentials into their Taylor series, followed by collecting terms:
\begin{align*}
e^{\lambda A}Be^{-\lambda A}&=\sum_{m,n\geq0} \frac{\lambda^{m+n}(-1)^n}{m!n!} A^m B A^n\\
&=\sum_{p\geq 0}\lambda^p \sum_{m,n\geq0}\delta_{m+n,p}
\frac{(-1)^n}{m!n!} A^m B A^n\\
&=\sum_{p\geq 0}\frac{\lambda^p}{p!}\sum_{k=0}^{p}\binom{p}{k}A^k B (-A)^{p-k}\,.
\end{align*}%
Finally, using the standard auxiliary formula
\[
ad_X^p Y=\sum_{k=0}^p \binom{p}{k}X^k Y (-X)^{p-k}\,,
\]
we obtain the statement.
\end{proof} 
\end{prop}

The following material is a variation on material by the first-named author as presented in~\cite{bdg2016,bCCDD2017}:

\thmGenEGFev*
\begin{proof}
Recall first the following crucial properties of the evolution semi-group $\cE(t)$ and of the infinitesimal generator $H$:
\begin{align*}
	\frac{d}{dt}\cE(t)&= H\cE(t)\,,\quad \bra{\cR}H=0\,.
\end{align*}%
Then the claim of the theorem follows from a judicious application of these properties and of Proposition~\ref{prop:conv}:
\begin{align*}
	\frac{d}{dt}\cM(t;\vec{\lambda})&=\frac{d}{dt}\bra{\cR}e^{\vec{\lambda}\cdot\vec{O}}\cE(t)\ket{\Psi(0)}\\
	&=\bra{\cR}e^{\vec{\lambda}\cdot\vec{O}}H\cE(t)\ket{\Psi(0)}\\
	&=\bra{\cR}\left(e^{\vec{\lambda}\cdot\vec{O}}H e^{-\vec{\lambda}\cdot\vec{O}}\right)e^{\vec{\lambda}\cdot\vec{O}}\cE(t)\ket{\Psi(0)}\\
	&=\bra{\cR}\left(e^{ad_{\vec{\lambda}\cdot\vec{O}}}H\right)e^{\vec{\lambda}\cdot\vec{O}}\ket{\Psi(t)}\\
	&=\bra{\cR}He^{\vec{\lambda}\cdot\vec{O}}\ket{\Psi(t)}
	+\sum_{q\geq 1}\frac{1}{q!}
	\bra{\cR}\left(ad_{\vec{\lambda}\cdot\vec{O}}^{\circ\:q}\:H\right)e^{\vec{\lambda}\cdot\vec{O}}\ket{\Psi(t)}\,,
\end{align*}%
from which the claim follows by noting that the first term vanishes due to $\bra{\cR}H=0$.
\end{proof}

\subsection{Proof of the evolution equation Theorem~2}
\label{app:proofThmEGFone}

We need the following auxiliary formula:
\begin{equation}
\begin{aligned}
	ad_{\lambda\hat{n}} a^{\dag\:o}a^i&=\lambda[a^{\dag}a,a^{\dag\:o}a^i]
    =\lambda\bigg[\sum_{k=0}^{min(1,o)}k!\binom{1}{k}\binom{o}{k} -\sum_{k=0}^{min(i,1)}k!\binom{1}{k}\binom{i}{k}\bigg]a^{\dag\:(o+1-k)}a^{(i+1-k)}\\
	&=\lambda(o-i)a^{\dag\:o}a^i\,.
\end{aligned}
\end{equation}
Consequently,
\begin{equation}
	ad^{\circ\:q}_{\lambda\hat{n}} a^{\dag\:o}a^i
	=\lambda^q(o-i)^qa^{\dag\:o}a^i\,.
\end{equation}
The claim of the theorem then follows by a direct computation:
\begin{equation}
\begin{aligned}
	\frac{d}{dt}\cM(t;\lambda)&=\left\langle e^{\lambda\hat{n}}H\right\rangle(t)=\left\langle \left(e^{\lambda\hat{n}}He^{-\lambda\hat{n}}\right)e^{\lambda\hat{n}}\right\rangle(t)\overset{(*)}{=}\sum_{q=1}^{\infty}\frac{1}{q!}\left\langle \left(ad_{\lambda\hat{n}}^{\circ\:q}\hat{H}\right)e^{\lambda\hat{n}}\right\rangle(t)\\
	&=\sum_{(i,o)}r_{i,o}
	\left(e^{\lambda(o-i)}-1\right)
	\left\langle a^{\dag\:o}a^ie^{\lambda\hat{n}}\right\rangle(t)\overset{(**)}{=}\mathbb{D}(\lambda,\partial_{\lambda})\cM(t;\lambda)\,.
\end{aligned}
\end{equation}
Here, in the step marked $(*)$ we have used that $\bra{\cR}H=0$ as well as the auxiliary relation
\begin{equation}
	[a^{\dag}a,H]=[a^{\dag}a,\hat{H}]\,,
\end{equation}
while in the step marked $(**)$ we have used that $\bra{\cR}a^{\dag}=\bra{\cR}$, and consequently
\begin{equation}
\left\langle a^{\dag\:o}a^ie^{\lambda\hat{n}}\right\rangle(t)
=\left\langle a^ie^{\lambda\hat{n}}\right\rangle(t)
=\left\langle a^{\dag\:i}a^ie^{\lambda\hat{n}}\right\rangle(t)
=\sum_{\ell=0}^i s_1(i,\ell)\left\langle \hat{n}^{\ell}e^{\lambda\hat{n}}\right\rangle(t)
=\left[\sum_{\ell=0}^i s_1(i,\ell)\frac{\partial^{\ell}}{\partial \lambda^{\ell}}\right]\cM(t;\lambda)\,.
\end{equation}

\subsection{Proof of the factorial moment EGF evolution Theorem~3}
\label{app:proofThm3}

The proof is a straightforward application of a Stirling transform followed by a change of variables (cf.\ Appendix~\ref{app:Stirling}). According to Theorem~\ref{thm:StirlingOne}, the exponential generating function $\cM(t;\lambda)$ of the moments of the number operator $\hat{n}=a^{\dag}a$ and the EGF $\cF(t;\nu)$ of the factorial moments $a^{\dag\:k}a^k$ are related via the Stirling transform
\begin{equation}
\cM(t;\lambda)=\cF(t;e^{\lambda}-1)\,.
\end{equation}
Therefore, 
\begin{equation}
\frac{d}{dt}\cM(t;\lambda)=\frac{d}{dt}\cF(t;e^{\lambda}-1)=\bD(\lambda,\partial_{\lambda})\cM(t;\lambda)=\bD(\lambda,\partial_{\lambda})\cF(t;e^{\lambda}-1)\,.
\end{equation}
Performing a change of variables
\begin{equation}
\nu:=e^{\lambda}-1\Rightarrow \lambda=\ln(\nu+1)\,,
\end{equation}
which in particular entails that
\begin{equation}
\frac{\partial^n}{\partial\lambda^n}=\sum_{k=0}^nS_2(n,k)(\nu+1)^k\frac{\partial^k}{\partial\nu^k}\,,
\end{equation}
the claim of the theorem follows from the formula~\eqref{eq:EGFoneEv} for $\bD(\lambda,\partial_{\lambda})$ (together with the orthogonality relation of the Stirling numbers as presented in Appendix~\ref{app:Stirling}):
\begin{equation}
\begin{aligned}
\mathbb{d}(\nu,\partial_{\nu})&=\bigg[\bD(\lambda,\partial_{\lambda})\bigg]_{\lambda\mapsto \nu:=e^{\lambda}-1}=\sum_{(i,o)}r_{i,o}\left((\nu+1)^{(o-i)}-1\right)\sum_{\ell=0}^{i}s_1(i,\ell)\frac{\partial^{\ell}}{\partial\lambda^{\ell}}\\
&=\sum_{(i,o)}r_{i,o}\left((\nu+1)^{(o-i)}-1\right)
\underbrace{\sum_{\ell=0}^{i}s_1(i,\ell)
\sum_{k=0}^{\ell}S_2(\ell,k)(\nu+1)^k\frac{\partial^k}{\partial\nu^k}}_{\overset{\eqref{eq:StirlingOrth}}{=}(\nu+1)^i\frac{\partial^i}{\partial\nu^i}}\,.
\end{aligned}
\end{equation}
This concludes the proof.

\subsection{Proof of the multi-species HW algebra auxiliary formula}
\label{app:HWmAuxproof}

In order to prove that
\begin{equation}
e^{\vec{\alpha}\cdot \vec{a}}f(\vec{a}^{\dag},\vec{a})e^{\vec{\beta}\cdot\vec{a}^{\dag}}=
e^{\vec{\alpha}\cdot\vec{\beta}}e^{\vec{\beta}\cdot\vec{a}^{\dag}}f\left(\vec{a}^{\dag}+\vec{\alpha},\vec{a}+\vec{\beta}\right)e^{\vec{\alpha}\cdot\vec{a}}\,,
\end{equation}
it suffices to prove the equation for an elementary normal ordered term $f=\vec{a}^{\dag\:\vec{r}}\vec{a}^{\vec{s}}$, namely via using the formula for the adjoint action of linear operators on one another,
\begin{equation}
e^A B e^{-A}=e^{ad_A}B\,.
\end{equation}
We begin with the proof for a single species:
\begin{equation}
e^{\alpha a}a^{\dag\: r}
=\left(e^{\alpha a} a^{\dag\:r}
e^{-\alpha a}\right)e^{\alpha a}
=\left(\sum_{n=0}^{\infty} \frac{\alpha^n}{n!} ad_{a}\left(a^{\dag\:r}\right)\right)e^{\alpha a}\,.
\end{equation}
It follows from the general normal ordering formula that
\begin{equation}
ad_a\left(a^{\dag\:r}\right)=[a,a^{\dag\:r}]=\Theta(r)r a^{\dag\:(r-1)}\,,
\end{equation}
and by induction
\begin{equation}
ad_a^{\circ\:n}\left(a^{\dag\:r}\right)=\Theta(r-n)(r)_b a^{\dag\:(r-n)}\,.
\end{equation}
This implies that
\begin{equation}
e^{\alpha a}a^{\dag\: r}
=\left(\sum_{n=0}^{r} \binom{r}{n}\alpha^n a^{\dag\:(r-n)} \right)e^{\alpha a}=(a^{\dag}+\alpha)^r e^{\alpha a}\,.
\end{equation}
It follows thus (with the second formula obtained by Hermitean conjugation)
\begin{equation}
e^{\vec{\alpha}\cdot \vec{a}}\vec{a}^{\dag\: \vec{r}}
=\left(\vec{a}^{\dag}+\vec{\alpha}\right)^{\vec{r}}e^{\vec{\alpha}\cdot\vec{a}}\,,\quad
\vec{a}^{\vec{s}}e^{\vec{\beta}\cdot\vec{a}^{\dag}}=e^{\vec{\beta}\cdot\vec{a}^{\dag}}\left(\vec{a}+\vec{\beta}\right)^{\vec{s}}\,.
\end{equation}
The claim of the proposition then follows by repeated application of these two identities.

\subsection{Proofs of the multi-species generating function evolution equation theorem}
\label{app:proofsHWMev}

To derive the evolution equation for the exponential moment generating function $\cM(t;\vec{\lambda})$, we first use the normal ordering formula~\eqref{eq:HWmultiNO} to demonstrate that
\begin{equation}
[\hat{n}_i,\vec{a}^{\dag\:\vec{r}}\vec{a}^{\vec{s}}]
=
[a_i^{\dag}a_i,a_i^{\dag\:r_i}a_i^{s_i}]\vec{a}^{\dag\:(\vec{r}-r_i\vec{\delta}_i)}\vec{a}^{(\vec{s}-s_i\vec{\delta}_i)}
=(r_i-s_i)\vec{a}^{\dag\:\vec{r}}\vec{a}^{\vec{s}}\quad \Rightarrow \quad ad_{\vec{\lambda}\cdot\vec{\hat{n}}}^{\circ\:q}\left(\vec{a}^{\dag\:\vec{r}}\vec{a}^{\vec{s}}\right)
=\left[\vec{\lambda}\cdot(\vec{r}-\vec{s})\right]^q\vec{a}^{\dag\:\vec{r}}\vec{a}^{\vec{s}}\,.
\end{equation}
The remaining part of the proof is then entirely analogous to the single-species case.\\

In order to prove the validity of the evolution equation for the multi-species factorial moment generating function $\cF(t;\vec{\nu})$, observe first that the canonical commutation relations~\eqref{eq:HWmultiCCR} entail that for all $i,j\in \mathbf{S}$ (with $\mathbf{S}$ the set of species) and for all $k,\ell\geq 0$,
\begin{equation}
\begin{aligned}
\left[a_i^{\dag\:k}a_i^k,a_j^{\dag\:\ell}a_j^{\ell}\right]=0\,.
\end{aligned}
\end{equation}
Therefore, the orthogonality relations~\eqref{eq:StirlingOrth} for the Stirling numbers $s_1(k,n)$ and $S_2(n,k)$ generalize to the following multi-species variants:
\begin{equation}\label{eq:StilingMultiOR}
\begin{aligned}
\sum_{\vec{n}=\vec{\ell}}^{\vec{k}} 
\vec{s}_1(\vec{k},\vec{n})
\vec{S}_2(\vec{n},\vec{\ell})
&=\sum_{\vec{n}=\vec{0}}^{\vec{k}}
\vec{s}_1(\vec{k},\vec{n})
\vec{S}_2(\vec{n},\vec{\ell})
=\delta_{\vec{k},\vec{\ell}}\,,\quad \sum_{\vec{n}=\vec{k}}^{\vec{\ell}} 
\vec{S}_2(\vec{\ell},\vec{n})
\vec{s}_1(\vec{n},\vec{k})=
\sum_{\vec{n}=\vec{0}}^{\vec{\ell}}
\vec{S}_2(\vec{\ell},\vec{n})
\vec{s}_1(\vec{n},\vec{k})=\delta_{\vec{k},\vec{\ell}}\\
\vec{s}_{1}(\vec{x},\vec{y})&:=\prod_{i\in \mathbf{S}}\left(s_{1}(x_i,y_i)\right)\,,\quad 
\vec{S}_{2}(\vec{x},\vec{y}):=\prod_{i\in \mathbf{S}}\left(S_{2}(x_i,y_i)\right)\,.
\end{aligned}
\end{equation}
For precisely the same reason, the relationships between the moments of the number operators $\hat{n}_i$ ($i\in \mathbf{S}$) and the factorial moments are entirely analogous to the single species case presented in Theorem~\ref{thm:StirlingOne}:
\begin{equation}\label{eq:MomentMultiRels}
\left(\hat{\vec{n}}\right)^{\vec{k}}=\sum_{\vec{\ell}=\vec{0}}^{\vec{k}}\vec{S}_2(\vec{k},\vec{\ell})\vec{a}^{\dag\:\vec{\ell}}\vec{a}^{\vec{\ell}}\,,\quad
\vec{a}^{\dag\:\vec{\ell}}\vec{a}^{\vec{\ell}}=\sum_{\vec{k}=\vec{0}}^{\vec{\ell}}\vec{s}_1(\vec{\ell},\vec{k})\left(\vec{\hat{n}}\right)^{\vec{k}}\,.
\end{equation}
Therefore, it is the case also for the multi-species variant of the generating functions that $\cF(t;\vec{nu})$ is a (multi-species) \emph{Stirling transform} of $\cM(t;\vec{\lambda})$:
\begin{equation}
\cM(t;\vec{\lambda})=\cF(t;\vec{\eta})\,,\quad \eta_i:=e^{\lambda_i}-1\,.
\end{equation}
Performing the \emph{change of variables}
\begin{equation}
\nu_i:=e^{\lambda_i}-1\;\Rightarrow\; \lambda_i=\ln(\nu_i+1)\,,\; \frac{\partial^{\vec{p}}}{\partial\vec{\lambda}^{\vec{p}}}=
\prod_{i\in \mathbf{S}} \left(\sum_{k_i=0}^{p_i} S_2(p_i,k_i)(\nu_i+1)^{k_i}
\frac{\partial^{k_i}}{\partial \nu_i^{k_i}}\right)\,,
\end{equation}
the derivation of the evolution equation for $\cF(t;\vec{\nu})$ is then entirely analogous to the single-species case (where the multi-species orthogonality relations~\eqref{eq:StilingMultiOR} are to be used).

\section{Details of the derivation of the PGFs for composite non-binary reaction systems}
\label{app:comp}

For a reaction system consisting of non-binary one-species reactions,
\begin{align*}
\emptyset&\xrightharpoonup{\;\beta\;} S\,,\quad	
\emptyset\xrightharpoonup{\;\gamma\;} 2S\,,\quad 
S \xrightharpoonup{\;\tau\;} \emptyset\,,\quad
S\xrightharpoonup{\;\alpha\;} 2S\,,
\end{align*}%
we find according to~\eqref{eq:PGev} the following Hamiltonian in the Bargmann-Fock basis:
\begin{equation}
\begin{aligned}\label{eq:hamNBR}
\cH&=\beta(\hat{x}-1)+\gamma(\hat{x}^2-1)+\tau(1-\hat{x})\partial_x+\alpha(\hat{x}^2-\hat{x})\partial_x\,.
\end{aligned}
\end{equation}
We thus recognize that $\cH$ is of the \emph{semi-linear} form,
\begin{equation}
\cH=q(\hat{x})\partial_x+v(\hat{x})\,,\quad q(\hat{x})=\tau(1-\hat{x})+\alpha(\hat{x}^2-\hat{x})\,,\quad v(\hat{x})=\beta(\hat{x}-1)+\gamma(\hat{x}^2-1)\,,
\end{equation}
where in particular the birth and pair creation reactions yield contributions to the polynomial $v(\hat{x})$ only. According to Theorem~\ref{thm:seminal}, acting with $e^{\lambda \cH}$ on a formal power series $F(0;x)$ yields a $\lambda$-dependent formal power series $F(\lambda;x)$ computable as
\begin{subequations}
\begin{align}
	F(\lambda;x)&=e^{\lambda \cH}F(0;x)=g(\lambda;x)F\big(0;T(\lambda;x)\big)\\
	\tfrac{\partial}{\partial\lambda}T(\lambda;x)&=q(T(\lambda;x))\,,\quad T(0;x)=x\\
	\ln g(\lambda;x)&=\int_0^{\lambda} v(T(\kappa;x))d\kappa\,.\label{eq:gEqApp}
\end{align}
\end{subequations}
This implies in particular that upon specializing $F(0;x)$ to a probability generating function $P(0;x)=\sum_{n\geq 0}p_n(0)x^n$ and the formal parameter $\lambda$ to a non-negative real parameter $t$, the resulting time-dependent PGF $P(t;x)$ may be written in the form
\begin{equation}
	\label{eq:allNBoneApp}
	\begin{aligned}
		P(t;x)&=g_{BDA}(t;x)g_{CDA}(t;x)P(0;T_{DA}(t;x))\\
		\tfrac{\partial}{\partial t}T_{DA}(t;x)&=q(T_{DA}(t;x))
		=\tau(1-T_{DA}(t;x))+\alpha\left(\big(T_{DA}(t;x)\big)^2-T_{DA}(t;x)\right)\,,\quad T_{DA}(0;x)=x\\
		\ln g_{BDA}(t;x)&=\beta\int_0^t (T_{DA}(\kappa;x)-1)d\kappa\\
		\ln g_{CDA}(t;x)&=\gamma\int_0^t \left(\big(T_{DA}(\kappa;x)\big)^2-1\right)d\kappa\,.
	\end{aligned}
\end{equation}
This entails that the concrete computations are strongly dependent on the particular values of the parameters $\tau$ and $\alpha$ (which determine the form of $T_{DA}(t;x)$), and of $\beta$ and $\gamma$ (which determine together with the result for $T_{DA}(t;x)$ the form of $g_{BDA}(t;x)$ and $g_{CDA}(t;x)$, respectively).\\

The simplest scenario is given for $\alpha=\tau=0$, in which case according to~\eqref{eq:allNBoneApp}
\begin{equation}
T^{\alpha=\tau=0}_{DA}(t;x)=x\quad\Rightarrow\quad
g^{\alpha=\tau=0}_{BDA}(t;x)=e^{\beta t(x-1)}\,,\quad g^{\alpha=\tau=0}_{CDA}(t;x)=e^{\gamma t(x^2-1)}\,.
\end{equation}
In all other cases, one first needs to compute the formula for $T_{DA}(t;x)$, upon which the formulae for $g_{BDA}(t;x)$ and $g_{CDA}(t;x)$ have to be computed via integration according to~\eqref{eq:allNBoneApp}. Unlike in the simplest scenario presented before, moreover additional work is necessary in order to identify the resulting PGF $P(t;x)$ as convolution and/or compound distribution of discrete distributions.\\

We exemplify this procedure for the representative special case $\alpha=\tau>0$. In the first step, $T(t;x)\equiv T^{\alpha=\tau>0}_{DA}(t;x)$ is computed as follows:
\begin{equation}
	\begin{aligned}\label{eq:TDAintApp}
		\tfrac{\partial}{\partial t}T(t;x)
	&=\alpha\left(1-T(t;x)\right)+\alpha\left(\big(T(t;x)\big)^2-T(t;x)\right)
	=\alpha\left(1-T(t;x)\right)^2\,,\quad T(0;x)=x\\
	\Rightarrow\quad T(t;x)&=\frac{1-x}{\alpha  t (x-1)-1}+1\,.
	\end{aligned}
\end{equation}
Comparing this result to the formulae for $T_D(t;x)$ and $T_A(t;x)$ as computed for \emph{individual} elementary non-binary reactions (cf.\ Table~\ref{tab:nonBinaryCRs}),
\begin{equation}
T_D(t;x)=Bern(e^{-t\tau},x)=(1-e^{-t\tau})+x e^{-t\tau}\,,\quad T_A(t;x)=Geom(e^{-t\alpha};x)=\frac{x e^{-t\alpha}}{1-x(1-e^{-t\alpha})}\,,
\end{equation}
one may recognize the result for $T(t;x)\equiv T_{DA}^{\alpha=\tau>0}(t;x)$ as presented in~\eqref{eq:TDAintApp} as a compound distribution of the form
\begin{equation}
\begin{aligned}
T_{DA}^{\alpha=\tau>0}(t;x)&=T_D(s(t);{\color{blue}T_A(s(t);x)})
=(1-e^{-s(t)\alpha})+{\color{blue}\left(
\frac{x e^{-s(t)\alpha}}{1-x(1-e^{-s(t)\alpha})}
\right) }e^{-s(t)\alpha}\\
&=\left(1-\tfrac{1}{(1+\alpha t)}\right)+{\color{blue}\left(
\frac{x \tfrac{1}{(1+\alpha t)}}{1-x(1-\tfrac{1}{(1+\alpha t)})}
\right) }\tfrac{1}{(1+\alpha t)}\\
&=1+\left(
\frac{(x-1)(1+ \alpha t)}{1- \alpha t (x-1)}
\right)\tfrac{1}{(1+\alpha t)}=\frac{1-x}{\alpha  t (x-1)-1}+1\,.\\
s(t)&=\frac{1}{\alpha}\ln (1+\alpha t)
\end{aligned}
\end{equation}
The remaining formulae as presented in Theorem~\ref{thm:NBCRN} may be derived in a similar fashion.

\section{Details of the derivation of the Sobolev-Jacobi polynomial formulae}

\subsection{On the Gurappa-Panigrahi type formulae for the Sobolev-Jacobi polynomials}
\label{app:GPSJ}

To complement the material presented in Section~\ref{sec:GP}, we review here some of the details of the approach of~\cite{gurappa1996new,gurappa1999equivalence,gurappa2001novel,gurappa2002linear,panigrahi2003coherent,gurappa2004polynomial}. Consider the eigenfunction equation for the Jacobi-type second-order differential operator,
\begin{equation}
	\begin{aligned}
	&(D^{\alpha,\beta}_{Jac}+n(n+\alpha+\beta+1))f(x)=0\\
	&\quad D^{\alpha,\beta}_{Jac}:=(1-x^2)\tfrac{d^2}{d x^2}
	+(\beta-\alpha-(\alpha+\beta+2)x)\tfrac{d}{d x}\,.
	\end{aligned}
\end{equation} 
Employing the auxiliary relation
\begin{equation}
x^2\tfrac{d^2}{d x^2}=\left(x \tfrac{d}{d x}\right)^2-\left(x \tfrac{d}{d x}\right)=D^2-D\qquad (D=x \tfrac{d}{d x})\,,
\end{equation}
we may express the eigenfunction equation for $D^{\alpha,\beta}_{Jac}$ in the alternative form
\begin{equation}
\begin{aligned}
&(F_n(D)+P(x,\tfrac{d}{d x}))f_n(x)=0\\
&\quad F_n(D)=-D^2-(\alpha+\beta+1)D+n(n+\alpha+\beta+1)=-(D-n)(D+n+\alpha+\beta+1)\\
&\quad P(x,\tfrac{d}{dx})=\tfrac{d^2}{d x^2}+(\beta-\alpha)\tfrac{d}{dx}\,.
\end{aligned}
\end{equation}
Specializing to $\alpha=-1$ and employing the Ansatz~\eqref{eq:BPans}
\begin{equation}
\begin{aligned}
	F_{\lambda}(D)x^{\lambda}&=0\;
	\Rightarrow\; y_{\lambda}(x)=c_{\lambda}\hat{G}_{\lambda}x^{\lambda}\,,\quad 
	\hat{G}_{\lambda}:=\sum_{m\geq 0}(-1)^m\left[\frac{1}{F(D)}P(x,\tfrac{d}{dx})\right]^m
\end{aligned}
\end{equation}
of Gurappa and Panigrahi, we may recover the formula for the Sobolev-Jacobi polynomials as presented in~\eqref{eq:GPSJ},
\begin{equation}
		\begin{aligned}
			\tilde{P}^{(-1,\beta)}_n(x)&=
			\hat{G}_{\beta,n}\:x^n\,,\quad 
			\hat{G}_{\beta,n}=\sum_{m\geq0}
			\bigg[\frac{1}{(D-n)(D+n+\beta)}\left(\tfrac{d^2}{dx^2}+(\beta+1)\tfrac{d}{dx}\right)\bigg]^m\,.
		\end{aligned}
\end{equation}

Finally, in order to prove the alternative formula~\eqref{eq:GPSJalt} for the case $\alpha=-1$, we have to employ the relation
\begin{equation}
\begin{aligned}
\frac{1}{(D-n)(D+n-1)}\tfrac{d^2}{d x^2} x^n&=
\frac{1}{(D+n-1)}\tfrac{d^2}{d x^2}\frac{1}{(D-n-2)} x^n
=-\frac{1}{2(D+n-1)}\tfrac{d^2}{d x^2}x^n\\
\Rightarrow\quad 
\left(\frac{1}{(D-n)(D+n-1)}\tfrac{d^2}{d x^2}\right)^m\; x^n&=
\frac{1}{m!}\left(-\frac{1}{2(D+n-1)}\tfrac{d^2}{d x^2}\right)^m\:x^n\,,
\end{aligned}
\end{equation}
where the latter statement may be derived from the former by induction on $m$. This concludes the proof of the alternative formula
\begin{equation}
\tilde{P}^{(-1,-1)}_n(x)
		=e^{-\tfrac{1}{2}\frac{1}{D+n-1}\tfrac{d^2}{dx^2}}x^n\,.
\end{equation}

\subsection{Decomposition coefficients}
\label{app:SJ}

The linear combination coefficient formula~\eqref{eq:AMn},
\begin{equation}
x^M	=\sum_{n=0}^M A^{M,n}_{\beta}\tilde{P}_n^{(-1,\beta)}\left(\frac{x-\frac{r_{\ell}}{2}}{1-\frac{r_{\ell}}{2}}\right)\,,
\end{equation}
may be transformed into a more tractable form as follows: via the change of variables
\begin{equation}
	y(x)=\frac{x-\frac{r_{\ell}}{2}}{1-\frac{r_{\ell}}{2}}\,,\quad \Rightarrow\quad x(y)=\tfrac{r_{\ell}}{2}+y\left(1-\tfrac{r_{\ell}}{2}\right)\,,
\end{equation}
which entails in particular that
\begin{equation}
x^M=\sum_{P=0}^M\binom{M}{P}\left(\tfrac{r_{\ell}}{2}\right)^{M-P}\left(1-\tfrac{r_{\ell}}{2}\right)^P y^P\,.
\end{equation}
Using the decomposition formula
\begin{equation}
y^P=\sum_{n=0}^P B^{P,n}_{\beta} \tilde{P}^{(-1,\beta)}_n(y)\,,
\end{equation}
one may obtain an expression of the coefficients $A^{M,n}_{\beta}$ in terms of the coefficients $B^{M,n}_{\beta}$ as follows:
\begin{equation}
\begin{aligned}
x^M&=\sum_{n=0}^M A^{M,n}_{\beta}\tilde{P}_n^{(-1,\beta)}\left(y\right)=\sum_{P=0}^M\binom{M}{P}\left(\tfrac{r_{\ell}}{2}\right)^{M-P}\left(1-\tfrac{r_{\ell}}{2}\right)^P \sum_{n=0}^P B^{P,n}_{\beta} \tilde{P}^{(-1,\beta)}_n(y)\\
&=\sum_{n=0}^M\left(\sum_{P=n}^M\binom{M}{P}\left(\tfrac{r_{\ell}}{2}\right)^{M-P}\left(1-\tfrac{r_{\ell}}{2}\right)^P  B^{P,n}_{\beta}\right) \tilde{P}^{(-1,\beta)}_n(y)\\
\Rightarrow \quad A^{M,n}_{\beta}&=\sum_{P=n}^M\binom{M}{P}\left(\tfrac{r_{\ell}}{2}\right)^{M-P}\left(1-\tfrac{r_{\ell}}{2}\right)^P  B^{P,n}_{\beta}\,.	
\end{aligned}
\end{equation}

In order to compute the decomposition coefficients $B^{M,n}_{\beta}$ for the two cases of interest, the following auxiliary formula proves quintessential:
\begin{lem}
Let $A\in \bZ_{\geq0}$ and $B\in \bR_{\geq0}$ be two parameters of the integral expression
\begin{equation}
	\cI_{A,B}:= \int_{-1}^{+1}dx\: (x-1)^A(x+1)^B\,.
\end{equation}
Then the explicit formula for $\cI_{A,B}$ reads
\begin{equation}
\label{eq:lemAuxInt}
\cI_{A,B}=\frac{(-1)^A A! \Gamma(B+1) 2^{A+B+1}}{\Gamma(A+B+2)}\,.
\end{equation}
\begin{proof}
Via partial integration and induction -- since by assumption $A\in \bZ_{\geq 0}$ is an integer, we may use partial integration in order to reduce the exponent of the term $(x-1)^A$ to zero, with a typical step of the reduction reading
\begin{align*}
\cI_{P,Q}&=\int_{-1}^{+1}dx\; (x-1)^P(x+1)^Q=\frac{1}{(Q+1)}\int_{-1}^{+1}dx\; (x-1)^P\left(\frac{d}{dx}(x+1)^{Q+1}\right)\\
&=\left[\frac{(x-1)^P(x+1)^{Q+1}}{(Q+1)}\right]_{-1}^{+1}
-\Theta(P-1)\frac{P}{(Q+1)}\int_{-1}^{+1} (x-1)^{P-1}(x+1)^{Q+1}\\
\Leftrightarrow\quad \cI_{P,Q}&=\delta_{P,0}\tfrac{2^{Q+1}\Gamma(Q+1)}{\Gamma(Q+2)}-\Theta(P-1)\tfrac{P}{(Q+1)}\cI_{P-1,Q+1}\,.
\end{align*}
\end{proof}
\end{lem}

Considering now the definitions of the Sobolev-type inner products $\Phi^{(-1,-1)}_{1,1}$ and $\Phi^{(-1,\beta)}_1$, respectively, we first need to compute the derivatives of the Sobolev-Jacobi polynomials $\tilde{P}^{(-1,\beta)}_n(x)$:
\begin{equation}
\tfrac{d}{dx}\tilde{P}^{(-1,\beta)}_n(x):=\delta_{n,1}+\Theta(n-2)\binom{2n-2}{n}^{-1}\sum_{k=1}^{n-1}\binom{n-1}{k}\binom{n-1}{n-k}\left[
(n-k)(x-1)^{n-1-k}(x+1)^k
+k (x-1)^{n-k}(x+1)^{k-1}
\right]
\end{equation}
With these preparations, and noting that
\begin{equation}
\tilde{P}^{(-1,-1)}_{n\geq 2}(-1)=\tilde{P}^{(-1,-1)}_{n\geq 2}(1)=0\,,
\end{equation}
the formulae for the coefficients may be derived. Using \textsc{Mathematica} to simplify the resulting summation formulae, the coefficients $B^{M,n}_{-1}$ for the Sobolev-Jacobi polynomials $\tilde{P}^{(-1,-1)}_n(x)$ and $M>0,n>0$ (with $M+n\in 2\bZ$) are computed as 
\begin{equation}
	\begin{aligned}
	B^{M,n}_{-1}&=\frac{M! \Gamma(n+\tfrac{1}{2})}{
			2^{M-n}
			n!\Gamma(\tfrac{M+n+1}{2})
			\Gamma(1+\tfrac{M-n}{2})}\,.	
	\end{aligned}
\end{equation}
The final form as presented in~\eqref{eq:SJPcoeffMoneMone} may be obtained via using the identities
\begin{equation}
	\begin{aligned}
		\Gamma(p+\tfrac{1}{2})&=\frac{(2p-1)!!}{2^p }\sqrt{\pi}\,,\quad (2p-1)!!=\frac{(2p)!}{2^p p!}\,.
	\end{aligned}
\end{equation}
The computations for the case $\beta>-1$ are entirely analogous, and are thus omitted here for brevity.

\end{document}